\documentclass[10pt,pdftex,a4paper]{article}

\usepackage{meta}
\usepackage{mathtools}
\usepackage{amssymb}
\usepackage{amsmath}
\usepackage{amsfonts}
\usepackage{amsthm}
\usepackage{stmaryrd}
\usepackage{mathabx}
\usepackage{mathrsfs}
\usepackage[nomath]{lmodern} %
\usepackage{afterpage}
\usepackage{picins}
\usepackage{moresize}
\usepackage{marginnote}
\usepackage[LGR, T1]{fontenc}
\usepackage[utf8]{inputenc}
\usepackage{textcomp}
\usepackage{caption}
\usepackage{subcaption}   %
\usepackage{rotating}
\usepackage{ppprog}       %
\usepackage{shorthand}    %
\usepackage{veralgs}      %
\usepackage[normalem]{ulem}
\usepackage{paralist}
\usepackage{notation}
\usepackage{xcolor}
\usepackage{mathtools}
\usepackage{framed}
\usepackage{accents}
\usepackage{lstlinebgrd}
\usepackage{tabularx}
\usepackage{multirow}
\usepackage{booktabs}
\usepackage{pdflscape}
\usepackage{rotating}
\usepackage{adjustbox}

\colorlet{shadecolor}{yellow!20}
\usepackage{tikz}
\usetikzlibrary{arrows,backgrounds,calc,decorations.pathmorphing,fit,shapes.geometric,matrix,shapes,patterns,positioning}
\usepackage{transactiondiagram}

\usepackage{enumerate} 

\usepackage{codeconfig-bw}

\begin{document}

\title{\mytitle}

\author{\myauthors
\\\myemailgroup
\\\myinstitute, \myuniversity
}

\maketitle

\begin{abstract}
\noindent Transactional Memory (TM) is an approach aiming to simplify concurrent
programming by automating synchronization while maintaining efficiency. TM
usually employs the optimistic concurrency control approach, which relies on
transactions aborting and restarting if conflicts occur. However, an aborted
transaction can still leave some effects in the system that cannot be cleaned up,
if irrevocable operations are present within its code. The pessimistic approach
eliminates that problem, since it relies on deferring operations in case of
conflict rather than aborting, but hitherto pessimistic TMs suffered from low
parallelism due to the need of serializing transactions. In this paper, we aim
to introduce OptSVA, a pessimistic TM concurrency control algorithm that
ensures a high level of parallelism through a battery of far-reaching
optimizations including early release, asynchronous execution, and the
extensive use of buffering.

\end{abstract}

\providecommand{\keywords}[1]{\noindent\textbf{Index terms---}#1}

\section{Introduction}
\label{sec:introduction}

In a world dominated by multicore processors and distributed applications even
the rank-and-file programmer is increasingly likely to have to turn to parallel
programming to take full advantage of 
various multiprocessor architectures.
However, concurrent execution can cause operations on separate processors to
interleave in ways that produce anomalous results, forcing the programmer to
predict 
and eliminate them through
synchronization.

Yet 
implementing synchronization correctly
is notoriously difficult, since the
programmer must reason about interactions among seemingly unrelated parts of
the system code. 
Furthermore, low-level mechanisms like barriers, monitors, and locks are
easily misused and performance, consistency, or progress fall prey to faulty
design or simple bugs. Worse still, the resulting errors like deadlocks or
livelocks are far reaching, difficult to track down, and often
non-deterministic.

Consequently, researchers seek ways to automate synchronization 
while retaining a decent level of efficiency. 
\emph{Transactional memory (TM)} \cite{HME93,ST95} is one such approach,
introduced for multiprocessor architectures, and then extended to distributed
systems as \emph{distributed TM} (see \cite{KAJ+08,TRP13}, among others). The TM
approach requires that the programmer
annotates blocks of code as \emph{transactions} that must be executed with
specific correctness guarantees (e.g., \emph{serializability} \cite{Pap79},
\emph{opacity} \cite{GK10}, and TMS1 \cite{DGLM13}).
The TM system then ensures these guarantees using an underlying concurrency
control algorithm, which provides synchronization as needed to give the
illusion of transaction atomicity and isolation, but whose details remain hidden from the
programmer.
In effect, TM reduces the effort required to implement correct and efficient
programs.

Most TM research emphasizes optimistic concurrency control.%
There are variations to this approach, but, generally, a transaction executes
regardless of other transactions using buffers, and only updates the state
of the system as it finishes executing (i.e. on \emph{commit}). If two
transactions try to access the same object, and one of them writes to it, they
\emph{conflict} and one of them \emph{aborts} by discarding its buffers and
\emph{restarts} to execute all of its operations anew.
Such an approach allows parallelism on multicores, as transactions do not block
one another during execution.
However, this requires an assumption that an aborted transaction does not have
any visible effect on the system.

Hence, the optimistic approach comes with its own set of problems.
Most notably, TM transactions can contain any code, including code with side
effects, such as: system calls, I/O operations, locking, or network communication.
These are referred to collectively as \emph{irrevocable operations}, since it
is practically impossible to revoke their effects. However, the \emph{modus
operandi} of optimistic transactions depends on aborted transactions cleaning
up after themselves.
The problem can be mitigated by using irrevocable transactions that run
sequentially, and so cannot abort \cite{WSA08}, or providing multiple versions
of objects on which transactions execute reads \cite{AH11,PFK10}. In other
cases, irrevocable operations are simply forbidden in transactions (e.g., in
Haskell \cite{HMJH04}) or moved to commit.  Other research suggests that a form
of compensation can be used to fix the computations, so that conflicting
transactions do not abort \cite{BMT10}. 
These solutions, however, introduce complexity and overhead, relax the
consistency guarantees of TM, or limit the applicability of TM.

A different approach, as suggested in \cite{MS12,AMS12,ADFK14} and our earlier
work \cite{Woj05b,Woj07}, is to use fully-pessimistic concurrency control
\cite{BHG87,WG02}. This involves transactions waiting until they have
permission to access shared objects. In effect, potentially conflicting
operations are postponed until they no longer conflict. Thus transactions, for
the most part, avoid forced aborts, and therefore, they also naturally avoid the
problems stemming from irrevocable operations.
However, the authors of \cite{MS12} show that the fully-pessimistic approach
can have negative impact on performance in high contention, since it depends on
serializing write transactions to prevent aborts, which inherently limits
parallelism.
The goal of this paper is to show that this penalty on parallelism is not
inherent in the pessimistic approach and can be overcome.

In our previous work \cite{Woj05b,Woj07,SW13}, we attempted to mitigate the
performance issue by introducing a pessimistic TM algorithm using \emph{early
release}---a technique that allows conflicting transactions to execute in
parallel and still commit.  Specifically, the \emph{Supremum Versioning
Algorithm} (\emph{SVA}) 
uses \emph{a priori} knowledge to allow transactions to safely 
release a variable it will no longer access, which, in turn, allows  other
transactions to read or update it without waiting for the first transaction to
finish completely.
This increases the number of allowed interleavings between transactions, which
then translates into promising performance results.

In this paper we present OptSVA, a TM concurrency control algorithm that builds
on the early release mechanism introduced in SVA, but introduces a number of
deep modifications that eliminate its predecessors limitations, which
effectively make OptSVA a novel and unique algorithm.
Most notably, OptSVA parallelizes reads where possible, and relies on buffering
to institute automatic \emph{privatization} of variables, where SVA is
operation-type--agnostic and writes to variables in-place. This allows to both expedite early
release and defer the moment of synchronization, resulting in greater
parallelism.
Furthermore, to the best of the authors' knowledge, OptSVA is the first TM
concurrency control algorithm to delegate specific concurrency-control--related
tasks to separate threads to achieve local asynchrony, allowing a transaction
to perform local computations and non-conflicting operations while waiting to
serialize conflicting operations with other transactions.
This feature is especially valuable in distributed systems, where network
communication introduces delays.

Furthermore, we show through formal analysis that OptSVA can produce tighter
interleavings than SVA due to the increased level of parallelism.
This is also born out by experimental evaluation, which verifies that the
higher level of complexity of the algorithm does not incur significant enough
penalties to nullify the advantage of better interleavings. 

In addition, we demonstrate that OptSVA meets the same correctness guarantees as
SVA, by presenting a proof for last-use opacity.
Last-use opacity \cite{SW14-disc,SW15-arxiv} is a strong safety property for TM systems, that relaxes
opacity to allow early release after the last use of a variable in a
transaction.

Given that opacity is defined as a property that must be demonstrated for all
prefixes of a given transactional schedule, proving it for a complex system is
typically troublesome, as demonstrated by research on markability \cite{LP14} and
graph representation of opacity \cite{GK10}, both techniques trying to work
around the basic definition of opacity. By extension, the same is true of
last-use opacity.
Furthermore, since last-use opacity is defined using histories, but buffering
algorithms like OptSVA divorce transactional operations from the actual
operations on memory, and perform synchronization based on the latter,
demonstrating last-use opacity is even more complex.
Hence, apart from the proof itself, we contribute a \emph{trace harmony}, a
proof technique that shows last-use opacity based on interrelationships among
memory accesses (and can be easily extended to show related properties like
opacity).

The paper is structured as follows: Following the introduction,
\rsec{sec:related-work} shows other research relating to OptSVA.  Then, in
\rsec{sec:optsva} we present the OptSVA algorithm in full.  In
\rsec{sec:comparison} we compare the parallelism of OptSVA histories to those
admitted by SVA. 
In \rsec{sec:trace-harmony}, we introduce trace harmony, a proof technique that allows us to demonstrate last-use opacity based on memory accesses.
Then, in \rsec{sec:proof} we employ that proof technique to show that despite
allowing greater parallelism OptSVA satisfies last-use opacity, i.e. the same
safety property as SVA.
Finally, we conclude in \rsec{sec:conclusions}.

\section{Related Work}
\label{sec:related-work}
\label{sec:rw}

A large number of TM systems were proposed to date. Here, we concentrate only
on those that use some of the same techniques as the one OptSVA is based on:
pessimistic concurrency control and early release. 

\subsection{Pessimistic TM systems}
Seeing as TM systems tend heavily towards optimistic concurrency control,
pessimistic systems are relatively rare. 
Examples of these include our previous work on the Basic Versioning Algorithm
and the Supremum Versioning Algorithm \cite{Woj05b,Woj07}. The former is an
opaque in-place TM that never aborts transactions and uses \emph{a priori} knowledge on
access sets to enforce disjoint-access parallelism. The latter adds an early
release mechanism in an effort to allow conflicting transactions to execute
partially in parallel. SVA was later extended to also allow optional aborts in
\cite{SW13}.  All of these algorithms were proposed for a system model using complex objects
defining custom methods rather than variables, so none of them distinguish between reads and writes
as in the traditional TM model.
The algorithm proposed here builds on both of these and introduces the
distinction between operation types, as well as other modifications, all aiming
to increase the degree of parallelism the of which TM system is capable.

Another example is the system proposed in \cite{MS12}, where read-only
transactions execute in parallel, but transactions that update are synchronized
using a global lock to execute one-at-a-time. This idea was improved upon in
Pessimistic Lock Elision (PLE) \cite{AMS12}, where a number of optimizations
were introduced, including encounter-time synchronization, rather than
commit-time. However, the authors show that sequential execution of update
transactions yields a performance penalty. In contrast, the algorithm proposed
in this paper maintains a high level of parallelism regardless of updates. In
particular, the entire transaction need not be read-only for a variable that is
read-only to be read-optimized.

SemanticTM is another pessimistic TM system \cite{DFK14}. Rather than using
versioning or blocking, transactions are scheduled and place their operations
in bulk into a producer-consumer queues attached to variables. The instructions
are then executed by a pool of non-blocking executor threads that use
statically derived access sets and dependencies between operations to ensure
the right order of execution. The scheduler ensures that all operations of one
transaction are executed before another's. In addition, statically derived
access sets and dependencies between operations are used to ensure that
operations are executed in the right order. Contrary to SVA and OptSVA, the
transactions cannot abort, forcibly, but also do not allow for manual aborts.
SemanticTM and versioning algorithms produce similar histories, but while the
latter are deadlock-free, SemanticTM is wait-free. However, even without
aborts, in contrast to SVA and OptSVA, SemanticTM does not guarantee that a
given operation is executed (at most) once.

While not exactly a pessimistic system \emph{per se},
Twilight STM \cite{BMT10} relaxes isolation to allow conflicting transactions
to reconcile using so-called twilight code at the end of the transaction and
commit nevertheless.
If a transaction reads a value that was modified by another transaction since
its start, twilight code can re-read the changed variables and re-write the 
variables the transaction modified to reflect the new state, allowing the
transaction to commit anyway.
Even though the operations are re-executed, as per optimistic concurrency
control, it means that transactions that execute twilight code always finish
successfully nevertheless.
This means, however, that regardless of transactions aborting or committing,
the code within them is prone to re-execution, which introduces problems with
irrevocable operations that SVA and OptSVA try to avoid.

\subsection{Early release TM systems}
A number of TM systems employ early release to improve parallelism. One example
is SVA, which we elaborate on earlier.

Another example is
Dynamic STM \cite{HLMS03}, the system 
that can be
credited with introducing the concept of early release in the TM context.
Dynamic STM allows transactions that only perform read operations on particular
variables to (manually) release them for use by other transactions. However, it
left the assurance of safety to the programmer, and, as the authors state, even
 linearizability cannot be guaranteed by the system. In contrast, versioning
 algorithms guarantee, at minimum, last-use opacity.

The authors of \cite{SK06} expanded on the work above and evaluated the concept
of early release with respect to read-only variables on several concurrent data
structures. The results showed that this form of early release does not
provide a significant advantage in most cases, although there are scenarios
where it would be advantageous if it were automated. We use a different
approach in SVA and OptSVA, where early release is not limited to read-only
variables.

DATM  \cite{RRHW09} is another noteworthy system with an early release mechanism.
DATM is an optimistic multicore-oriented TM based on TL2 \cite{DSS06}, augmented
with early-release support.
It allows a transaction $\tr_i$ to write to a variable that was accessed by
some uncommitted transaction $\tr_j$, as long as $\tr_j$ commits before
$\tr_i$.
DATM also allows transaction $\tr_i$ to read a speculative value, one written
by $\tr_j$ and accessed by $\tr_i$ before $\tr_j$ commits. 
DATM detects if $\tr_j$ overwrites the data
or aborts, in which case $\tr_i$ is forced to restart. 
DATM allows all schedules allowed by conflict-serializability. This means that
DATM allows overwriting, as well as cascading aborts. It also means that it
does not satisfy \lopacity{}. Hence, DATM is weaker than OptSVA (as well as
most TM systems). DATM can also incur very high transaction abort rates, in
comparison to OptSVA, whose abort rate will tend towards zero (depending on the
use of programamtic aborts).

\section{OptSVA}
\label{sec:optsva}

This section describes 
the
\emph{Optimized Supremum Versioning Algorithm
(OptSVA).} 
OptSVA is specified in full in \rfig{fig:optsva}.
Given the complexity of the algorithm we split the presentation into four
parts. In the first, we explain the rudiments of the use of versioning for
concurrency control, as well as the early release mechanism. These are the
foundations of the algorithm and the elements which reflect  the basic design of SVA.
Then, in the other three parts we discuss how the combination of explicit
read/write distinction, buffering, and asynchronous execution of specific
synchronization-related tasks is used to optimize accesses to read-only
variables, to delay synchronization of the initial operation upon an initial
write, and expedite early release to the last (closing) write.
First, though, we present the system model.

\newbox\codefill
\begin{lrbox}{\codefill}
    \begin{minipage}{\linewidth}
    \vfill
    \end{minipage}
\end{lrbox}

\def\magic#1#2{
  \newbox{#1}
  \begin{lrbox}{#1}
    \input{#2}
  \end{lrbox}  
}

\def\moremagic#1{
\begin{minipage}[t]{\linewidth}%
  \usebox#1 
  \end{minipage}
}

\def\ersatzsubfloat#1{\begin{subfigure}[b]{\linewidth} #1 \end{subfigure}}

\def\hlcolor{javagreen!10}
\def\betweenlistings{-10pt}

  \newbox{\codestart}
  \begin{lrbox}{\codestart}
    \lstset{
    }
    \begin{lstlisting}[name=optsva]
proc start(Transaction $\trobj_i$) {
  for($\obj$ ∈ $\accesses{i}$ sorted by $\prec_\lockf$)                                   $\label{l:lock-a}$
    lock $\lock{\obj}$                                                  $\label{l:lock-b}$
  for($\obj$ ∈ $\accesses{i}$) {                                      $\label{l:start-loop-1}$
    $\gv{\obj}$ ← $\gv{\obj}$ + 1                                       $\label{l:set-gv}$ 
    $\pv{\obj}{i}$ ← $\gv{\obj}$                                        $\label{l:set-pv}$
  }                                                                     $\label{l:end-loop-2}$
  for($\obj$ ∈ $\accesses{i}$ sorted by $\prec_\lockf$)                                   $\label{l:unlock-a}$ 
    unlock $\lock{\obj}$                                                $\label{l:unlock-b}$
  for($\obj$ ∈ $\accesses{i}$: $\obj$ is read-only)                   $\label{l:reads-a}$
    async run read_buffer($\trobj_i$,$\obj$)  $\label{l:reads-c}$
      when $\pv{\obj}{i}$ - 1 = $\lv{\obj}$ $\label{l:reads-b}$
}                                                                       $\label{l:end1}$
    \end{lstlisting}

  \end{lrbox}

  \newbox{\coderead}
  \begin{lrbox}{\coderead}
        \makeatletter
    \lstset{
        firstnumber=auto, 
    }
    \csname\@lst @SetFirstNumber\endcsname
    \makeatother
    \begin{lstlisting}[name=optsva]
proc read(Transaction $\trobj_i$, Var $\obj$) {
  if ($\obj$ is read-only) {                                    $\label{l:read-ro-start}$
    join with read_buffer($\trobj_i$,$\obj$)                    $\label{l:read-ro-wait}$
  } else if ($\wc{i}{\obj}$ = 0 and $\rc{i}{\obj}$ = 0) {        $\label{l:read-c2}$ $\label{l:read-c1}$
    wait until $\pv{\obj}{i}$ - 1 = $\lv{\obj}$               $\label{l:read-access}$
    checkpoint($\trobj_i$,$\obj$)                             $\label{l:read:checkpoint}$
    $\buf{i}{\obj}$ ← $\stored{i}{\obj}$                        $\label{l:first-read:set-buf}$ 
  }
  if (∃$\objy$: $\rv{i}{\objy}$ ≠ $\cv{\objy}$) $\label{l:read-abort}$
    abort($\trobj_i$,$\obj$) and exit $\label{l:read-abort-go}$
  $\rc{i}{\obj}$ ← $\rc{i}{\obj}$ + 1                            $\label{l:read-inc}$
  return $\buf{i}{\obj}$                                      $\label{l:read-only:return}\label{l:non-local-read:return}\label{l:local-read:return}$
}
proc read_buffer(Transaction $\trobj_i$, Var $\obj$) {
  $\rv{i}{\obj}$ ← $\cv{\obj}$                                          $\label{l:rb-checkpoint}$
  $\buf{i}{\obj}$ ← $\obj$                                      $\label{l:rb-copy}\label{l:read-buffer:set-buf}\label{l:read-buffer:get}$
  release($\trobj_i$,$\obj$)                                            $\label{l:rb-release}$
  async run read_commit($\trobj_i$,$\obj$) 
    when $\pv{\obj}{i}$ - 1 = $\ltv{\obj}$ $\label{l:rb-async}$
}
proc read_commit(Transaction $\trobj_i$, Var $\obj$) {
  $\cv{\obj}$ ← $\pv{\obj}{i}$                       $\label{l:rc-start}$
  if (∃$\objy$: $\rv{i}{\objy}$ > $\cv{\objy}$)
    abort($\trobj_i$) and exit
  $\ltv{\obj}$ ← $\pv{\obj}{i}$                      $\label{l:rc-end}$
}
    \end{lstlisting}

  \end{lrbox}

  \newbox{\codewrite}
  \begin{lrbox}{\codewrite}
        \makeatletter
    \lstset{linebackgroundcolor=%
        firstnumber=auto, 
    }
    \csname\@lst @SetFirstNumber\endcsname
    \makeatother
    \begin{lstlisting}[name=optsva]
proc write(Transaction $\trobj_i$, Var $\obj$, Value $\val$) {
  if (not $\val$ in domain of $\obj$)       $\label{l:write:domain}$
    abort($\trobj_i$) and exit $\label{l:write:domain:abort}$
  if (∃$\objy$: $\rv{i}{\objy}$ ≠ $\cv{\objy}$)
    abort($\trobj_i$) and exit 
  $\buf{i}{\obj}$ ← $\val$                  $\label{l:write:set-buf}$
  $\wc{i}{\obj}$ ← $\wc{i}{\obj}$ + 1       $\label{l:write-wc-inc}$
  if ($\wc{i}{\obj}$ = $\wub{i}{\obj}$)        $\label{l:write-no-reads-ub}\label{l:write-ub}$
    async run write_buffer($\trobj_i$,$\obj$) 
      when $\pv{\obj}{i}$ - 1 = $\lv{\obj}$ $\label{l:write-no-reads-async}\label{l:write-async}$
}
proc write_buffer(Transaction $\trobj_i$, Var $\obj$) {
  if (not checkpoint made)                                  $\label{l:wb-checkpoint-cond}$
    checkpoint($\trobj_i$,$\obj$)                            $\label{l:wb-checkpoint}$
  if (∃$\objy$: $\rv{i}{\objy}$ ≠ $\cv{\objy}$) $\label{l:wb-abort}$
    abort($\trobj_i$) and exit        
  $\obj$ ← $\buf{i}{\obj}$             $\label{l:wb-apply}\label{l:write-buffer:rset}$
  release($\trobj_i$,$\obj$)                               $\label{l:wb-release}$
}
    \end{lstlisting}

  \end{lrbox}

  \newbox{\codecommit}
  \begin{lrbox}{\codecommit}

\makeatletter
\lstset{linebackgroundcolor=%
    \rangecol{16-16}{\hlcolor}%
    \rangecol{35-35}{\hlcolor}%
    \rangecol{22-33}{\hlcolor}%
    \rangecol{37-38}{\hlcolor}%
    \rangecol{40-41}{\hlcolor},%
    firstnumber=auto 
}
\csname\@lst @SetFirstNumber\endcsname
\makeatother
    \begin{lstlisting}[name=optsva]
proc commit(Transaction $\trobj_i$) {
  for($\obj$ ∈ $\accesses{i}$) {                                $\label{l:commit-1}$
    if ($\obj$ is read-only)                                    $\label{l:commit-ro}$
      join with read_comit($\trobj_i$,$\obj$)                   $\label{l:commit-ro-wait}$
    else {
      if ($\wc{i}{\obj}$ = $\wub{i}{\obj}$) 
        join with write_buffer($\trobj_i$,$\obj$)
      else
        catch_up($\trobj_i$, $\obj$)                                  
      wait until $\pv{\obj}{i}$ - 1 = $\ltv{\obj}$              $\label{l:commit-commit}$
      if ($\pv{\obj}{i}$ - 1 = $\lv{\obj}$)
        $\lv{\obj}$ ← $\pv{\obj}{i}$                            $\label{l:commit-lv}$
      if ($\wc{i}{\obj}$ + $\rc{i}{\obj}$> 0 
           and $\rv{i}{\obj}$ = $\cv{\obj}$
           and $\pv{\obj}{i}$ - 1 > $\lv{\obj}$) 
        $\cv{\obj}$ ← $\pv{\obj}{i}$                            $\label{l:commit-cv}$
    } 
  }   
  if (∃$\objy$: $\rv{i}{\objy}$ > $\cv{\objy}$) $\label{l:commit-abort}$
    abort($\trobj_i$) and exit                                  $\label{l:commit-abort2}$ 
  for($\obj$ ∈ $\accesses{i}$)
    $\ltv{\obj}$ ← $\pv{\obj}{i}$                               $\label{l:commit-ltv}$
}
    \end{lstlisting}

  \end{lrbox}

  \newbox{\codeabort}
  \begin{lrbox}{\codeabort}
        \makeatletter
    \lstset{linebackgroundcolor=%
        \rangecol{44-52}{\hlcolor}%
        ,firstnumber=auto, 
    }
    \csname\@lst @SetFirstNumber\endcsname
    \makeatother
    \begin{lstlisting}[name=optsva]
proc abort(Transaction $\trobj_i$) {
  for($\obj$ ∈ $\accesses{i}$) {                                     $\label{l:abort-1}$
    wait until $\pv{\obj}{i}$ - 1 = $\ltv{\obj}$                 $\label{l:abort-commit}$
    if ($\wc{i}{\obj}$ > 0 
         and $\pv{\obj}{i}$ - 1 > $\lv{\obj}$ 
         and $\rv{i}{\obj}$ = $\cv{\obj}$) {                     $\label{l:abort-ro}\label{l:abort:recover-ro}$
      if ($\wc{i}{\obj}$ = $\wub{i}{\obj}$) 
        join with write_buffer($\trobj_i$,$\obj$) 
      $\obj$ ← $\stored{i}{\obj}$                               $\label{l:recover-copy}\label{l:abort:aset}$
      $\cv{\obj}$ ← $\rv{i}{\obj}$   $\label{l:recover-cv}\label{l:abort:recover-cv}$
    }
    if ($\pv{\obj}{i}$ - 1 = $\lv{\obj}$)                       $\label{l:abort-access}$
      $\lv{\obj}$ ← $\pv{\obj}{i}$                             $\label{l:abort-lv}$
    $\ltv{\obj}$ ← $\pv{\obj}{i}$                               $\label{l:abort-ltv}$
  }
}
    \end{lstlisting}

  \end{lrbox}

  \newbox{\codeaux}
  \begin{lrbox}{\codeaux}
        \makeatletter
    \lstset{
        firstnumber=auto, 
    }
    \csname\@lst @SetFirstNumber\endcsname
    \makeatother
    \begin{lstlisting}[name=optsva]
proc checkpoint(Transaction $\trobj_i$, Var $\obj$) {
  $\stored{i}{\obj}$ ← $\obj$        $\label{l:checkpoint-copy}\label{l:checkpoint:get}$
  $\rv{i}{\obj}$ ← $\cv{\obj}$       $\label{l:checkpoint-cv}$
}
proc release(Transaction $\trobj_i$, Var $\obj$) {
  $\cv{\obj}$ ← $\pv{\obj}{i}$   $\label{l:release-cv}$
  $\lv{\obj}$ ← $\pv{\obj}{i}$   $\label{l:release-lv}$
}
proc catch_up(Transaction $\trobj_i$, Var $\obj$) {
  wait until $\pv{\obj}{i}$ - 1 = $\lv{\obj}$             $\label{l:commit:access}$
  if (not checkpoint made) 
    checkpoint($\trobj_i$,$\obj$)                         $\label{l:commit-checkpoint}$
  if (∃$\objy$: $\rv{i}{\objy}$ ≠ $\cv{\objy}$)                         $\label{l:commit-abort-1}$
    abort($\trobj_i$) and exit 
  if ($\wc{i}{\obj}$ > 0)
    $\obj$ ← $\buf{i}{\obj}$ $\label{l:commit:rset}$ 
}
    \end{lstlisting}

  \end{lrbox}

  \newbox{\codeasync}
  \begin{lrbox}{\codeasync}
         \makeatletter
     \lstset{
         firstnumber=auto, 
     }
     \csname\@lst @SetFirstNumber\endcsname
     \makeatother

  \end{lrbox}

\afterpage{
    \begin{figure}
    \adjustbox{valign=t}{
        \begin{minipage}{.5\textwidth}%
        \ersatzsubfloat%
        {\moremagic\codestart} 
        \ersatzsubfloat%
        {\moremagic\coderead} 
        \ersatzsubfloat%
        {\moremagic\codewrite}
        \end{minipage}
    }
    \hfill
    \adjustbox{valign=t}{
        \begin{minipage}{.49\textwidth}%
        \ersatzsubfloat%
        {\moremagic\codecommit}
        \ersatzsubfloat%
        {\moremagic\codeabort}
        \ersatzsubfloat%
        {\moremagic\codeaux}
        \end{minipage}
    }
    \caption{\label{fig:optsva}Full listing of OptSVA.}
\end{figure}
}

\subsection{Transactional Memory System Model}

OptSVA operates in a system composed of a set of processes $\processes =
\{\proc_1, \proc_2, ..., \proc_n\}$ concurrently executing a set of finite
sequential programs $\prog = \{ \subprog_1, \subprog_2, ..., \subprog_n \}$,
where process $\proc_i$ executes $\subprog_i$.
Within these programs there is code that defines transactions.
A \emph{transaction} $\tr_i \in \transactions$ is some piece of code executed
by process $\proc_k$, as part of subprogram $\subprog_k$. Each process executes
transactions sequentially, one at a time.
Transactions contain local computations (that can be whatever) and invoke
operations on shared variables, or \emph{variables}, for short.
Each variable, denoted $\objx, \objy, \objz$ etc. supports the following
\emph{operations}, %
that allow to retrieve or modify its state:
\begin{enumerate}[a) ]
    \item \emph{write} operation (denoted $\twop{i}{\obj}{\val}$)
     that sets the
     state of $\obj$ to value $\val$;
     the operation's \emph{return value} is
     the constant $\ok_i$, indicating correct execution, or the constant $\ab_i$
     which indicates that transaction $\tr_i$ has been forcibly aborted due to
     some inconsistency,
     \item \emph{read} operation $\trop{i}{\obj}{}$ whose \emph{return value} 
     is the current state of $\obj$, or $\ab_i$ in case of a forced abort,
\end{enumerate}
In addition, transactions can execute the following operations related to themselves:
\begin{enumerate}[a) ]
    \setcounter{enumi}{2}
    \item \emph{start} (denoted $\init_i$) which initializes transaction $\tr_i$, and
        whose return value is the constant $\ok_i$,
    \item \emph{commit} (denoted $\tryC_i$) which attempts to commit $\tr_i$
        and returns either the constant $\co_i$, which signifies a
        successful commitment of the transaction or the constant $\ab_i$ in
        case of a forced abort,
    \item \emph{abort} (denoted $\tryA_i$) which aborts $\tr_i$ and returns $\ab_i$.
\end{enumerate}
The operations a--e defined above are part of the so-called transactional API. They can only be
invoked within a transaction.

Even though transactions are parts of subprograms evaluated by processes, it is 
convenient to talk about them as separate and independent entities. 
Thus, rather than saying $\proc_k$ executes some operation as part of 
transaction $\tr_i$, we will simply say that $\tr_i$ executes (or performs) 
some operation.

Transactions follow a particular life-cycle. At the outset each transaction
$\tr_i$ executes $\init_i$, following which it can execute local computations
and operations on variables. Eventually, $\tr_i$ executes either commit or
abort, which complete the execution of the transaction, following which the
transaction runs no further code. Furthermore, any operation other than start
can return $\ab_i$, which also means that the transaction was forced to
execute abort during that operation.  In addition, a transaction can execute
abort arbitrarily. In any case, the transaction can perform no other
computations following an execution of an abort.
If the transaction was forcibly aborted, though, it will be restarted by the
system, but it is easier to think of that as a separate consecutive transaction
(i.e. $\tr_j, i\neq j$) executed by the same process.

If transaction $\tr_i$ executed a read operation on variable $\objx$, we say
$\objx$ is in $\tr_i$'s \emph{read set}.
write operation on $\objx$, we say $\objx$ is in $\tr_i$'s \emph{write set}. If
$\tr_i$ executed either a read or a write operation on $\objx$, $\objx$ is in
$\tr_i$'s \emph{access set}, which we denote $\objx \in \accesses{i}$

In OptSVA specifically, executing $\init_i$ translates to executing procedure
\ppr{start}, $\tryC_i$ to executing \ppr{commit}, $\tryA_i$ to \ppr{abort},
$\twop{i}{\obj}{\val}$ to \ppr{write}, and $\trop{i}{\obj}{}$ to \ppr{read}.
If a transaction aborts as a result of some operation (returning $\ab_i$),
\ppr{abort} will also have been executed before the operation returns.

\subsection{Versioning Concurrency Control}
\label{sec:supremum-versioning}

OptSVA uses four version counters to determine whether a given transaction can
be allowed to access a particular shared variable, or whether the access should
be deferred to avoid conflicts.
The intuition behind how these counters work is by analogy to how the teller
may manage a queue in a bank: customers who come into the bank retrieve a
ticket with a number from a dispenser and wait before approaching the teller
until their number is called. Meanwhile the teller increments the number as she
finishes serving each consecutive customer.
In the analogy, each customer is a transaction, and the teller is some
resource, like a shared variable.  The number in the customer's hand is his
version for that variable, and it is being compared against the number that is
currently being served by the teller---the variable's version.
The design gets more involved as more variables and more counters are
introduced, and we explain it in detail below.

Whenever a transaction $\trobj_i$ starts, it retrieves a \emph{private version}
$\pv{\obj}{i}$ for every variable $\obj$ in its access set (lines
\ref{l:lock-a}--\ref{l:unlock-b}). The access set is assumed to be known
\emph{a priori}.
The values of private versions received by consecutive transactions are
generated from a \emph{global version} $\gv{\obj}$, which is initially $0$ and
is incremented with each starting transaction that has $\obj$ in its access
set.  Hence private versions are unique for a given variable and successive for
consecutive transactions. 
The assignment is also guarded by locks so that it is done atomically. In effect,
if one transaction $\tr_i$ has a greater private version for $\obj$ (or just
\emph{version} for $\obj$, for short) than another transaction $\tr_j$, then
all of its private versions are greater than $\tr_j$'s.

Given that private versions ascribe a link between transactions and variables,
OptSVA then uses them to permit or deny access to variables. Which transaction
can access variable $\obj$ is defined by its \emph{local version} $\lv{\obj}$.
Specifically, the local version of a variable $\obj$ is always equal to the
private version of the transaction $\tr_i$ that most recently finished working
on $\obj$, i.e. when $\tr_i$ commits or aborts it sets $\lv{\obj}$ to
$\pv{\obj}{i}$ (lines \ref{l:commit-lv} and \ref{l:abort-lv}, respectively).
The transaction that can access $\obj$ is the next transaction after the one
that stopped using $\obj$ last. That is, the one whose private version for
$\obj$ is one greater than the local version of $\obj$. Hence, in order to
access $\obj$, $\tr_i$ must wait until the condition $\pv{\obj}{i} - 1 =
\lv{\obj}$ is met (see, e.g., line \ref{l:read-access}). We will refer to this
condition as the \emph{access condition}.

An example of how this mechanism works is shown in \rfig{fig:access-control}.
The diagram depicts a history consisting of operations executed by transactions
on a time axis. Every line depicts the operations executed by a particular
transaction.
The symbol
\protect\tikz{
    \protect\draw[] (0,0) -- ++(0.25,0) node[dot] {} -- ++(0.25,0);
} 
denotes a complete operation execution.
The inscriptions above operation executions denote operations executed by the
transactions, e.g. $\trop{i}{\obj}{}\!\to\!1$ denotes that a read operation on
variable $\obj$ is executed by transaction $\tr_i$ and returns $1$, and
$\twop{i}{\obj}{1}$ denotes that a write operation writing $1$ to $\obj$ is
executed by $\tr_i$, and $\tryC_i\!\to\!\co_i$ indicates that $\tr_i$ attempts
to commit and succeeds because it returns $\co_i$.
On the other hand, the symbol 
\protect\tikz{
    \protect\draw[] (0,0) -- ++(0.1,0) node[inv] (a) {} 
                             ++(0.4,0) node[res] (b) {}
                          -- ++(0.1,0);
    \protect\draw[wait] (a) -- (b);
}
denotes an operation execution split into the invocation and the response
event to indicate waiting, or that the execution takes a long time.
In that case the inscription above is split between the events, e.g., a read
operation execution would show $\trop{i}{\obj}{}$ above the invocation, and
$\!\to\!1$ over the response.
If waiting is involved, the arrow 
\protect\tikz{
    \protect\draw[hb] (0,0.2) .. controls +(270:.25) and +(90:0.25) .. (0.5,0.0);
} 
is used to emphasize a happens before relation between two events.
Annotations below events emphasize the state of counters or performed operations
within the concurrency control algorithm (used as necessary).

In \rfig{fig:access-control}, $\trobj_i$ and $\trobj_j$ attempt to
access shared variable $\objx$ at the same time. Transaction $\trobj_i$ executes
start first, so  $\pv{\objx}{i} = 1$, and $\trobj_j$ executes second, so
$\pv{\objx}{i} = 2$. This then determines in which order the transactions
access $\objx$: initially $\lv{\objx} = 0$, so $\trobj_j$ is not able to pass
the access condition and execute its read operation. However, $\trobj_i$ can
pass the access condition and it executes its operation without delay. Once
$\trobj_i$ commits, it sets $\lv{\objx}$ to $1$, so  $\trobj_j$ then becomes
capable of passing the access condition and finishing executing its
read operation.
In the mean time, transaction $\trobj_k$ can proceed to access $\objy$
completely in parallel.

Below we describe the early release mechanism with means to ensure commit order
and forced aborts that we use to improve the effectiveness of the version
control mechanism.

\def\magicadjust{\hspace{1cm}}
\begin{figure}[t]
\magicadjust
\begin{tikzpicture}
     \draw
           (0,2.5)      node[tid]       {$\trobj_i$}
                        node[aop]       {$\init_i$} %
                        node[dot]       {}
                        node[not]       {$\pv{\objx}{i}\!\gets\!1$}

      -- ++(3,0)        node[aop]       {$\twop{i}{\objx}{1}$}
                        node[dot]       {}

      -- ++(2,0)        node[aop]       {$\tryC_i\!\to\!\co_i$}
                        node[dot] (ci)  {}         
                        node[not]       {$\lv{\objx}\!\gets\!1$}
                        ;

     \draw
         (0.25,1.25)    node[tid]       {$\trobj_{j}$}
                        node[aop]       {$\init_{j}$} %
                        node[dot]       {} 
                        node[not]       {$\pv{\objx}{j}\!\gets\!2$}

      -- ++(2.00,0)     node[aop]       {$\trop{j}{\objx}{}$}
                        node[inv] (rx0) {}
                        node[not]       {$\lv{\objx}\!=\!1?$}
 
         ++(4.00,0)     node[aop]       {$\to\!2$}
                        node[res] (rx)  {}                       

      -- ++(2.00,0)     node[aop]       {$\tryC_j\!\to\!\co_j$}
                        node[dot]       {}
                        node[not]       {$\lv{\objx}\!\gets\!1$}
                        ;
      \draw[hb] (ci) \squiggle (rx);
      \draw[wait] (rx0) -- (rx);

     \draw
           (0.5,0)      node[tid]       {$\trobj_k$}
                        node[aop]       {$\init_k$} %
                        node[dot]       {} 
                        node[not]       {$\pv{\objy}{k}\!\gets\!1$}

       -- ++(3.0,0)     node[aop]       {$\twop{k}{\objy}{1}$}
                        node[dot]       {}

      -- ++(3.00,0)     node[aop]       {$\tryC_k\!\to\!\co_k$}
                        node[dot]       {}         
                        node[not]       {}
                        ;

\end{tikzpicture}
\caption{\label{fig:access-control} Concurrency control \emph{via} versioning.}
\end{figure}
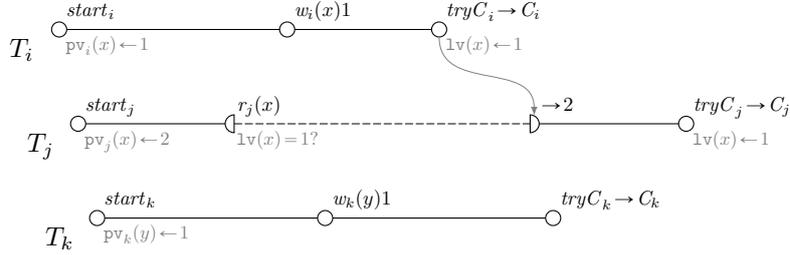

\subsubsection{Early release}
\label{sec:early-release}

The second basic feature of OptSVA is early release based on \emph{a priori}
knowledge.
Each transaction 
knows
the maximum
number of times it will read and write each individual variable at the start of execution.
(This information can be provided by the programmer, supplemented by a type
checker \cite{Woj05b}, or generated by static analysis \cite{SW12}.)
These \emph{upper bounds} for reads and writes 
are denoted for transaction $\tr_i$ and $\obj$ as, respectively, 
$\rub{i}{\obj}$ and
$\wub{i}{\obj}$. %
Then, each transaction can count accesses to each variable as they occur
using \emph{read} and \emph{write counters}: $\rc{i}{\obj}$
and $\wc{i}{\obj}$
(lines \ref{l:read-inc} and \ref{l:write-wc-inc}).
When the read or write counter for $\obj$ reaches the upper bound for $\obj$, 
the transaction knows that no further accesses of a
particular type will occur afterward. 
When it is apparent that transaction $\tr_i$ will perform no further
modifications on $\obj$, another transaction can start accessing $\obj$ right
away, without waiting for $\tr_i$ to commit. Hence, if after a write it is true
that $\wc{i}{\obj} = \wub{i}{\obj}$ (line \ref{l:write-no-reads-ub}), then
procedure \ppr{release} is (eventually) executed and $\tr_i$ sets $\lv{\obj}$
to $\pv{\obj}{i}$.

This is illustrated further in \rfig{fig:early-release}. Here, transaction
$\trobj_i$ and transaction $\trobj_j$ both try to access $\objx$. Like in
\rfig{fig:access-control}, since $\trobj_i$'s private version for $\objx$ is
lower than $\trobj_j$'s, the former manages to access $\objx$ first, and
$\trobj_j$ waits until $\trobj_i$ is released. Unlike in
\rfig{fig:access-control}, $\trobj_i$ has upper bound information on writes
on $\objx$ via $\wub{i}{\objx}$: it knows that it will not
read from $\obj$ and will write to $\objx$ at most once.
So, $\trobj_i$ releases $\objx$ immediately after its write to $\objx$, rather than
waiting to do so until commit. In effect, $\trobj_j$ can access $\objx$
earlier.

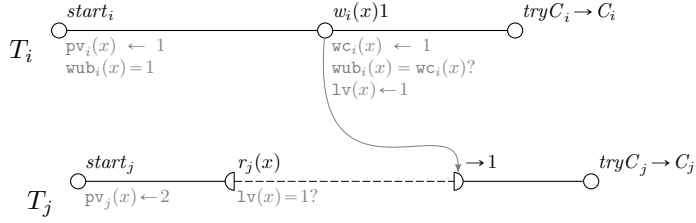
\begin{figure}[t]
\magicadjust
\begin{tikzpicture}
     \draw
           (0,2.5)      node[tid]       {$\trobj_i$}
                        node[aop]       {$\init_i$} %
                        node[dot]       {} 
                        node[not, text width=3cm] 
                                        {$\pv{\objx}{i}\!\gets\!1$
                                         $\wub{i}{\objx}\!=\!1$
                                         }

      -- ++(3.5,0)      node[aop]       {$\twop{i}{\objx}{1}$}
                        node[dot] (wx)  {}
                        node[not, text width=3cm]       
                                        {$\wc{i}{\objx}\!\gets\!1$
                                         $\wub{i}{\objx}\!=\!\wc{i}{\objx}?$
                                         $\lv{\objx}\!\gets\!1$}

      -- ++(2.5,0)      node[aop]       {$\tryC_i\!\to\!\co_i$}
                        node[dot] (ci)  {}         
                        ;

     \draw
         (0.25,0.5)     node[tid]       {$\trobj_{j}$}
                        node[aop]       {$\init_{j}$} %
                        node[dot]       {} 
                        node[not]       {$\pv{\objx}{j}\!\gets\!2$}

      -- ++(2.00,0)     node[aop]       {$\trop{j}{\objx}{}$}
                        node[inv] (rx0) {}
                        node[not]       {$\lv{\objx}\!=\!1?$}
 
         ++(3.00,0)     node[aop]       {$\to\!1$}
                        node[res] (rx)  {}                       

      -- ++(1.75,0)     node[aop]       {$\tryC_j\!\to\!\co_j$}
                        node[dot]       {}
                        node[not]       {}
                        ;

      \draw[hb] (wx) .. controls +(265:2.5) and +(90:1) .. (rx);
      \draw[wait] (rx0) -- (rx);

\end{tikzpicture}
\caption{\label{fig:early-release} Early release \emph{via} upper bounds.}
\end{figure}

\subsubsection{Commit Order}

Although OptSVA is pessimistic and prevents transactions from aborting on
conflict, for expressiveness, it allows the programmer to manually invoke the abort operation.
Therefore, it is possible for any transaction spontaneously to abort in effect.
Such manual aborts can be useful to the programmer to implement conditional
rollbacks, and to the system to recover from failures (e.g., in distributed
TM).
However, this design decision also makes it necessary to
enforce the order in which transactions commit
to prevent a situation where transaction $\tr_i$ releases $\obj$ early
and subsequently aborts, but before $\tr_i$ does abort, $\tr_j$ reads $\obj$ and commits.
That would mean that $\tr_j$ committed having acted on an invalid, inconsistent
value of $\obj$, which is incorrect behavior (e.g., according to serializability
\cite{Pap79}).

OptSVA prevents that sort of erroneous situation by ordering commits in the same order as
accesses to variables and forcing an abort if the previous transaction also aborted.
This is instituted through \emph{local terminal versions}, denoted
$\ltv{\obj}$, that specify which transaction that used $\obj$ last completed by
either committing or aborting. I.e., each transaction writes its private version
to $\ltv{\obj}$, if it only used $\obj$, just before finishing the \ppr{commit} or
\ppr{abort} procedure (lines \ref{l:commit-ltv} and \ref{l:abort-ltv}).
Then, the local terminal version is used to enforce commit order, as every
committing or aborting transaction must wait at \emph{termination condition} $\pv{\obj}{i} - 1 =
\ltv{\obj}$ (lines \ref{l:commit-commit} and \ref{l:abort-commit}) before it
can commit or abort.  Thus, a transaction that accesses $\obj$ does not
complete until the last transaction that previously accessed $\obj$ commits or
aborts.

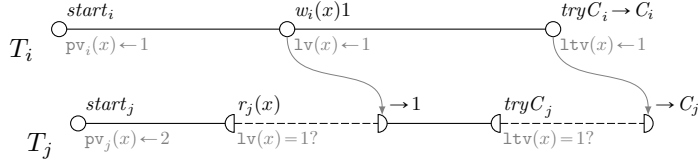
\begin{figure}[t]
\magicadjust
\begin{tikzpicture}
     \draw
           (0,2.5)      node[tid]       {$\trobj_i$}
                        node[aop]       {$\init_i$} %
                        node[dot]       {} 
                        node[not]       {$\pv{\objx}{i}\!\gets\!1$}

      -- ++(3.0,0)      node[aop]       {$\twop{i}{\objx}{1}$}
                        node[dot] (wx)  {}
                        node[not]       {$\lv{\objx}\!\gets\!1$}

      -- ++(3.5,0)     node[aop]       {$\tryC_i\!\to\!\co_i$}
                        node[dot] (ci)  {}
                        node[not]       {$\ltv{\objx}\!\gets\!1$}         
                        ;

     \draw
         (0.25,1.25)    node[tid]       {$\trobj_{j}$}
                        node[aop]       {$\init_{j}$} %
                        node[dot]       {} 
                        node[not]       {$\pv{\objx}{j}\!\gets\!2$}

      -- ++(2.00,0)     node[aop]       {$\trop{j}{\objx}{}$}
                        node[inv] (rx0) {}
                        node[not]       {$\lv{\objx}\!=\!1?$}
 
         ++(2.00,0)     node[aop]       {$\to\!1$}
                        node[res] (rx)  {}                       

      -- ++(1.50,0)     node[aop]       {$\tryC_j$}
                        node[inv] (cj0) {}
                        node[not]       {$\ltv{\objx}\!=\!1?$} 

         ++(2.,0)     node[aop]       {$\to\!\co_j$}
                        node[res] (cj)  {}
                        node[not]       {}
                        ;

      \draw[hb] (wx) \squiggle (rx);
      \draw[hb] (ci) \squiggle (cj);
      \draw[wait] (rx0) -- (rx);
      \draw[wait] (cj0) -- (cj);

\end{tikzpicture}
\caption{\label{fig:commit-order} Commit ordering preservation.}
\end{figure}

An example of how this mechanism affects execution is shown in
\rfig{fig:commit-order}. 
Here, $\trobj_i$, $\trobj_j$ both access $\objx$ and they respectively get the
values  of $1$ and $2$ of $\pvf$ for $\objx$.  Transaction $\trobj_i$ accesses
$\objx$ first and releases it early, setting $\lv{\objx}$ to $1$. This allows
$\trobj_j$ to pass the access condition and read $\objx$. Transaction
$\trobj_j$ subsequently attempts to commit. However, in order to commit
$\trobj_j$ must pass the termination condition $\pv{\objx}{j} - 1 =
\ltv{\objx}$, which will not be satisfied until $\trobj_i$ sets $\ltv{\objx}$
to its own $\pv{\objx}{i}$. Hence $\trobj_j$ can only complete to commit
after $\trobj_i$ commits.
In general, this condition is checked for every variables in the access set $\accesses{j}$ of transaction $\tr_j$.
In effect the commit order preserves the private version order of transactions
for every variable.

\subsubsection{Forced Aborts}
Furthermore, if some transaction reads a value written by another, and the latter aborts, 
then the former cannot be allowed to commit having possibly acted upon inconsistent state. 
Hence, the transaction must be forced to abort.

To enforce aborts, OptSVA marks which version of a variable is the last
consistent version via its \emph{current version} $\cv{\obj}$---a counter
shared by all transactions.
This is used in conjunction with its \emph{recovery version}
$\rv{i}{\obj}$---the last consistent version seen by $\tr_i$---to check
whether transaction $\tr_i$ is using a consistent (current) version or not.
Whenever transaction $\trobj_i$ gains access to
shared variable $\obj$ for the first time, it runs procedure \ppr{checkpoint}
where it reads the state of $\obj$
and stores it in its buffer $\stored{i}{\obj}$
(line \ref{l:checkpoint-copy}).
It then
sets its recovery version for $\obj$ $\rv{i}{\obj}$ to $\obj$'s current
version (line \ref{l:checkpoint-cv}).
Since $\cv{\obj}$ is set to some transaction $\tr_j$'s private version for
$\obj$ (line \ref{l:release-cv}) whenever $\tr_j$ releases $\obj$ or commits
(lines \ref{l:release-cv} and \ref{l:commit-cv}), then $\rv{i}{\obj}$ is equal
to the private version of a transaction that most recently finished operating
on $\obj$.
Then, whenever some transaction $\trobj_j$ aborts and restores shared variable
$\obj$ from the backup copy (line \ref{l:abort:aset})
it also sets $\cv{\obj}$
back to $\rv{j}{\obj}$ (line \ref{l:abort:recover-cv}).

In addition, whenever a transaction tries to commit or access a shared
variable, it must test the consistency of all the variables it operates on. Thus,
e.g., $\tr_i$ cannot proceed to access $\objx$ unless $\rv{i}{\objy} = \cv{\objy}$ for each variable $\objy$ in its access set,
and must abort otherwise (e.g. lines \ref{l:read-abort}--\ref{l:read-abort-go}). Similarly, if
$\trobj_i$ attempts to commit, there must be no variable $\objy$ in its access set
for which $\rv{i}{\objy} > \cv{\objy}$, or $\trobj_i$ must be forced to abort
(line \ref{l:commit-abort}--\ref{l:commit-abort2}).
Hence, if $\trobj_i$ gains access to $\obj$ after the
previous $\trobj_j$ releases it, and $\trobj_j$ subsequently
aborts and sets $\cv{\obj}$ to a new (lesser) value $\rv{j}{\objx}$, then $\trobj_i$ will be forced to
abort either when accessing $\objx$ later, or when attempting to commit.
The condition for aborting is always checked for all variables rather than just
the one being accessed, in order to abort as quickly as possible, and to prevent the
transaction from operating on both consistent and invalidated variables
simultaneously.

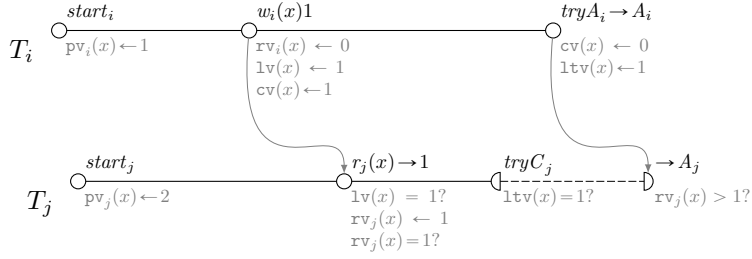
\begin{figure}[t]
\magicadjust
\begin{tikzpicture}
     \draw
           (0,2.5)      node[tid]       {$\trobj_i$}
                        node[aop]       {$\init_i$} %
                        node[dot]       {} 
                        node[not]       {$\pv{\objx}{i}\!\gets\!1$}

      -- ++(2.5,0)      node[aop]       {$\twop{i}{\objx}{1}$}
                        node[dot] (wx)  {}
                        node[not, text width=2.5cm]
                                        {   $\rv{i}{\objx}\!\gets\!0$
                                            $\lv{\objx}\!\gets\!1$
                                            $\cv{\objx}\!\gets\!1$
                                         }

      -- ++(4,0)     node[aop]       {$\tryA_i\!\to\!\ab_i$}
                        node[dot] (ci)  {}
                        node[not, text width=2.5cm]
                                        {   $\cv{\objx}\!\gets\!0$
                                            $\ltv{\objx}\!\gets\!1$}         
                        ;

     \draw
         (0.25,0.50)    node[tid]       {$\trobj_{j}$}
                        node[aop]       {$\init_{j}$} %
                        node[dot]       {} 
                        node[not]       {$\pv{\objx}{j}\!\gets\!2$}

      -- ++(3.5,0)      node[aop]       {$\trop{j}{\objx}{}\!\to\!1$}
                        node[dot] (rx)  {}          
                        node[not, text width=2.75cm]
                                        {
                                            $\lv{\objx}\!=\!1?$
                                            $\rv{j}{\objx}\!\gets\!1$
                                            $\rv{j}{\objx}\!=\!1?$
                                        }             

      -- ++(2,0)     node[aop]       {$\tryC_j$}
                        node[inv] (cj0) {}
                        node[not]       {$\ltv{\objx}\!=\!1?$} 

         ++(2.0,0)     node[aop]       {$\to\!\ab_j$}
                        node[res] (cj)  {}
                        node[not]       {$\rv{j}{\objx} > 1?$}
                        ;

      \draw[wait] (cj0) -- (cj);
      \draw[hb] (wx) .. controls +(265:2.5) and +(90:1) .. (rx);
      \draw[hb] (ci) .. controls +(265:2.5) and +(90:1) .. (cj);

\end{tikzpicture}
\caption{\label{fig:forced-abort} Forced abort.}
\end{figure}

We show an example of this in \rfig{fig:forced-abort}. Here, $\trobj_i$ and
$\trobj_j$ access $\objx$ and have private values for $\objx$ equal to
$1$ and $2$, respectively. Hence $\trobj_i$ accesses $\objx$ first. As this is
executed $\trobj_i$ sets its recovery version to $0$, the value of the
current version  for $\objx$. Then, after the write
operation finishes executing, the transaction releases $\objx$ by setting the
local version to $1$ and sets the current version  to its own private
version, i.e. $1$. Subsequently $\trobj_j$ meets the access condition and accesses
$\objx$ for the first time, setting its own recovery version to $1$ (as
$\cv{\objx} = 1$). Since $\rv{i}{\objx} =
\cv{\objx}$, the access is successful.  However, as $\trobj_j$ tries to commit,
it is delayed because it cannot satisfy the commit condition. Meanwhile
transaction $\trobj_i$ aborts. As it does so, it sets the current version
 to its recovery version  equal to $0$.  Then, $\trobj_i$ sets
the local terminal version to its own private version, allowing $\trobj_j$ to
resume committing. However, $\trobj_j$ 
cannot satisfy the condition 
$\rv{j}{\objx} > \cv{\objx}$ during commit, since $\rv{i}{\objx} = 1$ and
$\cv{\objx} = 0$. Hence, $\trobj_j$ is forced to abort.

Note that if no programmatic aborts occur in the system, the system does not
experience any forced aborts whatsoever. If the programmer does issue manual
aborts, however, cascading aborts can occur.

Below we proceed to describe the three optimizations with respect to the version control mechanism that are employed by OptSVA.

\subsection{Read-only Variables}
\label{sec:read-only-variables}

Since originally versioning algorithms did not distinguish between reads and
writes, they did not allow read-only transactions to be executed in parallel to
other read-only transactions.
This is a run-of-the-mill optimization found in all but a small number of TMs,
so it is also introduced in OptSVA. However, OptSVA goes a
step further, and allows partial parallelization of transactions whenever a
variable in a transaction is only read from and not written to, without
requiring that all the variables in a transaction are not written to.

Whenever transaction $\tr_i$ accesses $\obj$ in such a way that it reads from
$\obj$ but does not write to $\obj$ (and this is known \emph{a priori}---i.e.,
$\rub{i}{\obj} > 0$ and $\wub{i}{\obj} = 0$), we will refer to $\obj$ as being
a \emph{read-only variable} in $\tr_i$.
In the case of such variables, OptSVA can optimize the accesses by buffering
the variable and reading the buffer instead of the actual variable.
In addition, since all the reads will be done using the buffer, and the upper
bounds indicate that no writes will follow, the variable can be released after
it is buffered, irrespective of what operations the transaction will execute
later.

Obviously, it is best for parallelism to release any variable as soon as it is
no longer needed by a transaction, because it allows other transactions to start
acting sooner.
Since read-only variables are not needed after they are
buffered, they can be released immediately after this happens. The variable
must be buffered before or during the first read operation on it is executed,
but it could be buffered before that point, even during transaction start.
However, in order to buffer a variable, its state must be viewed, so, for the
sake of consistency, buffering within versioning concurrency control must be
done only after the transaction passes the access condition. 
Since waiting at the access condition would prevent the transaction from
executing operations on other variables or performing local computations, it is
best for parallelism for the transaction not to start waiting until it is
absolutely necessary.

The algorithm finds balance between buffering as soon as possible and delaying synchronization
much as necessary by executing it asynchronously.
This is achieved by using the  {\tt async run P when C} construct which relegates the execution of procedure {\tt P} to some separate thread. However, before the thread starts executing {\tt P} it waits until condition {\tt C} is satisfied.
This allows the transaction
to wait at condition {\tt C} without preventing the procedure from delaying
other operations that could be executing in the mean time. 
On the other hand, {\tt P} will be executed as soon as {\tt C} is satisfied, so
as soon as it is safe.

OptSVA executes buffering via procedure \ppr{read\_buffer}. This
procedure is relegated to asynchronous execution at lines
\ref{l:reads-c}--\ref{l:reads-b}, and will execute once the access condition is
satisfied.
Within \ppr{read\_buffer}, the transaction $\tr_i$ saves the value of some
variable $\obj$ to its buffer $\buf{i}{\obj}$ (line \ref{l:rb-copy}), and
releases it immediately afterward by executing \ppr{release} (line
\ref{l:rb-release}). Since it is possible that the transaction that wrote the
value of $\obj$ that is being buffered will subsequently abort, $\tr_i$ also
updates its recovery value (line \ref{l:rb-checkpoint}), but it does not need
to make a checkpoint for $\obj$, since the transaction will not modify $\obj$.
Once read-only variable $\obj$ is buffered, read operations can use the buffer
to retrieve that value, without accessing the variable (line
\ref{l:read-only:return}), so without waiting. However, a read on a read-only
variable cannot be executed until buffering is finished (line
\ref{l:read-ro-wait}), which we indicate using the {\tt join with P} construct.

Since a transaction does not modify a read-only variable, if it aborts, it does
not need to force other transactions to abort to maintain consistency.
Hence, the transaction tries to
immediately perform all commit-related operations for a read-only variable
immediately after buffering it. This involves waiting for the local terminal
version of the object, so by analogy to buffering, the procedure is executed
asynchronously, so as not to block other operations.
The procedure that executes the commit for variable $\obj$ is {\tt read\_commit} 
and it is started asynchronously at line \ref{l:rb-async}. The procedure
executes a simplified version of \ppr{commit} for just $\obj$. Hence, once
\ppr{commit} is executed by the transaction for other variables, it can be
skipped for $\obj$, and the transaction simply waits for \ppr{read\_commit} to
finish executing.

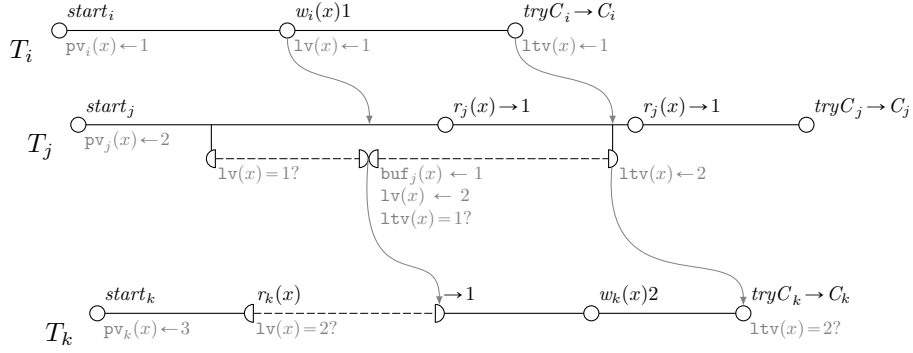
\begin{figure}[t]
\begin{center}
\begin{tikzpicture}
     \draw
           (0,2.5)      node[tid]       {$\trobj_i$}
                        node[aop]       {$\init_i$} %
                        node[dot]       {} 
                        node[not]       {$\pv{\objx}{i}\!\gets\!1$}

      -- ++(3.0,0)      node[aop]       {$\twop{i}{\objx}{1}$}
                        node[dot]  (wx) {}
                        node[not]       {$\lv{\obj}\!\gets\!1$}

      -- ++(3.0,0)      node[aop]       {$\tryC_i\!\to\!\co_i$}
                        node[dot]  (ci) {}
                        node[not]       {$\ltv{\objx}\!\gets\!1$}
                        ;

     \draw
          (0.25,1.25)   node[tid]       {$\trobj_j$}
                        node[aop]       {$\init_j$} %
                        node[dot]       {} 
                        node[not]       {$\pv{\objx}{j}\!\gets\!2$}

      -- ++(1.75,0)      node    (fork1) {}
      -- ++(2.08,0)      node    (join1) {}

      -- ++(1,0)      node[aop]       {$\trop{j}{\objx}{}\!\to\!1$}
                        node[dot]       {}
      
      -- ++(2.2,0)      node    (join2) {}

      -- ++(0.3,0)      node[aop]       {$\trop{j}{\objx}{}\!\to\!{1}$}
                        node[dot]       {}

      -- ++(2.25,0)     node[aop]       {$\tryC_j\!\to\!\co_j$}
                        node[dot] (cj)  {}
                        ;

      \draw
           (2.0,.8)   node[inv] (ro1) {} %
                        node[not]       {$\lv{\objx}\!=\!1?$}
 
          ++(2,0)    node[res] (ro2) {}

          ++(.14,0)     node[inv] (ro3) {}
                        node[not, text width=2.5cm]       {
                            $\buf{j}{\objx}\!\gets\!1$
                            $\lv{\objx}\!\gets\!2$
                            $\ltv{\objx}\!=\!1?$
                            }
          ++(3.14,0)    node[res] (ro4) {}
                        node[not]       {$\ltv{\objx}\!\gets\!2$}
                        ;
                       
      \draw[] (fork1.center) -- (ro1);
      \draw[] (ro4) -- (join2.center);
      \draw[hb] (ci) \squiggle (join2.center);
      \draw[wait] (ro1) -- (ro2);
      \draw[wait] (ro3) -- (ro4);
      \draw[hb] (wx) \squiggle (join1.center);

     \draw
         (0.5,-1.25)     node[tid]       {$\trobj_{k}$}
                        node[aop]       {$\init_{k}$} %
                        node[dot]       {} 
                        node[not]       {$\pv{\objx}{k}\!\gets\!3$}

      -- ++(2.0,0)      node[aop]       {$\trop{k}{\objx}{}$}
                        node[inv] (rx0) {}                     
                        node[not]       {$\lv{\objx}\!=\!2?$}

         ++(2.5,0)      node[aop]       {$\!\to\!1$}
                        node[res] (rx2) {} 
                       
      -- ++(2.,0)      node[aop]       {$\twop{k}{\objx}{2}$}
                        node[dot]        {}                     

       -- ++(2.,0)     node[aop]       {$\tryC_k\!\to\!\co_k$}
                        node[dot] (ck)  {}
                        node[not]       {$\ltv{\objx}\!=\!2?$} 
                        ;

      \draw[hb] (ro2.east) .. controls +(265:2) and +(90:1) .. (rx2);
      \draw[hb] (ro4) .. controls +(265:2) and +(90:1) .. (ck);
      \draw[wait] (rx0) -- (rx2);
\end{tikzpicture}
\end{center}
\caption{\label{fig:read-only-opt} Read-only variable. 
}
\end{figure}

We show an example of an execution of a transaction  $\trobj_j$ with
a read-only variable $\objx$ in \rfig{fig:read-only-opt}.
Transaction $\trobj_j$
asynchronously waits for the access condition on $\objx$ to be met right after
$\tr_j$ starts, but before any reads actually occur.
A parallel
line below transaction (such as the one below $\trobj_j$) indicates procedures
executed asynchronously with respect to the thread executing the transaction.
Meanwhile $\trobj_j$ can
perform local operations or operations on other variables without obstacle.
Once $\trobj_i$ releases $\objx$, $\trobj_j$ immediately buffers $\objx$,
and releases it. Then $\trobj_j$ asynchronously
tries to commit $\objx$, which requires that it waits for the appropriate local terminal version of $\objx$.  Meanwhile $\trobj_k$ can now access $\objx$ in parallel to $\trobj_j$
and even write to it, without interfering with $\trobj_j$'s consistency.
Once $\trobj_i$
commits, $\trobj_j$ can then asynchronously commit $\objx$, which then allows
$\trobj_k$ to commit earlier than it would have otherwise.
Since $\trobj_j$ treats $\objx$ as read-only and hence releases it earlier,
transaction $\trobj_k$ is able to execute its operations much sooner, and thus
shorten the total execution time of the three transactions.

From the example it is apparent, that the read-only variable optimization moves the point at which such a
variable is acquired, released, and committed forward in time.
The earlier a shared variable is released by a transaction, the earlier another
transaction can start using it, increasing the possibility of acting in
parallel, and, therefore, shortening the schedule of execution.

\subsection{Delayed Synchronization on First Write}
\label{sec:delayed-synchronization-on-first-write}

If the first operation that a transaction executes on a particular shared
variable is a write operation, then all read operations on that variable are
\emph{local}, i.e., they only need to view what the current transaction wrote,
and can ignore writes by other transactions. Hence, there is no need for the
transaction to synchronize on this variable with other transactions for the
sake of those operations. The synchronization is only needed to prevent the
current transaction from writing a value to the variable in the middle
of another transaction's operations on it. But if the write is saved to a
buffer, rather than immediately updating the state of the variable, the
synchronization can be delayed until after the write itself, or even after any
of the successive read operations.

Since it is beneficial to synchronize as late as possible while performing
other tasks beforehand, OptSVA then never checks access conditions on writes
(see procedure \ppr{write}): either the transaction started with a write, and
no synchronization is necessary, or there was a preceding read that already did
all the necessary synchronization. Instead, the operation is performed on a
buffer (line \ref{l:write:set-buf}). Then, since all the written values are
only visible to the current transaction, the transaction must at some point
update the state of the actual variable. This is done either upon executing the
last write or during commit.
In the former case, when the upper bound on writes is reached (line
\ref{l:write-ub}), the transaction asynchronously starts procedure
\ppr{write\_buffer} (line \ref{l:write-async}), which executes when the access
condition is met, and updates the state of the variable (line
\ref{l:write-buffer:rset}).
If the upper bound is not reached during execution, the transaction will
instead execute procedure \ppr{catch\_up} during commit, and update the
variable there (line \ref{l:commit:rset}), also after waiting at the access
condition (line \ref{l:commit:access}).

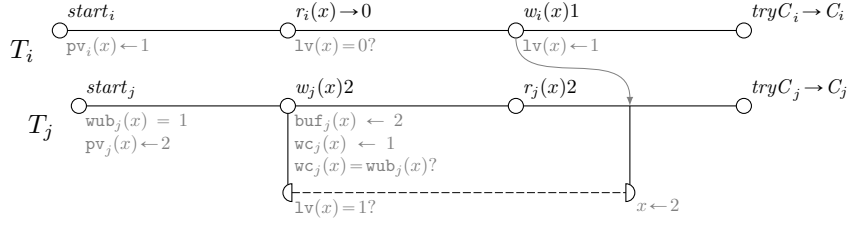
\begin{figure}[t]
\begin{center}
\begin{tikzpicture}
     \draw
           (0,3.5)      node[tid]       {$\trobj_i$}
                        node[aop]       {$\init_i$} %
                        node[dot]       {} 
                        node[not]       {$\pv{\objx}{i}\!\gets\!1$}

      -- ++(3,0)        node[aop]       {$\trop{i}{\objx}{}\!\to\!0$}
                        node[dot]  (rx) {}
                        node[not]       {$\lv{\obj}\!=\!0?$}

      -- ++(3,0)        node[aop]       {$\twop{i}{\objx}{1}$}
                        node[dot]  (wx) {}
                        node[not]       {$\lv{\obj}\!\gets\!1$}

      -- ++(3,0)        node[aop]       {$\tryC_i\!\to\!\co_i$}
                        node[dot]  (ci) {}
                        ;

     \draw
           (0.25,2.5)   node[tid]       {$\trobj_j$}
                        node[aop]       {$\init_j$} %
                        node[dot]       {} 
                        node[not, text width=2.75cm]
                                        {
                                        $\wub{j}{\objx}\!=\!1$
                                        $\pv{\objx}{j}\!\gets\!2$}

      -- ++(2.75,0)     node[aop]       {$\twop{j}{\obj}{2}$}
                        node[dot] (wjx)      {}
                        node[not, text width=3.0cm]       {
                            $\buf{j}{\objx}\!\gets\!2$
                            $\wc{j}{\objx}\!\gets\!1$
                            $\wc{j}{\objx}\!=\!\wub{j}{\objx}?$
                            }
      
      -- ++(3,0)      node[aop]         {$\trop{j}{\objx}{2}$}
                        node[dot] (rx0j){}

      -- ++(1.5,0)        node (join) {}

      -- ++(1.5,0)     node[aop]       {$\tryC_j\!\to\!\co_j$}
                        node[dot] (cj) {}
                        ;
      \draw
           (3,1.35)   node[inv] (ro1) {} %
                        node[not]     {$\lv{\objx}\!=\!1?$}
 
          ++(4.5,0)    node[res] (ro2) {}
                        node[not]      {$\objx\!\gets\!2$}
                        ;

      \draw[] (wjx) -- (ro1);
      \draw[] (ro2) -- (join.center);
      \draw[wait] (ro1) -- (ro2);
      \draw[hb] (wx) \squiggle (join.center);

\end{tikzpicture}
\end{center}
\caption{\label{fig:initial-writes-opt} Delayed synchronization on first write. 
}
\end{figure}

We illustrate this optimization further in \rfig{fig:initial-writes-opt}. Here
transaction $\trobj_i$ can pass access condition for $\objx$ first, but
nevertheless $\trobj_j$ performs a write simultaneously, since it writes to the
buffer rather than wait at the access condition.
Transaction $\trobj_j$ only waits at the access condition when it had performed
all of its write operations (of which there is one) and starts a separate
thread (indicated by the line below) to write the changes to the variable once the access condition is passed.
The thread passes the access condition once $\trobj_i$ releases $\objx$. Then,
$\trobj_j$ applies the value from the  buffer to $\objx$.

\subsection{Early Release on Last Write}
\label{sec:early-release-on-last-write}

Various TMs with early release determine the point at which variables are
released variously. For instance, DATM \cite{RRHW09} releases variables after
each operation,
erring on the side of efficiency and guaranteeing only
conflict-serializability. SVA, on the other hand, errs on the side of caution
and only allows early release after last access to some variable, which it must
do because it treats read and write operations uniformly. OptSVA improves on this,
since it distinguishes between reads and writes, so early release is done
after last write not last access. 
In effect all reads following last write are executed as if privatized.
We argue in \cite{SW15-arxiv} that this
approach is a solid compromise for TMs with early release.

The early release happens
if at some point in the execution of transaction $\tr_i$, the
upper bound on the number of writes for some variable $\objx$ is reached when performing a write (line \ref{l:write-ub}).
The transaction asynchronously executes \ppr{write\_buffer} in that
instance for the purpose of applying the changes from the buffer to the actual
shared variable. After this is done, $\objx$ will no longer be accessed
directly by the transaction, so $\tr_i$ also executes \ppr{release} (at line
\ref{l:wb-release}), which sets $\lv{\obj}$ to $\pv{\objx}{i}$, which allows
other transactions to pass the access condition.
Nevertheless, since $\obj$ was buffered during writes, subsequent reads still have access to
a local, consistent value of $\objx$ (retrieved from the buffer at line \ref{l:local-read:return}).

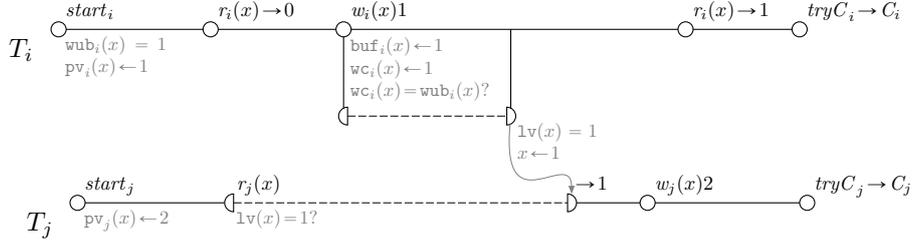
\begin{figure}[t]
\begin{center}

\begin{tikzpicture}
     \draw
           (0,2.5)      node[tid]       {$\trobj_i$}
                        node[aop]       {$\init_i$} %
                        node[dot]       {} 
                        node[not, text width=2.75cm]
                                        {$\wub{i}{\objx}\!=\!1$
                                        $\pv{\objx}{i}\!\gets\!1$}

      -- ++(2,0)     node[aop]       {$\trop{i}{\objx}{}\!\to\!0$}
                        node[dot] (rxi1){}
                        node[not]       {}

      -- ++(1.75,0)     node[aop]       {$\twop{i}{\objx}{1}$}
                        node[dot] (wxi) {}
                        node[not] {
                            \begin{multinote}
                            \hspace{-6pt}$\buf{i}{\objx}\!\gets\!1$\\
                            \hspace{-6pt}$\wc{i}{\objx}\!\gets\!1$\\
                            \hspace{-6pt}$\wc{i}{\objx}\!=\!\wub{i}{\objx}?$\\
                            \end{multinote}
                        }

      -- ++(2.2, 0)      node (join) {}                        

      -- ++(2.3,0)     node[aop]       {$\trop{i}{\objx}{}\!\to\!1$}
                        node[dot] (rxi1){}
                        node[not]       {}

      -- ++(1.5,0)     node[aop]       {$\tryC_i\!\to\!\co_i$}
                        node[dot] (ci)  {}
                        ;

      \draw
           (3.75,1.35)  node[inv] (ro1) {} %
                        node[not]     {}
 
          ++(2.2,0)     node[res] (ro2) {}
                        node[not, text width=2cm]      {
                        $\lv{\objx}\!=\!1$
                        $\objx\!\gets\!1$}
                        ;

      \draw
           (0.25,0.2)   node[tid]       {$\trobj_j$}
                        node[aop]       {$\init_j$} %
                        node[dot]       {} 
                        node[not]       {$\pv{\objx}{j}\!\gets\!2$}

      -- ++(2,0)        node[aop]       {$\trop{j}{\objx}{}$}
                        node[inv]  (rxj0) {}
                        node[not]       {$\lv{\obj}\!=\!1?$}

         ++(4.5,0)      node[aop]       {$\!\to\!1$}
                        node[res]  (rxj)     {}

      -- ++(1,0)     node[aop]       {$\twop{j}{\objx}{2}$}
                        node[dot]  (wx) {}

      -- ++(2.1,0)     node[aop]       {$\tryC_j\!\to\!\co_j$}
                        node[dot]  (ci) {}
                        ;                       

     \draw[wait] (rxj0) -- (rxj);
     \draw[hb] (ro2) .. controls +(265:1.5) and +(90:.75) .. 
                (rxj);

     \draw[] (wxi) -- (ro1);
     \draw[] (ro2) -- (join.center);
     \draw[wait] (ro1) -- (ro2);
\end{tikzpicture}

\end{center}
\caption{\label{fig:final-writes-opt} Early release on last write. 
}
\end{figure}

This is illustrated in \rfig{fig:final-writes-opt}. Here, $\trobj_i$ knows
\emph{a priori} that it will write to $\objx$ at most once, since
$\wub{i}{\obj} = 1$. Hence, after the one write to $\objx$, a separate thread
is started which releases $\objx$ by setting $\lv{\obj}$ to $1$. Since $\tr_i$
passes the access condition, this happens almost instantaneously (the figure
shows a wait time merely for the reason of aesthetics).
Once $\objx$ is released in this fashion, $\tr_j$, whose private version for $\obj$ is $2$, can execute its own read and write operations on $\obj$ freely. Nevertheless, $\tr_i$ can continue to execute reads on $\obj$ after releasing $\obj$, and since the value of $\obj$ is read from $\tr_i's$ buffer, $\tr_j$'s operations do not interfere.

\section{Interleaving Comparison}
\label{sec:comparison}
\label{sec:comp}

In this section we compare the interleavings, or \emph{histories}, admitted by
OptSVA to those admitted by its predecessor, the Supremum Versioning Algorithm
(SVA). SVA (with rollback support) is described in detail in
\cite{SW13,SW15-ijpp}. In short, it amounts to the mechanisms described in
\rsec{sec:supremum-versioning}, without the optimizations described in
Sections \ref{sec:read-only-variables}--\ref{sec:early-release-on-last-write}.

\subsection{Preliminaries}

In order to compare the interleavings of the two algorithms, let us first
provide definitions of transactional histories and relevant ancillary concepts
that extend the transactional system model defined in the previous section.

\subsubsection{Traces and Operation Executions}

Given program $\prog$ and a set of processes $\processes$, we denote an
execution of $\prog$ by $\processes$ as $\exec{\prog}{\processes}$.
An execution entails each process $\proc_k \in \processes$ evaluating some
prefix of subprogram $\subprog_k \in \prog$.
The evaluation of each statement $\stmt \in \subprog_k$ by any process is deterministic.
This evaluation produces a (possibly empty) sequence of events (steps)
which we denote $\stmteval{\stmt}$.

Furthermore by $\stmteval{\subprog_k}$ we denote a sequence s.t. given $\stmt_1,
\stmt_2, ..., \stmt_m = \subprog_k$, $\stmteval{\subprog_k} =
\stmteval{\stmt_1} \cdot \stmteval{\stmt_2} \cdot ... \cdot
\stmteval{\stmt_m}$.
By extension, $\exec{\prog}{\processes}$ produces a sequence of events,
which we call a trace $\trace$: $\trace \producedby \exec{\prog}{\processes}$
iff $\forall \proc_k \in \processes, \subprog_k \in \prog, \stmteval{\subprog_k}
\subseteq \trace$.
$\exec{\prog}{\processes}$ is concurrent, 
i.e. while the statements in subprogram $\subprog_k$ are evaluated sequentially by
a single process, the evaluation of statements by different processes can be
arbitrarily interleaved.
Hence, given $\trace \producedby \exec{\prog}{\processes}$ and $\trace'
\producedby \exec{\prog}{\processes}$, it is possible that $\trace \neq
\trace'$.
We call $\exec{\prog}{\processes}$ a \emph{complete} execution if each process
$\proc_k$ in $\processes$ evaluates all of the statements in $\subprog_k$.
Otherwise, we call $\exec{\prog}{\processes}$ a \emph{partial} execution.  By
extension, if $\exec{\prog}{\processes}$ is a complete execution, then $\trace
\producedby \exec{\prog}{\processes}$ is a \emph{complete} trace.

In order to execute some transactional operation $o$ on variable $\obj$ within
transaction $\tr_i$, process $\proc_k$
issues an \emph{invocation event} %
denoted
 $\inv{i}{k}{o}$, and receives a
\emph{response event} %
denoted $\res{i}{k}{\valu}$, where
$\valu$ is the return value of $o$.  %
More specifically, given the operations introduced as part of the transactional model,
if process
$\proc_k$ executes some operation as part of transaction $\tr_i$ it issues an
invocation event of the form $\inv{i}{k}{\init_i}$,
$\inv{i}{k}{o}$ %
for some $\obj$, or $\inv{i}{k}{\tryC_i}$, (or
possibly $\inv{i}{k}{\tryA_i}$) and receives a response of the form
$\res{i}{k}{\valu_i}$, where $\valu_i$ is a value, or 
the constant $\ok_i$, $\co_i$, or $\ab_i$.  
The superscript always denotes which process executes the operation, and the
subscript denotes of which transaction the operation is a part.
Each event is atomic and instantaneous, but the execution of the entire
operation composed of two events is not. 

A pair of these events composed of an invocation of operation $o$ and a response event to $o$ is called a
\emph{complete operation execution} and it is denoted 
$o^k_i \rightarrow \valu$, 
whereas an invocation event $\inv{i}{k}{o}$
without the corresponding response event is called a 
\emph{pending operation execution}. 
We refer to complete and pending 
operation executions as \emph{operation executions}, denoted by $op$.
The transactional model allows the following transactional operation executions
(executed by process $\proc_k$ within transaction $\tr_i$ as):
\begin{enumerate}[a) ]
    \item $\init_i^k \rightarrow \ok_i$,
    \item $\pfrop{i}{k}{\obj}{\val}\,$ or $\,\pfrop{i}{k}{\obj}{\ab_i}$,
    \item $\pfwop{i}{k}{\obj}{\val}{\ok_i}\,$ or $\,\pfwop{i}{k}{\obj}{\val}{\ab_i}$,
    \item $\tryC_i^k \rightarrow \co_i\,$ or $\,\tryC_i^k \rightarrow \ab_i$.
    \item $\tryA_i^k \rightarrow \ab_i$.
\end{enumerate}

Since it is convenient to talk about transactions as independent entities, and
their relations to specific processes is irrelevant, we will henceforth
simplify the notation of invocation and response events to $\inv{i}{}{o}$,
$\res{i}{}{o}$, and of complete executions to $o_i \rightarrow \valu$. Then,
the notation of operation executions becomes
    $\init_i \rightarrow \ok_i$,
    $\frop{i}{\obj}{\val}$, %
    $\fwop{i}{\obj}{\val}{\ok_i}$, %
    $\tryC_i \rightarrow \co_i$, etc. %

Whenever an operation execution refers to a value, but it is irrelevant to the discussion and inconveninent to
specify it,     
we use a placeholder value $\any$ in its place, writing e.g.
$\frop{i}{\obj}{\any}$ or $\fwop{i}{\obj}{\any}{\ok_i}$.

\subsubsection{Histories}

Given a trace $\trace \producedby \exec{\prog}{\processes}$, 
a TM \emph{history} $\hist$ is a subsequence of trace $\trace$
consisting only of executions of transactional operations s.t. for every event $e$, $e \in \hist$ iff $e \in \trace$ and $e$ is
either an invocation or a response event specified by the transactional model.
If $\hist \subset \trace$ we say $\trace$ produces $\hist$.
A \emph{subhistory} of a history $H$ is a subsequence of $H$. 

The sequence of events in a history $\hist_j$ can be denoted as $\hist_j = [
e_1, e_2, ..., e_m ]$. 
For instance, some history $\hist_1$ below is a history of a run of some
program that executes transactions $\tr_1$ and $\tr_2$:
\begin{equation*}
\begin{split}
\hist_1 = [~&
              \inv{1}{}{\init_1}, \res{1}{}{\ok_1},
              \inv{2}{}{\init_2}, \res{2}{}{\ok_2},  \\
            & \inv{1}{}{\twop{1}{\obj}{\val}},
              \inv{2}{}{\trop{2}{\obj}{}}, 
              \res{1}{}{\ok_1}, 
              \res{2}{}{\val},                       \\
            & \inv{1}{}{\tryC_1},\res{1}{}{\co_1},
              \inv{2}{}{\tryC_2},\res{2}{}{\co_2}    ~].
\end{split}
\end{equation*}

Given any history $\hist$, let $\hist|\tr_i$ be the longest 
subhistory of $\hist$ consisting only of invocations and responses executed 
by transaction $\tr_i$.
For example,  $\hist_1|\tr_2$ is defined as:
\begin{equation*}
\begin{split}
\hist_1|\tr_2 = [~ &
                     \inv{2}{}{\init_2}, \res{2}{}{\ok_2}, 
                    \inv{2}{}{\trop{2}{\obj}{}}, 
                    \res{2}{}{\val},             \\&         
                    \inv{2}{}{\tryC_2},\res{2}{}{\co_2}    ~].
\end{split}
\end{equation*}

We say transaction $\tr_i$ \emph{is in} $\hist$, which we denote $\tr_i \in \hist$, if $\hist|\tr_i \neq \varnothing$.

Let $\hist|\obj$ be the longest subhistory of $\hist$ consisting only of
invocations and responses executed on variable $\obj$, but only those that 
form complete operation executions.

Given complete operation execution $op$ that consists of an invocation
event $e'$ and a response event $e''$, we say $op$ \emph{is in} $\hist$ 
($\op \in \hist$) if $e' \in \hist$ and $e'' \in \hist$.
Given a pending operation execution $op$ consisting of an invocation $e'$, 
we say $op$ \emph{is in} $\hist$ ($\op \in \hist$) if $e' \in \hist$ and 
there is no other operation execution $op'$ consisting of an invocation 
event $e'$ and a response event $e''$ s.t. $\op' \in \hist$.      

Given two complete operation executions $\op'$ and $\op''$ in some history 
$\hist$, where $\op'$ contains the response event $\mathit{res}'$ and 
$\op''$ contains the invocation event $\mathit{inv}''$, 
we say $\op'$ \emph{precedes} $\op''$ in $\hist$ if
$\mathit{res}'$ precedes $\mathit{inv}''$ in $\hist$.

A history whose all operation executions are complete is a \emph{complete}
history.

Most of the time it will be convenient to denote any two adjoining events in a
history that represent the invocation and response of a complete execution of
an operation as that operation execution, using the syntax $e \rightarrow e'$. 
Then, an alternative representation of $\hist_1|\tr_2$ is denoted as follows: 
$$\hist_1|\tr_2 = [~
        \init_2 \rightarrow \ok_2,~
        \frop{2}{\obj}{\val},~
        \tryC_2 \rightarrow \co_2~].$$

History $\hist$ is \emph{well-formed} if, for every transaction $\tr_i$ in
$\hist$, $\hist|\tr_i$ is an alternating sequence of invocations and responses
s.t., 
\begin{enumerate}[a) ] 
    \item $\hist|\tr_i$ starts with an invocation $\inv{i}{}{\init_i}$, 
    \item no events in  $\hist|\tr_i$ follow $\res{i}{}{\co_i}$ or $\res{i}{}{\ab_i}$,
    \item no invocation event in $\hist|\tr_i$ follows $\inv{i}{}{\tryC_i}$ or $\inv{i}{}{\tryA_i}$,
    \item for any two transactions $\tr_i$ and $\tr_j$ s.t., $\tr_i$ and
    $\tr_j$ are executed by the same process $\proc_k$, the last event of
    $\hist|\tr_i$ precedes the first event of $\hist|\tr_j$ in $\hist$ or \emph{vice
    versa}.
\end{enumerate}
In the remainder of the paper we assume that all histories are well-formed.

History $\hist$ has \emph{unique writes} if, given transactions $\tr_i$ and
$\tr_j$ (where $i\neq j$ or $i=j$), for any two write operation executions
$\fwop{i}{\obj}{\val'}{\ok_i}$ and $\fwop{j}{\obj}{\val''}{\ok_j}$ it is true
that $\val' \neq \val''$ and neither $\val' = \val_0$ nor $\val'' = \val_0$.

\subsubsection{Accesses}

Given a history $\hist$ and a transaction $\tr_i$ in $\hist$, we say that $\tr_i$
\emph{reads} variable $\obj$ in $\hist$ if there exists an invocation
$\inv{i}{}{\trop{i}{\obj}{}}$ in $\hist|\tr_i$.
By analogy, we say that $\tr_i$ \emph{writes} to $\obj$ in $\hist$ if there
exists an invocation $\inv{i}{}{\twop{i}{\obj}{\val}}$ in $\hist|\tr_i$.
If $\tr_i$ reads $\obj$ or writes to $\obj$ in $\hist$, we say $\tr_i$
\emph{accesses} $\obj$ in $\hist$.
In addition, 
let $\tr_i$'s \emph{read set} be a set that contains every variable $\obj$,
s.t.  $\tr_i$ reads $\obj$. By analogy, $\tr_i$'s \emph{write set} contains
every $\obj$, s.t. $\tr_i$ writes to $\obj$. A transaction's \emph{access set},
denoted $\accesses{i}$, is the union of its read set and its write set.

Given a history $\hist$ and a pair of transactions $\tr_i, \tr_j \in \hist$, we
say $\tr_i$ and $\tr_j$ \emph{conflict} on variable $\obj$ in $\hist$ if
$\tr_i$ and $\tr_j$ are concurrent, both $\tr_i$ and $\tr_j$ access $\obj$, and
one or both of $\tr_i$ and $\tr_j$ write to $\obj$.

Given a history $\hist$ and a pair of transactions $\tr_i,
\tr_j \in \hist$, we say $\tr_i$ \emph{reads from} $\tr_j$ if there is some
variable $\obj$, for which 
there is a complete operation execution $\fwop{j}{\obj}{\val}{\ok_j}$ in $\hist|\tr_j$
and another complete operation execution $\frop{i}{\obj}{\valu}$ in
$\hist|\tr_i$, s.t. $\val = \valu$.

Given any transaction $\tr_i$ in some history $\hist$, any
operation execution on a variable $\obj$ within $\hist|\tr_i$ is either
\emph{local} or \emph{non-local}. Read operation execution
$\frop{i}{\obj}{\val}$ in $\hist|\tr_i$ is local if it is preceded in
$\hist|\tr_i$ by a write operation execution on $\obj$, and it is non-local
otherwise. Write operation execution $\fwop{i}{\obj}{\val}{\ok_i}$ in
$\hist|\tr_i$ is local if it is followed in $\hist|\tr_i$ by a
write operation execution on $\obj$, and non-local otherwise.

\subsubsection{Execution Time}

\def\time#1#2{\tau_{#1}(#2)}
\def\ttime#1{\tau_{#1}}
\def\stime#1#2{\tau^\leftarrow_{#1}(#2)}
\def\wtime#1#2{\tau^\leftrightarrow_{#1}(#2)}
\def\etime#1#2{\tau^\rightarrow_{#1}(#2)}
\def\rtime#1#2#3{\tau^{\,r}_{#1}(#2)}
\def\ctime#1#2#3{\tau^{\,c}_{#1}(#2)}

As program $\prog$ is being evaluated by some TM implementation, by a set of
processes $\processes$, it takes time to evaluate each statement. Hence, each
event $e$ in a trace $\trace \producedby \exec{\prog}{\processes}$ appears at a
specific point in time, which we denote $\time{\trace}{e}$. Since each process
$\proc_k$ executes statements in $\subprog_k$ in sequence, then, given two
events $e_1, e_2$ s.t. $e_1 \prec_\trace e_2$, $\time{\trace}{e_1} < \time{\trace}{e_2}$.
Given a complete operation execution $\op$ consisting of an invocation event
$e_1$ and a response event $e_2$, 
the time at which $\op$ finishes executing is
$\etime{\trace}{\op} = \time{\trace}{e_2}$.
The \emph{execution time} of trace $\trace \producedby \exec{\prog}{
\processes}$, denoted $\ttime{\trace}$, is equal to the largest
execution time for any event in $\trace$.
The \emph{release time} of variable $\obj$ in transaction $\tr_i$ in $\trace$,
denoted $\rtime{\trace}{\tr_i}{\obj}$, is the point in time at which $\tr_i$
updates $\lv{\obj}$.
The \emph{completion time} of variable $\obj$ in transaction $\tr_i$ in
$\trace$, $\ctime{\trace}{\tr_i}{\obj}$, is the point in time at which
$\tr_i$ updates $\ltv{\obj}$.

\subsection{Execution Time Comparison}

In this section, we show that the execution time of OptSVA histories is lower
than than of SVA histories resulting from the execution of the same program by
the same processes.

Let $\svaexec{\prog}{\processes}$ denote a complete execution of program
$\prog$ by processes $\processes$ according to the SVA concurrency control
algorithm, and $\optsvaexec{\prog}{\processes}$, an otherwise identical
execution, but according to OptSVA.
Then, there are traces $\tracesva \producedby \svaexec{\prog}{\processes}$ and
$\traceoptsva \producedby \optsvaexec{\prog}{\processes}$, and histories
$\histsva = \tohist\tracesva$ and $\histoptsva = \tohist\traceoptsva$.
The histories contain corresponding transactions: if $\tr_i \in \histsva$ then
$\tr_i \in \histoptsva$ and \emph{vice versa}. Let $\transactions$ be the set
of all transactions in $\histoptsva$ and $\histsva$.

For the purpose of the comparison we assume that the events in histories are
instantaneous.  We also do not account for the time it takes to execute
concurrency control code. Finally, we assume that apart from the details of
the concurrency control, the execution proceeds the same, regardless of whether
it is SVA or OptSVA.

\begin{lemma}[Early Release] \label{lemma:comp:early-release}
    For any $\tr_i \in \transactions$ and $\obj \in \accesses{i}$,
    $\ctime{\histoptsva}{\tr_i}{\obj} \leq \ctime{\histsva}{\tr_i}{\obj}$.
\end{lemma}

\begin{proof}
    An SVA transaction releases $\obj$ by updating $\lv{\obj}$ on commit, on
    abort, and during the
    last operation execution on $\obj$.
    An OptSVA transactions does so on commit, on abort, during the last write
    operation execution on $\obj$, and after buffering a read-only variable.

    \begin{enumerate}[a) ]
        \item If $\obj$ is a read-only variable an SVA transaction releases
            $\obj$ no sooner than the last operation execution on $\obj$, so
            given any read operation execution $\frop{i}{\obj}{\any} \in
            \histsva|\tr_i$:
            \[ \rtime{\histsva}{\tr_i}{\obj} \geq
               \etime{\histsva}{\frop{i}{\obj}{\any}}. \]
            On the other hand, OptSVA releases $\obj$ as soon as possible.
            That is during $\init_i\to\ok_i$ at the earliest, and no later than any
            $\frop{i}{\obj}{v} \in \histoptsva|\tr_i$ at the latest. Thus:
            \[ \rtime{\histoptsva}{\tr_i}{\obj} \leq
               \etime{\histoptsva}{\frop{i}{\obj}{\any}}. \]
            In that case, all things being equal: 
            \[ \rtime{\histoptsva}{\tr_i}{\obj} \leq
               \rtime{\histsva}{\tr_i}{\obj}. \]

        \item Alternatively, if the last operation execution in
            $\histsva|\tr_i|\obj$  is $\frop{i}{\obj}{\any}$, then an SVA
            transaction releases $\obj$ no sooner than $\frop{i}{\obj}{\any}$, so:
            \[ \rtime{\histsva}{\tr_i}{\obj} \geq
               \etime{\histsva}{\frop{i}{\obj}{\any}}. \]           
            On the other hand, if last operation execution in 
            $\histoptsva|\tr_i|\obj$ is $\frop{i}{\obj}{\any}$, then an OptSVA
            transaction releases $\obj$ no sooner than any
            $\fwop{i}{\obj}{\any}{\ok_i}$ in  $\histoptsva|\tr_i$. 
            \[ \rtime{\histoptsva}{\tr_i}{\obj} \geq
               \etime{\histoptsva}{\fwop{i}{\obj}{\any}{\ok_i}}. \]
            Since $\etime{\histoptsva}{\fwop{i}{\obj}{\any}{\ok_i}} <
            \etime{\histoptsva}{\frop{i}{\obj}{\any}}$, then, all things being
            equal:
            \[ \rtime{\histoptsva}{\tr_i}{\obj} \leq
               \rtime{\histsva}{\tr_i}{\obj}. \]

        \item Otherwise, the last operation execution in $\tr_i$
            is $\fwop{i}{\obj}{\any}{\ok_i}$, so both SVA and OptSVA
            transactions will release $\obj$ no sooner than
            $\fwop{i}{\obj}{\any}{\ok_i}$, so, all things being equal:
            \[ \rtime{\histoptsva}{\tr_i}{\obj} =
               \rtime{\histsva}{\tr_i}{\obj}. \]
    \end{enumerate}
\end{proof}

\begin{lemma}[Early Completion] \label{lemma:comp:early-completion}
    For any $\tr_i \in \transactions$ and $\obj \in \accesses{i}$,
    $\rtime{\histoptsva}{\tr_i}{\obj} \leq \rtime{\histsva}{\tr_i}{\obj}$.
\end{lemma}

\begin{proof}
    An SVA transaction updates $\ltv{\obj}$ on commit, or on
    abort, so:
    \[ \ctime{\histsva}{\tr_i}{\obj} = \time{\histsva}{\res{i}{}{\co_i}}
       ~\text{or}~
       \ctime{\histsva}{\tr_i}{\obj} = \time{\histsva}{\res{i}{}{\ab_i}}. \]
    An OptSVA transactions updates $\ltv{\obj}$ on commit, on abort, or after releasing a
    read-only variable. The latter-most potentially precedes a commit, so:
    \[ \ctime{\histoptsva}{\tr_i}{\obj} \leq \time{\histsva}{\res{i}{}{\co_i}}
       ~\text{or}~
       \ctime{\histoptsva}{\tr_i}{\obj} \leq \time{\histsva}{\res{i}{}{\ab_i}}. \]
    Thus, all things being equal:
    \[ \ctime{\histoptsva}{\tr_i}{\obj} \leq \ctime{\histsva}{\tr_i}{\obj}. \]
\end{proof}

\begin{lemma}[Early Operation Execution]
    For any $\tr_i \in \transactions$, and any operation execution $\op$ in
    $\histoptsva|\tr_i$ and $\histsva|\tr_i$,    
    $\etime{\histoptsva}{\op} \leq \etime{\histsva}{\op}$.
\end{lemma}

\begin{proof}
    The case for $\pv{\obj}{i} = 1$ is trivial. If $\pv{\obj}{i} > 1$ then
    there exists $\tr_j \in\transactions$ s.t. $\pv{\obj}{j} + 1 = \pv{\obj}{i}$.

    \begin{enumerate}[i) ]
        \item If $\op$ is a read operation execution, $\op$ can return a value
            and be a non-local read operation execution, or a local one, or an
            the operation can return $\ab_i$.
            \begin{enumerate}[a) ]
                \item If $\op = \frop{i}{\obj}{\val}$ is a non-local read
                    operation execution in both SVA and OptSVA the operation
                    execution will not finish before the access condition is
                    satisfied, so:
                    \[ \etime{\histsva}{\op} \geq \rtime{\histsva}{\tr_j}{\obj}, \]
                    \[ \etime{\histoptsva}{\op} \geq \rtime{\histoptsva}{\tr_j}{\obj}. \]
                    Then, from \rlemma{lemma:comp:early-release}:
                    \[  \rtime{\histoptsva}{\tr_j}{\obj} \leq \rtime{\histsva}{\tr_j}{\obj}, \]
                    So, all things being equal:
                    \[  \etime{\histoptsva}{\op} \leq \etime{\histsva}{\op}. \]
                \item If $\op = \frop{i}{\obj}{\val}$  is a local read
                    operation execution, then, by definition, local reads
                    follow a write operation execution, so $\exists \op_w =
                    \fwop{i}{\obj}{\any} \in \histsva|\tr_i$ s.t. $\op_w
                    \prec_\histsva \op$ and $\op_w = \fwop{\obj}{\val} \in
                    \histsva|\tr_i$ s.t. $\op_w \prec_\histsva \op$. In that
                    case:
                    \[ \etime{\histsva}{\op} \geq \etime{\histsva}{\op_w}, \]
                    \[ \etime{\histoptsva}{\op} \geq \etime{\histoptsva}{\op_w}. \]
                    Then, from ii:
                    \[ \etime{\histoptsva}{\op_w} \leq \etime{\histsva}{\op_w}. \]
                    Hence, all other things being equal:
                    \[ \etime{\histoptsva}{\op} \leq \etime{\histsva}{\op}. \]
                \item If $\op  = \frop{i}{\obj}{\ab_i}$, then  operation
                    execution in both SVA and OptSVA the operation execution
                    waits until $\ltv{\objy} = \pv{\objy}{i} - 1$ is true for
                    all $\objy \in \accesses{i}$. This means that each
                    transaction $\tr_k$ s.t. $\pv{\objy}{k} + 1 =
                    \pv{\objy}{k}$ must update $\pv{\objy}{i}$ to its private
                    version. Hence:
                    \[ \etime{\histsva}{\op} \geq 
                        \max_{
                            \;
                            \begin{subarray}{l}
                            \forall\tr_k\in\transactions,\;\;
                            \forall\objy\in\accesses{k},\\
                            \text{s.t.}\;\pv{\objy}{k} + 1= \pv{\objy}{i}
                            \end{subarray}
                        }\nolimits 
                        \ctime{\histsva}{\tr_k}{\objy}, \]
                    \[ \etime{\histoptsva}{\op} \geq 
                        \max_{
                            \;
                            \begin{subarray}{l}
                            \forall\tr_k\in\transactions,\;\;
                            \forall\objy\in\accesses{k},\\
                            \text{s.t.}\;\pv{\objy}{k} + 1= \pv{\objy}{i}
                            \end{subarray}
                        }\nolimits 
                        \ctime{\histoptsva}{\tr_k}{\objy}. \]
                    From \rlemma{lemma:comp:early-completion}:
                    \[ \ctime{\histoptsva}{\tr_j}{\obj} \leq \ctime{\histsva}{\tr_j}{\obj}, \]
                    So, all things being equal:
                    \[ \etime{\histoptsva}{\op} \leq \etime{\histsva}{\op}. \]
            \end{enumerate}
        \item If $\op$ is a write operation execution, $\op$ can return $\ok_i$
            and either be a preceded by a non-local read operation execution,
            or only by write and non-local read operation executions. Otherwise
            the write operation execution can return $\ab_i$.
            \begin{enumerate}[a) ]
                \item If $\op = \fwop{i}{\obj}{\any}{\ok_i}$ is preceded by
                    some non-local $\op_r = \frop{i}{\obj}{\any}$, then
                    in both SVA and OptSVA:
                    \[ \etime{\histsva}{\op} \geq \etime{\histsva}{\op_r}, \]
                    \[ \etime{\histoptsva}{\op} \geq \etime{\histoptsva}{\op_r}. \]
                    From i point a:
                    \[ \etime{\histoptsva}{\op_r} \leq \etime{\histsva}{\op_r}. \]
                    Thus, all things being equal:
                    \[  \etime{\histoptsva}{\op} \leq \etime{\histsva}{\op}. \]
                \item If $\op = \fwop{i}{\obj}{\any}{\ok_i}$ is not preceded by
                    non-local read operation executions, then there is such
                    $\op_w = \fwop{i}{\obj}{\any}{\ok_i}$ (possibly $\op_w = \op$) such that
                    $\op_w$ is the initial operation in $\histsva|\tr_i|\obj$
                    and $\histoptsva|\tr_i|\obj$.
                    In addition, $\op_w$ is necessarily preceded by $\op_s = \init_i \to \ok_i$, so:
                    \[ \etime{\histsva}{\op_w} \geq \etime{\histsva}{\op_s}, \]
                    \[ \etime{\histoptsva}{\op_w} \geq \etime{\histoptsva}{\op_s}. \]
                    In SVA an initial write waits for the access condition, so:
                    \[ \etime{\histsva}{\op} \geq
                       \max(\rtime{\histsva}{\tr_j}{\obj},
                       \etime{\histsva}{\op_s}). \]
                    In OptSVA writes do not wait for the access condition at all, so: 
                    \[ \etime{\histoptsva}{\op} \geq \etime{\histoptsva}{\op_s}, 
                    \,\text{regardless of}\,\rtime{\histsva}{\tr_j}{\obj}. \]
                    From v:
                    \[ \etime{\histoptsva}{\op_s} \leq \etime{\histsva}{\op_s}. \]
                    Hence, all other things being equal:
                    \[ \etime{\histoptsva}{\op_s} = \etime{\histsva}{\op_s}. \]
                    Then, since either $\op_w = \op$ or $\op_w$ precedes $\op$:
                    \[ \etime{\histoptsva}{\op} \leq \etime{\histsva}{\op}. \]
                \item If $\op = \fwop{i}{\obj}{\any}{\ab_i}$, then, by analogy to ii point c:
                    \[ \etime{\histoptsva}{\op} \leq \etime{\histsva}{\op}. \]
            \end{enumerate}
        \item If $\op = \tryC_i \to \any$, then in both SVA and OptSVA
            transactions wait until $\ltv{\objy} = \pv{\objy}{i} - 1$ is true
            for all $\objy \in \accesses{i}$ before returning from $\op$. 
            This means that each transaction $\tr_k$ s.t. $\pv{\objy}{k} + 1 =
            \pv{\objy}{k}$ must update $\pv{\objy}{i}$ to its private version.
            Hence:
            \[ \etime{\histsva}{\op} \geq \max_{ 
                    \; 
                    \begin{subarray}{l}
                    \forall\tr_k\in\transactions,\;\;
                    \forall\objy\in\accesses{k},\\
                    \text{s.t.}\;\pv{\objy}{k} + 1=\pv{\objy}{i} 
                    \end{subarray} 
               }\nolimits 
               \ctime{\histsva}{\tr_k}{\objy}, \] 
            \[ \etime{\histoptsva}{\op} \geq \max_{ 
                    \; 
                    \begin{subarray}{l}
                    \forall\tr_k\in\transactions,\;\; 
                    \forall\objy\in\accesses{k},\\
                    \text{s.t.}\;\pv{\objy}{k} + 1= \pv{\objy}{i} 
                    \end{subarray} 
                }\nolimits
                \ctime{\histoptsva}{\tr_k}{\objy}. \]
            From \rlemma{lemma:comp:early-completion}: 
            \[
                \ctime{\histoptsva}{\tr_j}{\obj} \leq
                \ctime{\histsva}{\tr_j}{\obj}, \] 
            So, all things being equal:
            \[ \etime{\histoptsva}{\op} \leq \etime{\histsva}{\op}. \]
        \item If $\op = \tryA_i \to \ab_i$, then, by analogy to iv:
            \[ \etime{\histoptsva}{\op} \leq \etime{\histsva}{\op}. \]
        \item If $\op = \init_i \to \ok_i$, then trivially, 
            \[ \etime{\histoptsva}{\op} = \etime{\histsva}{\op}. \]      
    \end{enumerate}
\end{proof}

\begin{corollary}[Lower Execution Time]
    $\ttime{\histoptsva} \leq \ttime{\histsva}$.
\end{corollary}

Thus, the execution time of OptSVA is no worse than SVA. Intuitively, OptSVA is
likely perform better in almost all cases though, and especially, if high
contention causes many transactions to wait to access the same object---then,
the expedited release times and delayed synchronization come into play.

\subsection{Practical Comparison {\tt [Proposition for Consideration]}}

Given that the theoretical considerations above ignore the complexity of the
concurrency control algorithm itself, it could be argued that the cost of
executing individual operations and delegating execution to separate threads
are heavy enough to wipe out any theoretical scheduling advantage in practice.
Thus, in this section, we perform a practical comparison of the two algorithms,
that bears out the conclusions from the previous section, by showing comparing
the performance of SVA to OptSVA given a variety of workloads.

For evaluation we used EigenBench \cite{SOJB+10}, a flexible,
powerful, and lightweight benchmark that can be used for comprehensive
evaluation of mutlicore TM systems by simulating a variety of transactional
application characteristics. It generates a traffic of client transactions,
which access objects at random (with a specified degree of locality) according
to a predefined ratio of reads to writes from three different array types: the
hot array contains variables where transactions can conflict, the mild array
contains variables accessed transactionally but without the possibility of
conflict, and the cold array contains non-transactional variables. 

The experiment was run 
The benchmark was executed on a 10-node cluster with two quad-core Intel Xeon
L3260 processors at 2.67 GHz and 4 GB of RAM per node, running OpenSUSE 13.1
(kernel 3.11.10, x86\_64 architecture), and connected with a 1Gb network. The
implementations of SVA and OptSVA run on the 
64-bit Java HotSpot(TM) Java Virtual Machine version 1.8 (build 1.8.0\_25-b17),
as does the benchmark. 

The benchmark executes 80 concurrent threads (one per processor core), each of
which executes 10 consecutive transactions. The transactions have three
parameters: length, read-to-write ratio, and contention. Long transactions
execute 10 operations on shared variables each, while short ones execute 5
each. These operations either have an $5\!\div\!1$ or $1\!\div\!5$
read-to-write ratio.  The high contention scenarios provide a total of 20
shared variables to the transaction, while the low contention ones provide
80. We only use hot arrays for the purpose of this presentation, since
only they impact contention.  Locality of operations is at 50\% and is based on
a 5-variable--long history.
We measure total execution time of the entire workload, and
throughput---operations executed per second.

\newcolumntype{s}{>{\hsize=.5\hsize\centering\arraybackslash}X}

\begin{figure}
    {\footnotesize
    \begin{tabularx}{\linewidth}{l s s s s c}\toprule        
        \multirow{2}{*}{Parameters} 
        & \multicolumn{2}{c}{Execution Time [s]} 
        & \multicolumn{2}{c}{Throughput [ops/s]} 
        & \multirow{2}{*}{Gain [\%]} \\ 
        \cmidrule{2-3}\cmidrule{4-5}
        & {\scriptsize SVA} 
        & {\scriptsize OptSVA} 
        & {\scriptsize SVA} 
        & {\scriptsize OptSVA} 
        & \\
        \midrule
         {\scriptsize Short, RW $5\!\div\!1$, high cont.} & 492.2 & 257.2 & 8.1  & 15.5 & 47.7 \\
         {\scriptsize Short, RW $1\!\div\!5$, high cont.} & 486.1 & 266.3 & 8.2	 & 15.0 & 45.2 \\
          {\scriptsize Long, RW $5\!\div\!1$, high cont.} & 994.9 & 576.0 & 8.0  & 13.9 & 42.1 \\
          {\scriptsize Long, RW $1\!\div\!5$, high cont.} & 979.7 & 640.1 & 8.2  & 12.5 & 34.7 \\
          {\scriptsize Short, RW $5\!\div\!1$, low cont.} & 206.5 &	169.9 & 19.3 & 23.5 & 17.7 \\
          {\scriptsize Short, RW $1\!\div\!5$, low cont.} & 210.5 & 168.8 & 19.0 & 23.6 & 19.8 \\
           {\scriptsize Long, RW $5\!\div\!1$, low cont.} & 439.9 & 308.5 & 18.2 & 25.9 & 29.9 \\
           {\scriptsize Long, RW $1\!\div\!5$, low cont.} & 442.7 & 297.1 & 18.1 & 26.9 & 32.9 \\
        \bottomrule
    \end{tabularx}
    }
    \caption{\label{table:comp-results} Experimental comparison between SVA and
    OptSVA.}
\end{figure}

The results shown in \rfig{table:comp-results} confirm the theoretical
comparison, showing that the execution time of OptSVA is lower than that of
SVA, and that, in practice this is the typical result. The advantage of OptSVA
over SVA is affected by contention, since OptSVA optimizations have more impact
when the rate of potential conflicts is higher. Thus, in high contention the
execution time of OptSVA is 34.7--40.7\% lower than that of SVA, whereas in low
contention the difference drops to only between 17.7 and 32.9\%.
Note that the advantage occurs regardless of the fact that the threads
executing asynchronous computations for OptSVA transactions have to share
processors with transaction threads. This does not have a large impact, since
those threads are mostly waiting at access conditions.
In addition, a higher incidence of reads is also better optimized by OptSVA
(due to read-only variables), but the difference is not very pronounced (no
more than 10\%), since the optimization has more impact in long transactions,
but read-only variables are increasingly less likely to occur in EigenBench as
transactions get longer.
On the other hand SVA treats reads and writes the same, so it performs
consistently regardless of the read-to-write ratio.
The abort rate is 0 in all cases.

\section{Correctness}
\label{sec:correctness}
\label{sec:proof}

In this section we show that OptSVA satisfies last-use opacity--the same safety
property as SVA, meaning that the parallelism optimization does not sacrifice
or otherwise relax correctness.

Last-use opacity \cite{SW14-disc,SW15-arxiv} is a TM correctness property that
provides the same guarantees as opacity, with the exception that it allows
reading from live transactions after they performed their \emph{closing
write}---the last write write in that transaction in any possible execution of
that program. (For convenience we repeat the definition of the property after
the original paper in the appendix.)

Given that OptSVA divorces the operations performed on shared variables within
the code of the transaction from the actual accesses to memory that are
executed, and since last-use opacity is defined on operations on shared variables,
showing correctness is not straightforward.
Instead, we use a different method, where we show that the behavior of view and
update events in traces generated by OptSVA satisfy a set of specific
event-related properties, which we refer to in aggregate as
\emph{trace harmony}. 

First, we present the preliminary material that defines how operations on
memory are represented within traces. Then, we give the definitions making up
trace harmony are presented below in this section.
We show in \rappx{sec:harmony-to-last-use-opacity} that any  harmonious trace
implies a last-use opaque history in general.
Finally, we demonstrate that OptSVA traces are harmonious in
\rsec{sec:optsva-harmony}, and so, that OptSVA is last-use opaque.

\subsection{Events}
\label{sec:correctness:preliminaries}
\label{sec:events}

Events are the results of transactions directly interacting with the memory
representing shared variables. When during the execution of some program, some
transaction accesses a variable's state (either viewing it or updating it), it
issues an update event that is logged in the trace resulting from the
execution.

A \emph{view event} $\get{i}{\obj}{\val}$ is any event that represents some
transaction $\tr_i$ accessing the state of variable $\obj$ (i.e. reading the
memory location where the value of $\obj$ is stored) and retrieving the value
of $\val$.
An \emph{update event} $\set{i}{\obj}{\val}$ is any event that represents a
modification of the state of variable $\obj$ by transaction $\tr_i$, setting it
to the value of $\val$.

Some operations can abort the transaction, rather than doing what they are
intended to do. For instance a write operation may fail with an abort rather
than setting a new value of some variable. In such cases the transaction will
execute specific code that is meant to clean up after the transaction and
revert any variables the transaction modified to a previous (consistent) state.
We will refer to this code as the recovery procedure. Any update events
executed as part of a recovery procedure are called \emph{recovery} (update)
events. In contrast, all update events that are not recovery events are called
\emph{routine} (update) events. For distinction, we denote a routine update
$\rset{i}{\obj}{\val}$ and a recovery update $\aset{i}{\obj}{\val}$.

Given a view event $\get{i}{\obj}{\val}$ (for some $\tr_i$), $\val$ is
specified by the most recent preceding update event on $\obj$ in a given trace.
I.e., if the most recent preceding update event on $\obj$ is some
$\set{j}{\obj}{\val'}$ (for some $\tr_j$), then $\val = \val'$.
Note, that this distinction does not depend on how the events appear in the
trace, but is intrinsic to the code that executes them.
Event $e = \set{i}{\obj}{\val}$ is the \emph{ultimate update} event on $\obj$
in $\trace$ iff there is no $e' = \set{j}{\obj}{\val'}$ s.t. $e \prec_\trace
e'$.
Event $e = \set{i}{\obj}{\val}$ is the \emph{ultimate routine update} event on
$\obj$ in $\trace$ iff $e$ is routine and there is no $e' =
\set{j}{\obj}{\val'}$ s.t. $e \prec_\trace e'$ and $e'$ is routine.

Given a view event $e_v = \get{i}{\obj}{\any}$ in some $\tr_i$ and an update
event $e = \set{j}{\obj}{\any}$  in some $\tr_j$, $e$ prefaces $e_v$ in trace
$\trace$, denoted $e \pref_\trace e_v$ iff $e \prec_\trace e_v$ and there is no
update event $e' = \set{k}{\obj}{\any}$  in any $\tr_k$ s.t. $e \prec_\trace e'
\prec_\trace e_v$.
Given a read operation execution $\op_r \in \trace$ s.t., $\op_r =
\frop{i}{\obj}{\val}$ and $\op_r$ that consists of an invocation event $e_i$ and a
response event $e_r$, and a view event $e_v = \get{i}{\obj}{\val'}$, $\op_r$
\emph{depends on} $e_v$ (denoted $\op_r \dependson e_v$) iff $\val' = \val$ and
$e_v \pref_\trace e_r$.
Given a write operation execution $\op_w \in \trace$ s.t., $\op_w =
\fwop{i}{\obj}{\val}{\ok_i}$ and $\op_w$ consists of an invocation event $e_i$
and a response event $e_r$, and an update event $e_u = \get{i}{\obj}{\val'}$,
$\op_w$ \emph{instigates} $e_u$ (denoted $op_w \instigates e_u$) iff $\val' =
\val$ and $e_i \pref_\trace e_u$.  

Transaction $\tr_i$ views transaction $\tr_j$ ($\tr_i \views \tr_j$) if
$\exists e_u, e_v \in \trace$ s.t. $e_v = \get{i}{\obj}{\val}$ and $e_u =
\rset{j}{\obj}{\val}$ and $e_u \pref_\trace e_v$.
Transaction $\tr_i$ virtually views transaction $\tr_j$ ($\tr_i \vviews \tr_j$)
if $\exists e_u, e_v \in \trace$ s.t. $e_v = \get{i}{\obj}{\val}$ and $e_u =
\rset{j}{\obj}{\val}$ and $e_u \prec_\trace e_v$.

Event access set $\eset{i}$ for some transaction $\tr_i \in \trace$ is such a
set of variables such that $\obj \in \eset{i} \iff \exists e \in \trace|\tr_i$
s.t. $e = \rset{i}{\obj}{\val}$ or $e = \get{i}{\obj}{\val}$.

Given $\tr_i \in \trace$, s.t. $e_v = \get{i}{\obj}{\val} \in \trace|\tr_i$ and
$e_v$ is initial in $\trace|\tr_i$, let $\abset{\trace}{\tr_i}{\obj}$ be such
longest sequence of transactions that:
\begin{inparaenum}[a) ]
\item if $\exists \tr_j \in \trace$ s.t. $e_a = \aset{j}{\obj}{\val}
    \trace|\tr_j$ and $e_a\pref_\trace e_v$ then $\abset{\trace}{\tr_i}{\obj} =
    \abset{\trace}{\tr_i}{\obj} \cdot \tr_i$, otherwise
\item $\abset{\trace}{\tr_i}{\obj} = \varnothing \cdot \tr_i$.
\end{inparaenum}

Let a view chain $\vchain{\trace}{i}{j}$ be a sequence of transactions s.t.
$\tr_i$ is the first element, and $\tr_j$ is the last element, and for each
pair of consecutive transactions $\tr_k, \tr_l$, it is true that $\tr_l \vviews
\tr_k$.
Let $\hist|\vchain{\trace}{i}{j}$ be the longest subsequence of $\trace$ s.t.
$e \in \hist|\vchain{\trace}{i}{j}$ iff $e \in \trace|\tr_i$ and $\tr_i \in
\vchain{\trace}{i}{j}$.

\subsection{Definitions}
\label{sec:harmony-definitions}
\label{sec:harmony}

Since OptSVA limits events within a transaction to at most a single routine
update event, at most a single single recovery update event, and at most a
single view event per variable, we limit the method presented below to such a
case. This is represented by the definition of \emph{minimalism} below. 
(However, the method can be extended to allow multiple routine update events and
multiple view events per transaction.)

\begin{definition} [Minimalism] \label{def:minimalism} 
Given transaction $\tr_i \in \trace$, for each $\obj$, $\trace|\tr_i$ contains:
\begin{enumerate}[a)]
    \item either none or one view event $\get{i}{\obj}{\any}$,
    \item either none or one routine update event $\rset{i}{\obj}{\any}$,
    \item either none or one recovery update event $\aset{i}{\obj}{\any}$.
\end{enumerate}
\end{definition}

\emph{Trace isolation} stipulates, 
that once a transaction starts accessing the memory of some variable, it has
exclusive access to it until it is done performing routine updates and view
events on it. Hence a transaction is not interfered with by other transaction
when it is performing memory accesses, unless an abort is required.
Furthermore, if one transaction accesses the memory of one variable before
another transaction, then that other transaction cannot access any other
variable before the first transaction does.

\begin{definition} [Trace Isolation] \label{def:trace-isolation}
Trace $\trace$ is isolated, iff given any two transactions $\tr_i$
and $\tr_j$ in $\trace$ 
for every $\obj \in \eset{i} \cap \eset{j}$, it is true that
given any event $e_i$ s.t. $e_i = \get{i}{\obj}{\val} \in \trace|\tr_i$ or
    $e_i = \rset{i}{\obj}{\val'} \in \trace|\tr_j$, and any event $e_j$ s.t.
    $e_j = \get{j}{\obj}{\val'} \in \trace|\tr_j$ or a routine update event
    $e_j = \rset{j}{\obj}{\val'} \in \trace|\tr_j$,
$e_i \prec_\trace e_j$.
\end{definition}

\emph{Isolation order} imposes an order on transactions in a trace that respects the
order of executing update and view events on variables.
Given an isolated trace, there exist the following orders:

\def\isoorder{\dot\prec}
\def\nisoorder{\dot\nprec}
\def\disoorder{\ddot\prec}
\def\iso#1#2{{\isoorder}^{\obj}_{\trace}}
\begin{definition} [Variable Isolation Order] \label{def:variable-isolation-order}
Two transactions $\tr_i$ and $\tr_j$ are isolation-ordered in trace $\trace$
with respect to $\obj$, which we denote $\tr_i \iso{\obj}{\trace} \tr_j$,
if given any event $e_i$ s.t. $e_i = \get{i}{\obj}{\val} \in \trace|\tr_i$ or
    $e_i = \rset{i}{\obj}{\val'} \in \trace|\tr_j$, and any event $e_j$ s.t.
    $e_j = \get{j}{\obj}{\val'} \in \trace|\tr_j$ or a routine update event
    $e_j = \rset{j}{\obj}{\val'} \in \trace|\tr_j$, and $e_i \prec_\trace e_j$. 
\end{definition}

\begin{definition} [Direct Isolation Order] \label{def:direct-isolation-order}
Two transactions $\tr_i$ and $\tr_j$ are directly isolation-ordered 
$\tr_i \disoorder_\trace \tr_j$ 
if for every $\obj \in \eset{i} \cap \eset{j}$, $\tr_i \iso{\obj}{\trace}
\tr_j$.
\end{definition}

\begin{definition} [Isolation Order] \label{def:isolation-order}
Two transactions $\tr_i$ and $\tr_j$ are isolation-ordered $\tr_i
\isoorder_\trace \tr_j$ there exists a sequence of transactions $\epsilon =
\tr_i \cdot ... \cdot \tr_j$, where for every pair of consecutive transactions
$\tr_n, \tr_m \in \epsilon$, $\tr_n \disoorder_\trace \tr_m$.
\end{definition}

Note that if $\tr_i \prec_\trace \tr_j$ and $\obj \in \eset{i} \cap \eset{j}$,
then $\tr_i \isoorder_\trace \tr_j$, so the isolation order preserves real-time
order.

\emph{Consonance} describes when a particular event or operation involve a
value that can be considered correct, which is determined by other events or
operations that either precede or follow the one in question.

Specifically, a view event is consonant if it retrieves
the value that was written there by a preceding
event, or the initial value, if no events preceded. A consonant read operation
must then return a value that was retrieved by a view event beforehand. On the
other hand, a routine update event must be caused by some write operation.
Whereas a consonant recovery update event is one that cleans up after a routine
update and returns the state of a variable to a value that was retrieved by a
view event that view the unmodified state of the variable in question.

\begin{definition}[View Consonance] \label{def:view-consonance}
Given some $\tr_i \in \trace$, a 
view event $e_v =
\get{i}{\obj}{\val}$ is consonant in $\trace$ iff either:
\begin{enumerate}[a)]
\item $\val = 0$ and $\nexists e_u \in \trace$ s.t. $e_u = \set{j}{\obj}{\val'}$
for any $\tr_j$, and $e_u \prec_\trace e_r$,
\item $\val \neq 0$ and $\exists e_u \in \trace$ s.t. $e_u =
\rset{j}{\obj}{\val}$ for some $\tr_j$, $i\neq j$, $e_u \pref_\trace e_r$, 
and $e_u$ is the ultimate routine update on $\obj$ in
$\trace|\tr_j$, or
\item $\exists e_u \in \trace$ s.t. $e_u = \aset{j}{\obj}{\val}$ for some
$tr_j$, $i\neq j$, $e_u \pref_\trace e_r$.
\end{enumerate}
\end{definition}

\begin{definition}[Routine Update Consonance] \label{def:routing-update-consonance}
Given some $\tr_i \in \trace$, a routine update event $e_u =
\rset{i}{\obj}{\val}$ is consonant in $\trace$ iff 
$e_u$ is instigated in $\trace$ by a consonant write operation
execution.
\end{definition}

\begin{definition} [Recovery Update Consonance] \label{def:recovery-update-consonance}
Given some $\tr_i \in \trace$, event $e_s = \aset{i}{\obj}{\val}$ is \emph{consonant} in
$\trace$ iff: %
\begin{enumerate}[a) ]
    \item $e_a$ is \emph{conservative} in $\trace$, i.e.  there exists a
        consonant non-local view event $e_v$ in $\trace|\tr_i$ that is initial
        in $\trace|\tr_i$,
    \item $e_a$ is \emph{needed} in $\trace$, i.e. $\exists e_u =
        \rset{i}{\obj}{\val'} \in \trace|\tr_i$ s.t. $e_u \prec_{\trace|\tr_i}
        e_a$,
    \item $e_a$ is \emph{dooming} in $\trace$, i.e. $\nexists r \in \trace$
        s.t. $r = \res{i}{}{\co_i}$, $e_a \prec_{\trace|\tr_i} r$,
    \item $e_a$ is \emph{ending} in $\trace$, i.e. $\nexists e \in \trace$ s.t.
        $e = \get{i}{\obj}{\val'}$ or $e = \set{i}{\obj}{\val'}$, $e_a
        \prec_{\trace|\tr_i} e$,
    \item $e_a$ is \emph{clean} in $\trace$, i.e.  given view $e_v$ that
        justifies that $e_s$ is is conservative, there is no event $e_a' =
        \aset{j}{\obj}{\val'}$ in any $\tr_j$ s.t.  $\tr_j \iso{\obj}{\trace}
        \tr_i$ and $e_v \prec_\trace e_a' \prec_\trace e_a$.
\end{enumerate}    
\end{definition}

\begin{definition} [Non-local Read Consonance] \label{def:non-local-read-consonance}
A non-local read operation execution is consonant in trace $\trace$ iff it
depends in $\trace$ on a consonant non-local view event.
\end{definition}

\begin{definition} [Local Read Consonance] \label{def:local-read-consonance}
Given some $\tr_i \in \trace$, a local read operation execution $\op_r =
\frop{i}{\obj}{\val}$ is consonant in trace $\trace$ iff there exists $\op_w =
\fwop{i}{\obj}{\val}{\ok_i} \in \trace|\tr_i$ s.t. $\op_w \pref_{\trace|\tr_i} \op_r$, and
$\op_w$ is consonant.
\end{definition}

\begin{definition} [Write Consonance] \label{def:write-consonance}
A write operation execution $\fwop{i}{\obj}{\val}{\ok_i}$ in some $\tr_i$ is
consonant in trace $\trace$ iff $\val \neq 0$ and $\val$ is within the domain
of $\obj$.
\end{definition}

\begin{definition} [Trace Consonance] \label{def:trace-consonance}
Trace $\trace$ is consonant iff all operation executions, update events, and
view events in trace $\trace$ are consonant.
\end{definition}

\emph{Obbligato} ensures that update events required by write operations happen
on time, so that the values written to variables by operation executions are
actually set in memory by the time the transaction relinquishes control of each
variable.
This means that a routine update event is required after a write operation by
the time a transaction commits (\emph{committed write obbligato}), one is
required after a closing write operation, before any other transaction attempts
to access that variable (\emph{closing write obbligato}), and one is required
if a non-aborted transaction executed write operations and another transaction
accesses the variables in question (\emph{view write obbligato}).

\begin{definition} [Committed Write Obbligato] 
    \label{def:write-actualization}
    \label{def:commit-write-obbligato}
    Given $\tr_i \in \trace$, if $\exists \op_w \in \trace|\tr_i$ s.t. $\op_w =
    \fwop{i}{\obj}{\val}{\ok_i}$, $\op_w$ is non-local, and $\exists r \in
    \trace|\tr_i$ s.t. $r = \res{i}{}{\co_i} \in \trace|\tr_i$, then $\op_w$ is
    \emph{in obbligato} iff $\exists e_s \in \trace|\tr_i$ s.t. $e_s =
    \rset{i}{x}{v}$ and $\op_w \instigates e_s$ and $e_s \prec_{\trace|\tr_i}
    r$.
\end{definition}

\begin{definition} [Closing Write Obbligato] 
    \label{def:exclusive-write-actualization}
    \label{def:closing-write-obbligato}
    Given $\tr_i \in \trace$, if $\exists \op_w \in \trace|\tr_i$
    if $\exists \tr_j \in \trace$ s.t. $\tr_i \isoorder_\trace \tr_j$
    if there is $\op_i = \fwop{i}{\obj}{\any}{\ok_i} \in \trace|\tr_i$,
    $\op_i$ is a closing write, and
    there is $e_v = \get{j}{\obj}{\any} \in \trace|\tr_j$, 
    then $\op_i$ is
    \emph{in closing obbligato} iff $\exists e_u \in \trace|\tr_i$ s.t. $e_u =
    \rset{i}{\obj}{v}$ and $\op_i \instigates e_u$, and $e_u \prec_{\trace} e$.
\end{definition}

\begin{definition} [View Write Obbligato] 
    \label{def:exclusivity}
    \label{def:view-write-obbligato}
    Given $\tr_i \in \trace$, 
    if $\exists \tr_j \in \trace$, s.t. $\tr_i \isoorder_\trace \tr_j$,  
    if there is $\op_i = \fwop{i}{\obj}{\any}{\ok_i} \in \trace|\tr_i$,
    and $e_v = \get{j}{\obj}{\any} \in \trace|\tr_j$, then 
    $\op_i$ is \emph{in view write obbligato} iff
    there is $e_u = \set{i}{\obj}{\any} \in \trace|\tr_i$
    s.t. $e_u \prec_\trace e_v$ or 
    $\exists r = \res{i}{}{\ab_i} \in \trace|\tr_i$ s.t. $r \prec_{\trace} e_v$.
\end{definition}

\begin{definition} [Obbligato] 
    \label{def:actualization}
    \label{def:obbligato}
Trace $\trace$ is obbligato iff 
    \begin{enumerate}[a) ]
        \item all non-local writes in all transactions
            committed in $\trace$ are in committed obbligato,
        \item all closing writes whose effects are potentially viewed are in closing write obbligato,
        \item all writes whose effects are potentially viewed are in view write
            obbligato.
\end{enumerate}
\end{definition}

\emph{Decisiveness} is achieved, when transactions do not let other transactions
to view the values they set to the variables they modify until they commit or
perform their closing writes.

\begin{definition} [Decisiveness] \label{def:decisiveness}
Trace $\trace$ is \emph{decisive} iff given any pair of transactions
$\tr_i, \tr_j \in \trace$, s.t. $\tr_i \views \tr_j$ for any $e_u =
\rset{j}{\obj}{\val} \in \trace|\tr_j$ and $e_v = \get{i}{\obj}{\val} \in
\trace|\tr_i$, then either 
$\tr_j$ is decided on $\obj$,
$\exists r = \res{j}{}{\co_j} \in \trace|\tr_j$ s.t. $e_u \prec_\trace r
\prec_\trace e_v$.
\end{definition}

\emph{Abort accord} is a relation between two transactions, where if one of
them views the update events performed by the other, and the other transaction
aborts, then the first transaction is not permitted to abort.

\begin{definition} [Abort Accord] 
    \label{def:abort-abidement}
    \label{def:abort-accord}
Trace $\trace$ is in abort accord iff for any two transactions $\tr_i$ and
$\tr_j$ in $\trace$ s.t. 
\begin{inparaenum}[a)]
    \item %
        $\tr_j \views \tr_i$, if $\tr_i$ is aborted in $\trace$, 
        then $\tr_j$ is either live or aborted in $\trace$,
    \item %
        $\exists e_u = \rset{i}{\obj}{\any} \in \trace|\tr_i$ and $e =
        \rset{j}{\obj}{\any} \in \trace|\tr_j$ or $e = \get{j}{\obj}{\any} \in \trace|\tr_j$ and $e_a =
        \aset{i}{\obj}{\any} \in \trace|\tr_i$, and $e_u \prec_\trace e \prec_\trace
        e_a$,
        then $\tr_j$ is either live or aborted in $\trace$.
\end{inparaenum}
\end{definition}

\emph{Commit accord} is a similar relation, where given two transactions such that one of
them views the update events performed by the other, and the first transaction
commits, then the first transaction must have also committed.

\begin{definition} [Commit Accord] 
    \label{def:commit-abidement}
    \label{def:commit-accord}
trace $\trace$ is in commit accord iff for any two transactions $\tr_i$ and
$\tr_j$ in $\trace$ s.t. $\tr_j \views \tr_i$, if $\tr_j$ is committed in
$\trace$, then $\tr_i$ is committed in $\trace$.
\end{definition}

\emph{Coherence} specifies, that if a transaction commits, all preceding
transactions according to the isolation order either committed or aborted
beforehand.

\begin{definition} [Coherence] \label{def:coherence}
Trace $\trace$ is coherent iff for any two transactions  $\tr_i$ and
$\tr_j$ in $\trace$ s.t. $\tr_i \iso{\obj}{\trace} \tr_j$, if $\exists r_j =
\res{j}{}{\co_j} \in \trace|\tr_j$, then $\exists r_i = \res{i}{}{\co_i}$ or
$r_i = \res{i}{}{\ab_i}$ and $r_i \prec_\trace r_j$.
\end{definition}

\emph{Abort Coda} specifies when a recovery event can be expected to be issued.
If a transaction updates the state of some variable and eventually aborts,
either it or another transaction will issue a recovery event to clean up that
update before the transaction in question completes aborting. On the other
hand, if the transaction commits, neither it or any other transaction will
issue a recovery event to revert the state of that variable to another value.

\begin{definition} [Abort Coda] \label{def:abort-sanitation}
                                \label{def:abort-coda}
    Trace $\trace$ has \emph{coda} iff for any transaction $\tr_i$ 
    \begin{enumerate}[a)]
        \item if $\tr_i$ aborts in $\trace$ 
            (so $r = \res{i}{}{\ab_i} \in \trace|\tr_i$), then 
            if $\exists e_u = \rset{i}{\obj}{\val} \in \trace|\tr_i$, then 
            for some $\tr_j$ s.t. $i=j$ or $\tr_i \iso{\obj}{\trace} \tr_j$
            $\exists e_a = \aset{i}{\obj}{\val'} \in \trace$ 
            s.t. $e_u \prec_\trace e_a \prec_\trace$,
        \item if $\tr_i$ commits in $\trace$ 
            (so $\exists r = \res{i}{}{\co_i} \in \trace|\tr_i$), 
            then if $\exists e = \rset{i}{\obj}{\val} \in \trace|\tr_i$ 
                 or $e = \get{i}{\obj}{\val} \in \trace|\tr_i$, then
            for any $\tr_j$ s.t. $i=j$ or $\tr_i \iso{\obj}{\trace} \tr_j$
            $\nexists e_a = \aset{i}{\obj}{\val'} \in \trace$ 
            s.t. $e \prec_\trace e_a \prec_\trace$.
    \end{enumerate}
\end{definition}

\emph{Chain consistency} describes what events are allowed and barred from a chain of transactions.
Specifically, \emph{chain isolation} stipulates that, a chain of transactions executing view
and update events is not broken by a revert event, so a transaction cannot view
an inconsistent state where the value of one variable is retrieved before an
abort was performed, and another one after.
\emph{Chain self-containment} if the values viewed by a transaction in some
chain always come from that chain.

\begin{definition} [Chain Isolation] \label{def:chain-isolation}
Given trace $\trace$, transactions $\tr_i, \tr_j \in \trace$,
$\vchain{\trace}{i}{j}$ is isolated if for $\forall \tr_k \in
\vchain{\trace}{i}{j}$ s.t. $e^k = \rset{k}{\obj}{\val}$, there is no $\tr_l$
(possibly $\tr_l \not\in \vchain{\trace}{i}{j}$) s.t. $\exists e^l =
\aset{l}{\obj}{\val'}$ where $\val = \val'$ and $e^l$ is between $e^k$ and any
other event in any transaction in $\vchain{\trace}{i}{j}$.
\end{definition}

\begin{definition} [Chain Self-containment] \label{def:chain-self-containment}
      Given trace $\trace$, transactions $\tr_i, \tr_j \in \trace$,
      $\vchain{\trace}{i}{j}$ is self-contained iff given any transactions
      $\tr_k, \tr_l \in \vchain{\trace}{i}{j}$, s.t. $k\neq l$ and $\exists e_u^k =
      \rset{k}{\obj}{\val} \in \trace|\tr$ $\exists e_v^l =
      \get{l}{\obj}{\val'} \in \trace|\tr$ and $e_u^k \prec_\trace e_v^l$, then
      either $\val = \val'$ or $\exists e_u^m = \rset{m}{\obj}{\val'} \in
      \trace|\tr_m$ for some $\tr_m \in \vchain{\trace}{i}{j}$ s.t. 
      $\tr_m$ precedes $\tr_l$ and follows $\tr_k$ in $\vchain{\trace}{i}{j}$ and
      $e_u^k \prec_\trace e_u^m \prec_\trace e_v^l$.
\end{definition}

\begin{definition} [Chain Consistency] \label{def:trace-chain-consistency}
    An isolated trace $\trace$ is chain-consistent if given
    any $\vchain{\trace}{i}{j}$ is chain-isolated and self-contained
    (for some $\tr_i, \tr_j \in \trace$).
\end{definition}

Finally, \emph{harmony} is satisfied if the preceding properties are satisfied
within the entire trace.

\begin{definition} [Harmony] \label{def:harmony}
    Trace $\trace$ is harmonious iff it satisfies all of the following:
    \begin{inparaenum}[a)]
        \item minimalism,
        \item consonance, 
        \item obbligato, 
        \item coherence, commit accord, abort accord, and abort coda,
        \item isolation, 
        \item decisiveness,
        \item chain consistency,
        \item unique writes.
    \end{inparaenum}
\end{definition}

\subsection{Last-use Opacity from Harmony}

\begin{theorem} [Harmonious Trace Last-use Opacity] \label{thm:harmony-to-lu-opacity}
Given harmonious $\hist$, s.t. $\hist = \tohist\trace$, if $\trace$ is
harmonious, $\hist$ is last-use opaque.
\end{theorem}

\section{OptSVA Harmony}
\label{sec:optsva-harmony}

Let $\traceos$ be any trace produced by OptSVA.

\begin{observation} [Memory Access Pattern]
OptSVA generates view and update events for variable $\obj$ precisely as a result of executing the
following lines:
\begin{itemize}
    \item in procedure \ppr{checkpoint} at \rln{checkpoint:get}---view event,
    \item in procedure \ppr{read\_buffer} at \rln{read-buffer:get}---view event,
    \item in procedure \ppr{write\_buffer} at \rln{write-buffer:rset}---routine update event,
    \item in procedure \ppr{commit} at \rln{commit:rset}---routine update event,
    \item in procedure \ppr{abort} at \rln{abort:aset}---recovery update event.
\end{itemize}
\end{observation}

\begin{observation} [Closing Write Identification]
    If after executing a write operation on $\obj$ by $\tr_i$ it is true that
    $\wub{i}{\obj} = \wc{i}{\obj}$, then that is the closing write operation
    execution on $\obj$ in $\tr_i$.
\end{observation}

\begin{lemma} [Version Order]
    Any two transactions $\tr_i,\tr_j \in \traceos$ s.t. $\accesses{i} \cap
    \accesses{j} \neq \varnothing$  are \emph{isolation ordered}:  if $\exists
    \obj \in \accesses{i} \cap \accesses{j}$ $\pv{\obj}{i} < \pv{\obj}{j}$,
    then $\forall \obj,\objy \in \accesses{i} \cap \accesses{j}, \pv{\obj}{i} <
    \pv{\obj}{j}$.
\end{lemma}

\begin{proof}
    During \ppr{start} every transaction acquires a value of $\pv{\obj}{i}$.
    Since the acquisition is guarded by locks, it is performed atomically, so
    that if transaction $\tr_i$, starts acquiring $\forall \obj \in
    \accesses{i}, \pv{\obj}{i}$, then no other $\tr_j$ acquires $\forall \obj
    \in \accesses{j} \cap \accesses{i}, \pv{\obj}{i}$ until transaction $\tr_i$
    completes acquiring and releases the locks. Hence, if for any two
    $\tr_i,\tr_j$, if $\exists \obj \in \accesses{i} \cap \accesses{j}$
    $\pv{\obj}{i} < \pv{\obj}{j}$, then $\forall \obj,\objy \in \accesses{i}
    \cap \accesses{j}, \pv{\obj}{i} < \pv{\obj}{j}$.
\end{proof}

\begin{corollary}
    [Version Order from Isolation Order] \label{lemma:lambda}
    \label{cor:lambda}
    Given transactions $\tr_i, \tr_j$ s.t. $\tr_i \isoorder_{\traceos} \tr_j$,
    then $\forall \obj \in \accesses{i} \cap \accesses{j}, \pv{\obj}{j} <
    \pv{\obj}{i}$.
\end{corollary}

\begin{lemma} [Minimalism] \label{lemma:minimalism}
    $\traceos$ is minimalistic.
\end{lemma}

\begin{proof} %
    If $\obj$ is read-only in $\tr_i$, then there is exactly one view event on
    $\obj$ in $\tr_i$ (\rln{read-buffer:get}). If $\obj$ is not read-only, then
    there is exactly one view event on $\obj$ in $\tr_i$ executed as part of
    procedure \ppr{checkpoint} (\rln{checkpoint:get}), either during the first
    \ppr{read}, the closing \ppr{write} (in \ppr{write\_buffer}), or, if not
    previously invoked, during \ppr{commit}.

    Routine update events are executed only after the closing \ppr{write} (in
    \ppr{write\_buffer}---\rln{write-buffer:rset}), so at most once, or during
    \ppr{commit} (\rln{commit:rset}), if there were writes, but the upper bound
    on writes was not reached. Hence, routine update events occur at most once
    per variable.

    A recovery update event can occur only during \ppr{abort} \rln{abort:aset},
    at most once per variable.
\end{proof}

\begin{lemma} [Obligatory Checkpoints] \label{lemma:epsilon}
    If $\tr_i$ issues an update event or a view update event, $\tr_i$
    invoked \ppr{checkpoint}.
\end{lemma}

\begin{proof}
    View events are only executed as part of \ppr{checkpoint}. 
    
    A routine update
    event is only executed as part of \ppr{write\_buffer} at
    \rln{write-buffer:rset}, which is dominated by lines
    \ref{l:wb-checkpoint-cond}--\ref{l:wb-checkpoint}, which executes
    \ppr{checkpoint} if it was not previously executed. 

    A recovery update event occurs as a result of executing \rln{abort:aset},
    which is guarded by a condition that $\wc{i}{\obj} > 0$, so a write must
    have been executed. Furthermore, $\pv{\obj}{i} - 1 > \lv{\obj}$ must be
    true, which implies that $\tr_i$ released $\obj$, which means the closing
    write executed, so \ppr{write\_buffer} was started asynchronously. 
    That procedure executes a \ppr{checkpoint} if it was not executed beforehand at 
    lines \ref{l:wb-checkpoint-cond}--\ref{l:wb-checkpoint}.
\end{proof}

\begin{lemma} [Always View Before Update]\label{lemma:thorn}
    If transaction $\tr_i$ issues an update event $e_u = \set{i}{\obj}{\any}$
    in trace $\traceos$, then there is $e_v = \get{i}{\obj}{\any} \in
    \traceos|\tr_i$ s.t. $e_v \prec_\traceos e_u$.
\end{lemma}

\begin{proof}
    From \rlemma{lemma:epsilon}, if $\tr_i$ executes an update event, then it
    executes \ppr{checkpoint} before the event is issued. Since
    \ppr{checkpoint} issues a view event, then a view event is issued before an
    update event.
\end{proof}

\begin{lemma} [Wait at Access] \label{lemma:gamma}
    Given transactions $\tr_i, \tr_j$ s.t. $\pv{\obj}{j} < \pv{\obj}{i}$,
    $\tr_i$ does not issue a view or update event on $\obj$ until $\tr_j$
    executes \ppr{release} on $\obj$, \ppr{abort}, or \ppr{commit}.
\end{lemma}

\begin{proof}
    Let $\tr_k$ be such that $\pv{\obj}{k} = \pv{\obj}{i} - 1$.
    Every invocation of \ppr{checkpoint} is dominated by an instruction that
    waits until the condition    $\pv{\obj}{i} - 1 = \lv{\obj}$:
    \rln{read:checkpoint} by \rln{read-access}, \rln{commit-checkpoint} by
    \rln{commit:access}, and \rln{wb-checkpoint} by
    \rln{write-no-reads-async}.
    Since, from \rlemma{lemma:epsilon}, every view or update event is preceded
    by the invocation of \ppr{checkpoint}, then each view or update event is
    dominated by an instruction that waits until $\pv{\obj}{i} - 1 =
    \lv{\obj}$.
    Hence in order for $\tr_i$ to issue a event it must be true that
    $\pv{\obj}{i} - 1 = \lv{\obj}$.
    
    In order for that condition to be met, some transaction must set
    $\lv{\obj}$ to $\pv{\obj}{i} - 1$ (or $\pv{\obj}{i} = 1$, but then there
    could not be such $\tr_j$ as assumed).
    Some transaction $\tr_k$ modifies a new value of $\lv{\obj}$ during
    \ppr{release}, \ppr{abort}, or \ppr{commit} and the value is there set to $\pv{\obj}{k}$.
    Hence $\tr_i$ cannot issue any view or update event until some $\tr_k$ such
    that $\pv{\obj}{k} = \pv{\obj}{i} - 1$ executes \ppr{release}, \ppr{abort},
    or \ppr{commit}.

    Every invocation of \ppr{release} (by $\tr_k$) is dominated by an
    instruction that waits until the condition $\pv{\obj}{k} - 1 = \lv{\obj}$
    is met: the invocation at 
    \rln{rb-release} by \rln{reads-b}, and the one at \rln{wb-release}
    by \rln{write-no-reads-async}.
    Furthermore, modifying $\lv{\obj}$ within \ppr{commit} (\rln{commit-lv}) or
    \ppr{abort} (\rln{abort-lv}) also requires that $\pv{\obj}{k} - 1 =
    \lv{\obj}$ be first satisfied (at \rln{abort-access} and
    \rln{abort-access}, respectively).
    Hence $\tr_k$ cannot set $\lv{\obj}$ to $\pv{\obj}{k}$ view or update event
    unless $\pv{\obj}{k} = 1$ or until some $\tr_l$ such that $\pv{\obj}{l} =
    \pv{\obj}{k} - 1$ executes \ppr{release}, \ppr{abort}, or \ppr{commit}.

    Assuming that $\pv{\obj}{k} > 1$, and that some $\tr_l$ s.t. $\pv{\obj}{l}
    = \pv{\obj}{k} - 1$ exists, then, since $\tr_i$ cannot issue any view or
    update event until $\tr_k$ sets $\lv{\obj}$ in \ppr{release}, \ppr{abort},
    or \ppr{commit} and since $\tr_k$ cannot set $\lv{\obj}$ until $\tr_l$
    executes \ppr{release}, \ppr{abort}, or \ppr{commit}, then $\tr_i$ cannot
    issue any view or update events until $\tr_l$ executes \ppr{release},
    \ppr{abort}, or \ppr{commit}. Since $\pv{\obj}{i} - 1 = \pv{\obj}{k}$ and
    $\pv{\obj}{k} - 1 = \pv{\obj}{l}$ then $\pv{\obj}{l} < \pv{\obj}{i}$.
    
    It follows by induction then that given any $\tr_j$ s.t. $\pv{\obj}{j} <
    \pv{\obj}{i}$, $\tr_i$ does not issue a view or update event on $\obj$
    until $\tr_j$ executes \ppr{release}, \ppr{abort}, or \ppr{commit}.
\end{proof}

\begin{lemma} [Recovery Versions from Version Order] \label{lemma:kappa}
    Given transactions $\tr_i, \tr_j$ s.t. $\pv{\obj}{j} < \pv{\obj}{i}$, 
    if $\tr_j$
    executes \ppr{abort} before $\tr_i$ executes \ppr{checkpoint}, 
    $\rv{j}{\obj} \leq \rv{i}{\obj}$.
    otherwise $\rv{j}{\obj} < \rv{i}{\obj}$.
\end{lemma}

\begin{proof}
    Transaction $\tr_i$ sets $\rv{i}{\obj}$ to $\cv{\obj}$ only during
    \ppr{checkpoint} (\rln{checkpoint-cv}).
    Every invocation of \ppr{checkpoint} is dominated by an instruction that
    waits until the condition    $\pv{\obj}{i} - 1 = \lv{\obj}$:
    \rln{read:checkpoint} by \rln{read-access}, \rln{commit-checkpoint} by
    \rln{commit:access}, and \rln{wb-checkpoint} by
    \rln{write-no-reads-async}.

    In order for that condition to be met, some transaction must set
    $\lv{\obj}$ to $\pv{\obj}{i} - 1$ (or $\pv{\obj}{i} = 1$, but then there
    could not be such $\tr_j$ as assumed, so necessarily $\pv{\obj}{i} > 1$).
    Some transaction $\tr_k$ can set a new value of $\lv{\obj}$ during
    \ppr{release}, \ppr{abort}, or \ppr{commit}. 
    Hence $\tr_i$ sets the value of $\rv{i}{\obj}$ only after $\tr_k$ such
    that $\pv{\obj}{k} = \pv{\obj}{i} - 1$ executes \ppr{release}, \ppr{abort},
    or \ppr{commit}. 
    Thus, since value of $\cv{\obj}$ is there set by $\tr_k$ to $\pv{\obj}{k}$
    in case of \ppr{release} (\rln{release-cv}) and \ppr{commit}
    (\rln{commit-cv}), or $\rv{k}{\obj}$ in case of \ppr{abort}
    (\rln{abort:recover-cv}), $\rv{i}{\obj} = \rv{k}{\obj}$ if $\tr_k$ aborts before
    $\tr_i$ executes \ppr{checkpoint} and $\rv{i}{\obj} = \pv{\obj}{k}$
    otherwise.

    Since $\rv{k}{\obj}$ either trivially equals $0$ if $\pv{\obj}{k} =
    1$, or is acquired by analogy from some $\tr_l$ s.t. $\pv{\obj}{l} =
    \pv{\obj}{k} - 1$, then $\rv{k}{\obj} \leq \rv{i}{\obj}$.
    
    Furthermore, under the assumption that $\tr_k$ does not execute \ppr{abort}
    prior to $\tr_i$ executing \ppr{checkpoint}, then value of $\cv{\obj}$ is
    there set by $\tr_k$ only either within \ppr{release} or \ppr{commit}, and
    thus $\cv{\obj} = \pv{\obj}{k}$ during $\tr_i$'s \ppr{checkpoint}, so
    $\rv{i}{\obj} = \pv{\obj}{k}$. Since $\pv{\obj}{k} < \pv{\obj}{i}$, then
    $\rv{k}{\obj} < \rv{i}{\obj}$.

    By extension, given $\tr_j$ s.t. $\pv{\obj}{j} < \pv{\obj}{i}$, either
    $k=j$, or $\pv{\obj}{j} < \pv{\obj}{k}$. 

    In the former case, necessarily $\rv{k}{\obj} \leq \rv{i}{\obj}$ if $\tr_j$
    executes \ppr{abort} before $\tr_i$ executes \ppr{checkpoint}, or
    $\rv{k}{\obj} < \rv{i}{\obj}$.

    In the latter case, there must be some $\tr_l$ s.t.
    $\pv{\obj}{l} = \pv{\obj}{k}$. Then, if $\tr_l$ executes \ppr{abort} before
    $\tr_k$ executes \ppr{checkpoint}, $\rv{l}{\obj} \leq \rv{k}{\obj}$,
    otherwise $\rv{l}{\obj} < \rv{k}{\obj}$. 
    Furthermore, either $l=j$, or $\pv{\obj}{j} < \pv{\obj}{l}$.
    
    It then follows by induction that given any $\tr_j$ s.t. $\pv{\obj}{j} <
    \pv{\obj}{i}$, if $\tr_j$ executes \ppr{abort} before
    $\tr_i$ executes \ppr{checkpoint}, $\rv{l}{\obj} \leq \rv{k}{\obj}$,
    otherwise $\rv{l}{\obj} < \rv{k}{\obj}$. 
\end{proof}

\begin{lemma} [Isolation] \label{lemma:isolation}
    Trace $\traceos$ is isolated.
\end{lemma}

\begin{proof}    
    Every routine update event, view event, and recovery event is dominated by
    an access conditions ($\pv{\obj}{i} - 1 = \lv{\obj}$). This condition is
    satisfied for $\tr_i$ if $\lv{\obj} = 0$ and $\pv{\obj}{i} = 0$, or if some
    transaction $\tr_j$ s.t. $\pv{\obj}{j} = \pv{\obj}{i} - 1$ releases $\obj$
    by setting $\lv{\obj}$ to $\pv{\obj}{j}$ during commit or after closing
    write or after the first non-local read (and thus after any
    $\rset{j}{\obj}{\any}$ or $\get{j}{\obj}\any$).
    
    Since events are guarded by access conditions, since variables are released
    after all view or routine update events are issued by a transaction, and
    since each transactions are version-ordered, then for any $\tr_i, \tr_j$, 
    $\exists \obj \in \accesses{i} \cap \accesses{j}$ 
    if $\exists e_i = \rset{i}{\obj}{\any} \in \traceos|\tr_i$ 
    or $e_i = \get{i}{\obj}{\any} \in \traceos|\tr_i$ 
    and $\exists e_j = \rset{j}{\obj}{\any} \in \traceos|\tr_j$ 
    or $e_j = \get{j}{\obj}{\any} \in \traceos|\tr_j$, 
    and $e_i \prec_\trace e_j$ then
    $\forall \objy \in \accesses{i} \cap \accesses{j}$,
    if $\exists e_i' = \rset{i}{\objy}{\any} \in \traceos|\tr_i$ 
    or $e_i' = \get{i}{\objy}{\any} \in \traceos|\tr_i$ 
    and $\exists e_j' = \rset{j}{\objy}{\any} \in \traceos|\tr_j$ 
    or $e_j' = \get{j}{\objy}{\any} \in \traceos|\tr_j$, 
    and $e_i' \prec_\trace e_j'$.
\end{proof}

\begin{corollary} [Isolation Order]
    Trace $\traceos$ is isolation-ordered.
\end{corollary}

\begin{lemma} [Write Consonance] \label{lemma:write-consonance}
    Any (complete) write operation in $\traceos$ is consonant.
\end{lemma}

\begin{proof} %
    Each write is guarded by the condition at \rln{write:domain}, which aborts
    the transaction if the value that is supposed to be written is not within
    the domain of the variable. 
    Thus, each write is consonant.
\end{proof}

\begin{lemma} [Routine Update Consonance] \label{lemma:routine-update-consonance}
    Any routine update event in $\traceos$ is consonant.
\end{lemma}

\begin{proof} %
    A routine update event $\rset{i}{\obj}{\val}$ occurs either as a result of
    executing a closing write operation on $\obj$ (\rln{write-buffer:rset}) or
    $\tr_i$ committing (\rln{commit:rset}), if the transaction executed writes,
    but the upper bound for writes was not reached for $\obj$.
    Clearly, then, if there was a routine update event, then $\tr_i$
    executed a write operation on $\obj$.
    In both cases above $\val = \buf{i}{\obj}$ and $\buf{i}{\obj}$ can be set
    by any write operation, the first non-local read operation, or during start
    for read-only variables. If there was a write, then $\obj$ is not
    read-only, and the first non-local read cannot follow a write, so
    $\buf{i}{\obj}$ is set within (the most recent) write operation executed by
    $\tr_i$ and corresponds to the value written by that operation.
    Thus, $\forall \tr_i, \forall e^i_u = \rset{i}{\obj}{\val} \in
    \traceos|\tr_i$, $\exists op_i = \fwop{i}{\obj}{\val}{\ok_i} \in
    \traceos|\tr_i$ s.t. $\op_i \instigates e^i_u$.
    Therefore, $e^i_u$ is consonant.
\end{proof}

\begin{lemma} [View Consonance] \label{lemma:view-consonance}
    Any view event in $\traceos$ is consonant.
\end{lemma}

\begin{proof}
    If a view event occurs, it views the current state of a variable.  So given
    transaction $\tr_i$, if there is a view event $e_v = \get{i}{\obj}{\val}
    \in \traceos|\tr_i$, $\val$ corresponds to the current state of $\obj$.
    The only way to change the state of $\obj$ is via an update event on
    $\obj$.
    Thus, trivially, for some $\tr_j$, if there is $e_u = \rset{j}{\obj}{\val'}
    \in \traceos|\tr_j$ or $e_a = \aset{j}{\obj}{\val'} \in \traceos|\tr_j$, if
    either $e_u$ or $e_a$ precede $e_v$ so that no other update event on $\obj$
    occurs between either $e_u$ or $e_a$ and $e_v$, then $v=v'$.
    Furthermore, from unique routine updates, there cannot there is no $e_u$
    s.t. $\val' = 0$, and since $\obj$ is initially $0$, then $\nexists e_u$
    s.t. $\val' = 0$ and $e_u \pref e_v$.
\end{proof}

\begin{lemma} [Local Read Consonance] \label{lemma:local-read-consonance}
    Any local read operation execution in $\traceos$ is consonant.
\end{lemma}

\begin{proof}
    If transaction executes a local read on $\obj$, then it previously executed
    a write operation on $\obj$, so $\wc{i}{\obj} > 0$. Thus, the read procedure
    returns at \rln{local-read:return}, returning $\buf{i}{\obj}$. 
    The value of $\buf{i}{\obj}$ can be set by any write operation, the first
    non-local read operation, or during start for read-only variables. If there
    was a write, then $\obj$ is not read-only, and the first non-local read
    cannot follow a write, so $\buf{i}{\obj}$ is set within (the most recent)
    write operation executed by $\tr_i$ and corresponds to the value written by
    that operation.
    Thus, $\forall \tr_i, \forall op^i = \frop{i}{\obj}{\val} \in
    \traceos|\tr_i$ if $\op_i$ is local, $\exists op_i' =
    \fwop{i}{\obj}{\val'}{\ok_i} \in \traceos|\tr_i$ s.t. $\val' = \val$.
    Therefore, $\op_i$ is consonant.
\end{proof}

\begin{lemma} [Non-local Read Consonance] \label{lemma:non-local-read-consonance}
    Any non-local read operation execution in $\traceos$ is consonant.
\end{lemma}

\begin{proof}
    If $\obj$ is read-only in $\tr_i$, then during \ppr{start}, a view event
    occurs within \ppr{read\_buffer} (\rln{read-buffer:get}), and the state of
    $\obj$ is saved in $\buf{i}{\obj}$. Then, subsequent writes return the
    value of $\buf{i}{\obj}$ (\rln{read-only:return}) (waiting if necessary).
    Thus, they depend on that view event.

    Otherwise, a non-local read operation on $\obj$ is one that is not preceded
    by a write on $\obj$, so $\wc{i}{\obj} = 0$. The first such read executes
    \ppr{checkpoint} which initiates a view event (\rln{checkpoint:get}). The
    value obtained by that event is saved in $\stored{i}{\obj}$ and later
    $\buf{i}{\obj}$ is set to the same value. Finally, that value is returned
    at \rln{non-local-read:return}.  Subsequent non-local reads use the same
    value stored in $\buf{i}{\obj}$. The value remains unchanged, since it only
    be overwritten by a write, but the occurrence of a preceding write would
    mean the read is local (and since $\obj$ is not read-only, and there was a
    preceding non-local read). Thus, all non-local reads depend on the view
    event issued during \ppr{checkpoint}.

    Thus, $\forall \tr_i, \forall op^i = \frop{i}{\obj}{\val} \in
    \traceos|\tr_i$ if $\op_i$ is non-local, $\exists e_v^i =
    \get{i}{\obj}{\val'} \in \traceos|\tr_i$ s.t. $\val' = \val$.    
\end{proof}

\begin{lemma} [Conservative Recovery Update Events] \label{lemma:conservative}
    Any recovery update event in $\traceos$ is conservative.
\end{lemma}

\begin{proof}
    The recovery update event in $\tr_i$ occurs as a result of executing
    \rln{abort:aset}, which updates the state of $\obj$ to $\stored{i}{\obj}$.
    This is done only if
    $\rv{i}{\obj} \neq \cv{\obj}$ and $\wc{i}{\obj} > 0$.
    If is to be true, Since $\rv{i}{\obj}$ is set to the value of $\cv{\obj}$
    only during \ppr{checkpoint} and \ppr{read\_buffer}, and since the
    requirement that $\wc{i}{\obj} > 0$ excludes the latter, this condition
    checks whether the current transaction previously made a checkpoint.
    Executing \ppr{checkpoint} entails a view event that sets
    $\stored{i}{\obj}$ to the current value of $\obj$.  Hence, if $e_a =
    \aset{i}{\obj}{\val} \in \trace|\tr_i$ then there exists  $e_v =
    \get{i}{\obj}{\val} \in \trace|\tr_i$ s.t. $e_a \prec e_v$.
\end{proof}

\begin{lemma} [Clean Recovery Update Events] \label{lemma:clean}
Any recovery update event in $\traceos$ is clean.
\end{lemma}

\begin{proof}
    Assume by contradiction that there exists $e_a = \aset{i}{\obj}{\val}$ in
    $\traceos|\tr_i$ and $e_v = get{i}{\obj}{\val}$ that justifies that $e_s$ is
    is conservative, and $e_a' = \aset{j}{\obj}{\val'}$ in $\traceos|\tr_j$ s.t.
    $\tr_j \iso{\obj}{\traceos} \tr_i$ and $e_v \prec_\trace e_a' \prec_\traceos
    e_a$.
    This implies that $\tr_j$ executes \ppr{abort} 
    (and satisfies the condition
    $\rv{j}{\obj} = \cv{\obj}$) between the point at which $\tr_i$ executes
    \ppr{checkpoint} and \ppr{abort} %
    If that is the case, as a result of executing \ppr{abort}, $\tr_j$ sets
    $\cv{\obj}$ to $\rv{j}{\obj}$, and the.
    
    Given that $\tr_j \iso{\obj}{\traceos} \tr_i$, then $\pv{\obj}{j} <
    \pv{\obj}{i}$.  
    Since any execution of \ppr{checkpoint} for some $\tr_k$ is guarded by the
    condition $\pv{\obj}{k} - 1 = \lv{\obj}$, then $\tr_j$ executes
    \ppr{checkpoint} before $\tr_i$. 
    Hence, $\tr_j$ acquires $\rv{j}{\obj}$ from $\cv{\obj}$ before $\tr_i$
    acquires $\rv{i}{\obj}$ from $\cv{\obj}$.

    The value of $\rv{j}{\obj}$ is equal to the value of $\cv{\obj}$ at the
    point when $\tr_j$ executed checkpoint (i.e. when $\pv{\obj}{i} - 1 =
    \lv{\obj}$). The value of $\cv{\obj}$ is set to $\pv{\obj}{k}$ when $\tr_k$
    executes \ppr{release} or \ppr{commit} or to $\rv{k}{\obj}$ when $\tr_k$
    aborts. Thus, when $\tr_j$ executes \ppr{checkpoint}, since $\pv{\obj}{i} -
    1 = \lv{\obj}$, then either: 
    \begin{enumerate}[a) ]
        \item $\cv{\obj} = \pv{\obj}{k} = \pv{\obj}{j} - 1$ (if $\tr_k$
            released $\obj$ or committed),
        \item $\cv{\obj} = \rv{k}{\obj}$ and $\rv{k}{\obj} < \pv{\obj}{j}$ (if
            $\tr_k$ aborted), or 
        \item $\cv{\obj} = 0$ (if there is no such $\tr_k$).
    \end{enumerate}
    In any case, $\rv{j}{\obj} < \pv{\obj}{i}$.

    $\tr_i$ is capable of executing \ppr{checkpoint} after $\tr_j$ commits,
    aborts, or releases $\obj$.  Since $\tr_j$ executes \ppr{abort} between
    $\tr_i$'s \ppr{checkpoint} and \ppr{abort}, then only the third option
    remains. If $\tr_j$ executes \ppr{release} for $\obj$, then it sets
    $\cv{\obj}$ to $\pv{\obj}{j}$. 
    Following the logic from the previous paragraph, this means that when
    $\tr_i$ assigns $\cv{\obj}$ to $\rc{i}{\obj}$, $\pv{\obj}{j} \leq
    \cv{\obj}$, so $\pv{\obj}{j} \leq \rv{i}{\obj}$, so $\rv{j}{\obj} <
    \rv{i}{\obj}$.
    
    Hence, after $\cv{\obj}$ to $\rv{j}{\obj}$ during abort, it is not true
    that $\cv{\obj} = \rv{j}{\obj}$. Thus, $e_a$ cannot occur once $e_a'$
    occurs, which is a contradiction.
\end{proof}

\begin{lemma} [Needed Recovery Update Events] \label{lemma:needed}
    Any recovery update event in $\traceos$ is needed.
\end{lemma}

\begin{proof}
    The recovery update event occurs as a result of executing \rln{abort:aset},
    which is guarded by a condition that $\wc{i}{\obj} > 0$, so a write must
    have been executed. Furthermore, $\pv{\obj}{i} - 1 > \lv{\obj}$ must be
    true, which implies that $\tr_i$ released $\obj$, which means the closing
    write executed, so \ppr{write\_buffer} was started asynchronously. If that
    is the case, the recovery update event cannot execute until
    \ppr{write\_buffer}, which means a routine update event on $\obj$ will have
    executed before the recovery update event on $\obj$.
\end{proof}

\begin{lemma} [Dooming Recovery Update Events] \label{lemma:dooming}
Any recovery update event in $\traceos$ is dooming.
\end{lemma}

\begin{proof}
    Trivially, since any recovery update event occurs only within \ppr{abort}.
\end{proof}

\begin{lemma} [Ending Recovery Update Events] \label{lemma:ending}
Any recovery update event in $\traceos$ is ending.
\end{lemma}

\begin{proof}
    Trivially, since any recovery update event occurs only within \ppr{abort},
    and there are no other update or view events on the same variable in
    \ppr{abort}.
\end{proof}

\begin{lemma} [Recovery Update Consonance] \label{lemma:recovery-update-consonance}
Any recovery update event in $\traceos$ is consonant.
\end{lemma}

\begin{proof}
    Since each recovery update is conservative (from
    \rlemma{lemma:conservative}), needed (from \rlemma{lemma:needed}), dooming
    (from \rlemma{lemma:dooming}), ending (from \rlemma{lemma:ending}), and
    clean (from \rlemma{lemma:clean}), then each recovery update is consonant.
\end{proof}

\begin{lemma} [Trace Consonance] \label{lemma:consonance}
Trace $\traceos$ is consonant.
\end{lemma}

\begin{proof}
    From Lemmata
    \ref{lemma:write-consonance}--\ref{lemma:non-local-read-consonance} and
    \ref{lemma:recovery-update-consonance}.
\end{proof}

\begin{lemma} [Comitted Write Obbligato] \label{lemma:comitted-write-obbligato}
Given $\tr_i$ that is committed in $\traceos$, every non-local write operation execution
$\op_i = \fwop{i}{\obj}{\val}{ok_i} \in \traceos|\tr_i$ is in committed obbligato.
\end{lemma}

\begin{proof}
    If $\tr_i$ executes a write corresponding to $\op_i$, then, if at the end
    of the execution it is true that $\wc{i}{\obj} = \wub{i}{\obj}$,
    \ppr{write\_buffer} is executed, which causes a routine update event to
    execute, writing the value of $\buf{i}{\obj}$ to $\obj$. 
    
    Since
    $\wc{i}{\obj} = \wub{i}{\obj}$ no other writes follow, and since $\obj$ is
    not read-only in $\tr_i$, 
    then the
    value written to $\obj$ in \ppr{write\_buffer} is the value passed to the
    write operation. In that case there is $e_u = \rset{i}{\obj}{\val} \in
    \traceos|\tr_i$. Since commit will not return until \ppr{write\_buffer}
    finishes executing, then trivially $\inv{i}{}{\twop{i}{\obj}{\val}}
    \prec_{\traceos|\tr_i} e_u \prec_{\traceos|\tr_i} \res{i}{}{\co_i}$.
    
    If it is true that $\wc{i}{\obj} = \wub{i}{\obj}$, then \ppr{write\_buffer}
    is not executed, but during \ppr{commit}, the same condition is checked
    again, and if it is not satisfied, $\tr_i$ writes the value from
    $\buf{i}{\obj}$ to $\obj$. Thus, by analogy to the paragraph above, there
    is $e_u = \rset{i}{\obj}{\val} \in \traceos|\tr_i$. Since this is executed
    within commit, then $\inv{i}{}{\twop{i}{\obj}{\val}} \prec_{\traceos|\tr_i}
    e_u \prec_{\traceos|\tr_i} \res{i}{}{\co_i}$.
\end{proof}

\begin{lemma} [Closing Write Obbligato] \label{lemma:closing-write-obbligato}
Given $\tr_i$ that is decided on $\obj$ in $\traceos$, every non-local write
operation execution $\op_i = \fwop{i}{\obj}{\val}{ok_i} \in \traceos|\tr_i$ is
in closing write obbligato.
\end{lemma}

\begin{proof}
    If $\tr_i$ executes a write corresponding to $\op_i$, then, at the end of
    the execution, if $\op_i$ is a closing write it is necessarily true that
    $\wc{i}{\obj} = \wub{i}{\obj}$. This causes \ppr{write\_buffer} to be
    executed,\rln{write-no-reads-async} which causes a routine update event to
    execute, writing the value of $\buf{i}{\obj}$ to $\obj$.
    
    Since
    $\wc{i}{\obj} = \wub{i}{\obj}$ no other writes follow, and since $\obj$ is
    not read-only in $\tr_i$, 
    then the
    value written to $\obj$ in \ppr{write\_buffer} is the value passed to the
    write operation. In that case there is $e_u = \rset{i}{\obj}{\val} \in
    \traceos|\tr_i$. Since commit will not return until \ppr{write\_buffer}
    finishes executing, then trivially $\inv{i}{}{\twop{i}{\obj}{\val}}
    \prec_{\traceos|\tr_i} e_u \prec_{\traceos|\tr_i} \res{i}{}{\co_i}$.
    Hence $\op_i \instigates e_u$. $\tr_i$ executes \ppr{release} only
    following issuing $e_u$ \ppr{wb:release}.

    If $\tr_i \isoorder_\traceos \tr_j$, then $\pv{\obj}{i} < \pv{\obj}{j}$
    (\rcor{cor:lambda}). From \rlemma{lemma:gamma}, to issue any $e_v =
    \get{j}{\obj}{\any}$, since $\pv{\obj}{i} < \pv{\obj}{j}$, $\tr_i$ must
    executed \ppr{abort}, \ppr{commit}, or \ppr{release}. Hence $\tr_j$ does
    not issue $e_v$ before $\tr_i$ executes \ppr{release}, which requires that
    $\tr_i$ issues $e_u$ so that $e_u \prec_\traceos e_v$.
\end{proof}

\begin{lemma} [View Write Obbligato] \label{lemma:view-write-obbligato}
     Given $\tr_i \in \traceos$, 
     if $\exists \tr_j \in \trace$, s.t. $\tr_i \isoorder_\traceos \tr_j$,  
     if there is $\op_i = \wop{i}{\obj}{\any}{\ok_i} \in \traceos|\tr_i$,
     and $e_v = \get{j}{\obj}{\any} \in \traceos|\tr_j$, then 
     $\op_i$ is in in view write obbligato.
\end{lemma}

\begin{proof}
    If $\tr_i \isoorder_\traceos \tr_j$, then $\pv{\obj}{i} < \pv{\obj}{j}$
    \rcor{cor:lambda}.
    From \rlemma{lemma:gamma}  $e_v$  occurs only after  $\tr_i$ releases $\obj$,
    commits, or aborts.
    Since according to the assumption, $\tr_i$ cannot abort prior to $\tr_j$
    issuing $e_v$, $\tr_i$ either releases $\obj$ or commits prior to $\tr_j$ issuing $e_v$.

    If $\tr_i$ releases $\obj$ it executes \ppr{release}. This can occur as a
    result of executing \rln{rb-release} or \rln{wb-release}. Since $\tr_i$ executes $\op_i$,
    then \rln{rb-release} cannot be executed, since it can only be reached if $\tr_i$
    only ever reads $\obj$ (condition at \rln{reads-a}). 
    Hence $\tr_i$ must execute \rln{wb-release}, which is dominated by \rln{write-buffer:rset},
    which issues a write event $e_u = \set{i}{\obj}{\val}$, where $\val$ is the
    value of  $\buf{i}{\obj}$.

    Since \ppr{release} was executed at \rln{wb-release}, \ppr{write\_buffer} must
    have been executed at \rln{write-no-reads-async}.
    Then the
    value written to $\obj$ in \ppr{write\_buffer} is the value passed to the
    write operation. In that case there is $e_u = \rset{i}{\obj}{\val} \in
    \traceos|\tr_i$. Since commit will not return until \ppr{write\_buffer}
    finishes executing, then trivially $\inv{i}{}{\twop{i}{\obj}{\val}}
    \prec_{\traceos|\tr_i} e_u \prec_{\traceos|\tr_i} \res{i}{}{\co_i}$.
    Hence $\op_i \instigates e_u$. $\tr_i$ executes \ppr{release} only
    following issuing $e_u$ \ppr{wb:release}.
\end{proof}

\begin{lemma} [Obbligato] \label{lemma:obbligato}
    Trace $\traceos$ is obbligato.
\end{lemma}

\begin{proof}
    From Lemmas \ref{lemma:comitted-write-obbligato},
    \ref{lemma:closing-write-obbligato}, and
    \ref{lemma:view-write-obbligato}.
\end{proof}

\begin{lemma} [Decisiveness] \label{lemma:decisiveness}
    Trace $\traceos$ is decisive.
\end{lemma}

\begin{proof}
    If $\tr_j \vviews \tr_i$, then for any $\obj \in \accesses{i} \cap
    \accesses{j}$, $\pv{\obj}{i} < \pv{\obj}{j}$ (\rlemma{lemma:isolation}).
    Before any view event occurs, $\tr_j$ must pass the condition $\pv{\obj}{k}
    - 1 = \lv{\obj}$ in \ppr{read\_buffer} or \ppr{checkpoint}.
    Hence, before any $e_v = \get{j}{\obj}{\val}$ can occur, some $\tr_k$ s.t.
    $\pv{\obj}{i} - 1 =\pv{\obj}{k}$ must set $\pv{\obj}{k}$ to $\lv{\obj}$.
    Transaction $\tr_j$ issues a routine update event $e_u =
    \rset{j}{\obj}{\val'}$, whenever it commits or releases $\obj$.
    If it releases, it means that $\wc{j}{\obj} = \wub{j}{\obj}$, which implies
    that $\op_i$ is closing.
    Otherwise, $\tr_j$ will update on commit, meaning that it will issue
    $\res{j}{}{\co_j}$ once it returns from the \ppr{commit} procedure.
    Before returning from the closing write or commit, $\tr_j$ sets
    $\lv{\obj}$ to $\pv{\obj}{j}$. In either case, this happens only afterward $e_u$ is issued.
    Since there is no waiting between $e_u$ and either a commit or a last write
    returning, no other transaction may execute anything on $\obj$ in the
    meantime.
    Thus, any transaction $\tr_k$ s.t. $\pv{\obj}{k}$ that waits until
    $\pv{\obj}{k} - 1 = \lv{\obj}$ and $\pv{\obj}{k} - 1 = \pv{\obj}{j}$ will
    wait until $\tr_j$ returns from the closing write or commit procedure, and
    so will any subsequent transactions according to version order.
    Thus, if $\pv{\obj}{i} < \pv{\obj}{j}$ and $\tr_j$ commits, then either
    $\op_j$ is closing or $e_u \prec_\traceos \res{j}{}{\co_j} \prec e_v$.
\end{proof}

\begin{lemma} [Abort Accord] \label{lemma:abort-accord} 
    Trace $\traceos$ is in abort accord.
\end{lemma}

\begin{proof}
    Let $\tr_i, \tr_j$ be two transactions in $\traceos$ s.t. 
    \begin{enumerate}[a)]
    \item %
        $\tr_j \views \tr_i$, $\tr_i$ is aborted in $\traceos$.
        Assume by contradiction that $\tr_j$ commits in $\traceos$, meaning it
        executes \ppr{commit} successfully. Thus it passes $\forall \obj,
        \cv{\obj} \geq \rv{j}{\obj}$.

        Since $\tr_j \views \tr_i$, $\exists e_v = \get{j}{\obj}{\val} \in
        \traceos|\tr_j$ and $\exists e_u = \rset{i}{\obj}{\val} \in
        \traceos|\tr_i$.

        If $\tr_i$ aborted before $e_v$ was issued, then from abort coda,
        $\exists e_a = \aset{k}{\obj}{\any}$ in some $\tr_k$ that precedes the
        abort, which contradicts that $\tr_j \views \tr_i$. Hence $\tr_i$ aborts
        only after $e_v$ is issued.

        Since $\tr_j \views \tr_i$, then $\tr_i \iso{\traceos}{\obj} \tr_j$, so
        from \rcor{cor:lambda}, $\pv{\obj}{i} < \pv{\obj}{j}$. From
        \rlemma{lemma:gamma}, $e_v$ cannot occur until $\tr_i$ aborts,
        commits, or releases $\obj$. Since $\tr_i$ aborts after $e_v$, then it
        must therefore release $\obj$ prior to $e_v$.

        Since $\pv{\obj}{i} < \pv{\obj}{j}$, then from \rlemma{lemma:kappa},
        $\rv{i}{\obj} < \rv{j}{\obj}$. When $\tr_i$ aborts, it sets $\cv{\obj}$
        to $\rv{i}{\obj}$. From coherence, $\tr_j$ commits after $\tr_i$
        aborts. Thus, when $\tr_j$ commits, $\cv{\obj} = \rv{i}{\obj}$, so
        since $\rv{i}{\obj} < \rv{j}{\obj}$, then $\cv{\obj} < \rv{j}{\obj}$,
        which contradicts the condition that $\cv{\obj} > \rv{j}{\obj}$. 
        
        Thus $\tr_j$ cannot commit.

    \item %
        $\exists e_u = \rset{i}{\obj}{\any} \in \traceos|\tr_i$ and     $e =
        \rset{j}{\obj}{\any} \in \traceos|\tr_j$ or $e = \get{j}{\obj}{\any}
        \in \traceos|\tr_j$ and     $e_a = \aset{i}{\obj}{\any} \in
        \traceos|\tr_i$, and     $e_u \prec_\traceos e \prec_\traceos e_a$.
        Assume by contradiction that $\tr_j$ commits in $\traceos$, meaning it
        executes \ppr{commit} successfully. Thus it passes $\forall \obj,
        \cv{\obj} \geq \rv{j}{\obj}$.

        Since $e_u \prec_\traceos e$, then from isolation it follows that
        $\tr_i \iso{\traceos}{\obj} \tr_j$.  Hence, from \rcor{cor:lambda},
        $\pv{\obj}{i} < \pv{\obj}{j}$.  Since $e_a$ must be issued during
        \ppr{abort}, then from coherence, $\tr_j$ cannot commit prior to $e_a$
        occurring.  Furthermore, $\tr_j$ cannot commit until $\tr_i$ returns
        from \ppr{abort}.

        If  $\tr_i$ returned from \ppr{abort}, then it executed
        \rln{abort:recover-cv}, so $\cv{\obj} = \rv{i}{\obj}$ prior to $\tr_j$
        committing.

        From \rlemma{lemma:kappa}, since $\pv{\obj}{i} < \pv{\obj}{j}$, then
        $\rv{i}{\obj} < \rv{j}{\obj}$. When $\tr_i$ aborts, it sets $\cv{\obj}$
        to $\rv{i}{\obj}$. From coherence, $\tr_j$ commits after $\tr_i$
        aborts. Thus, when $\tr_j$ commits, $\cv{\obj} = \rv{i}{\obj}$, so
        since $\rv{i}{\obj} < \rv{j}{\obj}$, then $\cv{\obj} < \rv{j}{\obj}$,
        which contradicts the condition that $\cv{\obj} > \rv{j}{\obj}$. 
        
        Thus $\tr_j$ cannot commit.
    \end{enumerate}
\end{proof}

\begin{lemma} [Commit Accord] \label{lemma:commit-accord}
    Trace $\traceos$ is in commit accord.
\end{lemma}

\begin{proof}
    Let $\tr_i, \tr_j$ be transaction in $\trace$ s.t. $\tr_j \views \tr_i$ and
    $\tr_j$ is committed in $\traceos$.
    
    Let us assume by contradiction that $\tr_i$ is not committed in $\traceos$.
    So $\tr_i$ is either aborted or live in $\traceos$. From coherence, if
    $\tr_j$ cannot commit until $\tr_i$ commits or aborts. Thus $\tr_i$ is not
    live in $\traceos$, so it is aborted in $\traceos$.

    Since $\tr_j \views \tr_i$, $\exists e_v = \get{j}{\obj}{\val} \in
    \traceos|\tr_j$ and $\exists e_u = \rset{i}{\obj}{\val} \in
    \traceos|\tr_i$.

    If $\tr_i$ aborted before $e_v$ was issued, then from abort coda,
    $\exists e_a = \aset{k}{\obj}{\any}$ in some $\tr_k$ that precedes the
    abort, which contradicts that $\tr_j \views \tr_i$. Hence $\tr_i$ aborts
    only after $e_v$ is issued.

    Since $\tr_j \views \tr_i$, then $\tr_i \iso{\traceos}{\obj} \tr_j$, so
    from \rcor{cor:lambda}, $\pv{\obj}{i} < \pv{\obj}{j}$. From
    \rlemma{lemma:gamma}, $e_v$ cannot occur until $\tr_i$ aborts,
    commits, or releases $\obj$. Since $\tr_i$ aborts after $e_v$, then it
    must therefore release $\obj$ prior to $e_v$.

    Since $\pv{\obj}{i} < \pv{\obj}{j}$, then from \rlemma{lemma:kappa},
    $\rv{i}{\obj} < \rv{j}{\obj}$. When $\tr_i$ aborts, it sets $\cv{\obj}$
    to $\rv{i}{\obj}$. From coherence, $\tr_j$ commits after $\tr_i$
    aborts. Thus, when $\tr_j$ commits, $\cv{\obj} = \rv{i}{\obj}$, so
    since $\rv{i}{\obj} < \rv{j}{\obj}$, then $\cv{\obj} < \rv{j}{\obj}$,
    which contradicts the condition that $\cv{\obj} > \rv{j}{\obj}$.         

    Thus $\tr_i$ cannot abort.

\end{proof}

\begin{lemma} [Abort Coda] \label{lemma:abort-coda}
    Trace $\traceos$ has coda.
\end{lemma}

\begin{proof}
    \begin{enumerate}[a)]
        \item
            If $\tr_i$ aborts in $\traceos$ (so $r = \res{i}{}{\ab_i} \in
            \traceos|\tr_i$) then if $\exists e_u = \rset{\obj}{\val} \in
            \traceos|\tr_i$ then for some $\tr_j$ (possibly $i=j$) s.t. $j=i$
            or $\tr_j \iso{\obj}{\traceos} \tr_i$, $\exists e_a =
            \aset{j}{\obj}{\any} \in \traceos|\tr_j$ s.t. $e_u \prec_\traceos
            e_a \prec_\traceos r$.

            If there is such $e_u$, then $\tr_i$ executes \ppr{checkpoint}
            (\rlemma{lemma:epsilon}).  
            Since there is such $e_u$ there is also a write operation execution
            on $\obj$ in $\traceos|\tr_i$
            (\rlemma{lemma:routine-update-consonance}), so $\wc{i}{\obj} > 0$
            (\rln{write-wc-inc}).
            
            If there is such $e_u$, then $\tr_i$ executes \ppr{release} for
            $\obj$ or \ppr{commit}. Since $\tr_i$ aborts, then \ppr{commit} is
            not possible, so $\tr_i$ executes \ppr{release} for $\obj$.
            Therefore $\tr_i$ sets $\lv{\obj}$ to $\pv{i}{\obj} - 1$.

            \begin{enumerate}[i)]
                \item If no other transaction modified $\cv{\obj}$ between the
                    point at which $\tr_i$ executed \ppr{checkpoint} and
                    \ppr{abort}, then $\cv{\obj} = \rv{i}{\obj}$, thus during
                    abort $\tr_i$ satisfies the condition on \rln{abort-ro} and
                    executes \rln{abort:aset}, issuing the recovery event
                    $e_a$. Since $e_a$ is issued during \ppr{abort}, then $e_u
                    \prec_\traceos e_a \prec_\traceos r$.

                \item If there is $\tr_j$ s.t. $\tr_j$ modifies $\cv{\obj}$
                    between the points at which $\tr_i$ executed
                    \ppr{checkpoint} and \ppr{abort}, s.t. $\tr_j
                    \iso{\obj}{\traceos} \tr_i$, then $\tr_j$ executes
                    \ppr{release}, \ppr{commit}, or \ppr{abort} between the
                    points at which $\tr_i$ executed \ppr{checkpoint} and
                    \ppr{abort}. For the sake of simplicity we assume that
                    there is no other $\tr_k$ that modifies $\cv{\obj}$ between
                    those two points s.t. $\tr_k \iso{\obj}{\traceos} \tr_i$.

                    Since $\tr_i$ executes checkpoint, it issues a view event
                    $e_v = \get{i}{\obj}{\any}$. In addition, since $\tr_j
                    \iso{\obj}{\traceos} \tr_i$, then from
                    \rcor{cor:lambda}, $\pv{\obj}{j} < \pv{\obj}{i}$.
                    From \rlemma{lemma:gamma}, $e_v$ cannot occur until $\tr_j$
                    aborts, commits, or releases $\obj$. Since $\tr_j$ is
                    supposed to execute \ppr{release}, \ppr{abort}, or
                    \ppr{commit} after $\tr_i$ executes \ppr{checkpoint}, hence
                    after $e_v$, then $\tr_j$ must therefore release $\obj$
                    prior to $e_v$. Hence  $\tr_j$ executes
                    \ppr{commit}, or \ppr{abort} between the
                    points at which $\tr_i$ executed \ppr{checkpoint} and
                    \ppr{abort}.

                    If $\tr_i$ executes \ppr{commit}, then in order to set
                    $\cv{\obj}$ to $\pv{\obj}{j}$, it must be true that
                    $\rv{i}{\obj} = \cv{\obj}$. But if $\tr_j$ executed
                    \ppr{release}, then $\cv{\obj}$ is set to $\pv{\obj}{j}$.
                    Since $\rv{\obj}{j}\neq \pv{\obj}{j}$, then $\rv{i}{\obj}
                    \neq \cv{\obj}$, so $\tr_j$ cannot set $\cv{\obj}$ as a
                    result of a commit. Hence, $\tr_j$ executes
                    \ppr{abort} between the
                    points at which $\tr_i$ executed \ppr{checkpoint} and
                    \ppr{abort}.

                    If $\tr_j$ executes \ppr{abort}, then this implies that
                    $\tr_j$ executes \rln{abort:recover-cv}, and therefore also
                    \rln{abort:aset}, thus $tr_j$ issues recovery event $e_a$
                    during \ppr{abort}. Thus, $e_u \prec_\traceos e_a
                    \prec_\traceos r$.
                \end{enumerate}

        \item 
            If $\tr_i$ commits in $\traceos$ (so $r = \res{i}{}{\co_i} \in
            \traceos|\tr_i$) then if $\exists e = \rset{\obj}{\val} \in
            \traceos|\tr_i$ or $e = \get{\obj}{\val} \in
            \traceos|\tr_i$ then for no $\tr_j$ s.t. $j=i$ or $\tr_j
            \iso{\obj}{\traceos} \tr_i$, $\exists e_a = \aset{j}{\obj}{\any}
            \in \traceos|\tr_j$ s.t. $e_u \prec_\traceos e_a \prec_\traceos r$.

            If there is such $e$, then $\tr_i$ executes \ppr{checkpoint}
            (\rlemma{lemma:epsilon}). 
            If $\tr_i$ successfully commits, then means that $\tr_i$ passes the
            condition for $\obj$ that $\cv{\obj} \geq \rv{i}{\obj}$.

            Assume by contradiction that there is such $e_a$ in some $\tr_j$.
            Since $e_a \in \trace|\tr_j$, then $\tr_j$ must execute
            \rln{abort:aset}, which also means that it executes
            \rln{abort:recover-cv} and therefore sets $\cv{\obj}$ to $\rv{j}{\obj}$.

            Since $\tr_j \iso{\obj}{\traceos} \tr_i$, then from
            \rcor{cor:lambda} $\pv{\obj}{j} < \pv{\obj}{i}$ and from \rlemma{lemma:kappa},
            $\rv{j}{\obj} < \rv{i}{\obj}$.
            Thus, $\cv{\obj} < \rv{i}{\obj}$ which contradicts that $\cv{\obj}
            \geq \rc{i}{\obj}$. Thus there is no such $\tr_j$.
    \end{enumerate}
\end{proof}

\begin{lemma} [Coherence] \label{lemma:coherence}
    Trace $\traceos$ is coherent.
\end{lemma}

\begin{proof}
    If $\tr_i \iso{\obj}{\traceos} \tr_j$, then $\pv{\obj}{i} < \pv{\obj}{j}$.
    In order to commit or abort, any $\tr_k$ must pass the condition
    $\pv{\obj}{k} - 1 = \ltv{\obj}$. In addition, each $\tr_k$ sets
    $\ltv{\obj}$ to $\pv{\obj}{k}$ only at the end of either committing or
    aborting. Hence, if $\tr_k$ cannot commit or abort until some $\tr_l$ s.t.
    $\pv{\obj}{k} - 1 = \pv{\obj}{l}$ finishes committing or aborting.
    Hence if $\tr_j$ committed, it must have passed the condition $\pv{\obj}{k}
    - 1 = \ltv{\obj}$, and since $\pv{\obj}{i} < \pv{\obj}{j}$, $\tr_i$ must
    have committed or aborted before $\tr_j$ committed.
    Thus, given $r_j = \res{j}{}{\co_j} \in \traceos|\tr_j$, then there is $r_i
    = \res{i}{}{\co_i} \in \traceos|\tr_i$ or $r_i = \res{i}{}{\ab_i} \in
    \traceos|\tr_i$ and $r_i \prec_\trace r_j$.
\end{proof}

\begin{lemma} [Chain Isolation] \label{lemma:chain-isolation}
    Given trace $\traceos$ and transactions $\tr_i, \tr_j \in \traceos$ s.t.
    there is $\vchain{\trace}{i}{j}$, $\forall \tr_k \in \vchain{\trace}{i}{j}$
    s.t. $e^k_u = \rset{k}{\obj}{\val}$, there is no $\tr_l$ s.t. $\exists e^l_a =
    \aset{l}{\obj}{\val'}$ where $\val = \val'$ and $e^l$ is between $e^k$ and
    any other event in any transaction in $\vchain{\trace}{i}{j}$.
\end{lemma}

\begin{proof}
    Assume by contradiction that there exists $\tr_l$ such that $\exists e^l =
    \aset{l}{\obj}{\val'}$ and $e^l_a$ is between $e^k_u$ and any other event
    $e \in \traceos|\vchain{\traceos}{i}{j}$. 
    This means that either $e \in \traceos|\tr_l$ and $e^k_u
    \prec_\traceos e^l_a \prec_\traceos e$, or  $\exists \tr_n$ s.t. $e \in
    \traceos|\tr_n$.

    \begin{enumerate}[a) ]
        \item Assume $e \in \traceos|\tr_l$ and $e^k_u \prec_\traceos e^l_a \prec_\traceos e$.

            From \rlemma{lemma:thorn} there is a view event $e_v^k  =
            \get{k}{\obj}{\any} \in \traceos|\tr_k$, and from minimalism there
            is only one such event in $\traceos|\tr_k$, so $e$ must be a
            recovery event $e = \aset{k}{\obj}{\val}$.
            If $\tr_l$ executes $e_a^l$, then from \rlemma{lemma:needed},
            $\exists e_u^l = \rset{l}{\obj}{\any}$ in $\traceos|\tr_l$ s.t.
            $e_u^k \prec_\traceos e_a^l$. Hence either $e_u^l \prec_\traceos
            e_u^k$ or $e_u^k \prec_\traceos e_u^l$. So either $\tr_l
            \iso{\traceos}{\obj} \tr_k$ or $\tr_k \iso{\traceos}{\obj} \tr_l$.

            If $\tr_l \iso{\traceos}{\obj} \tr_k$, then $e_u^l \prec_\traceos
            e_u^k$ , so from \rlemma{lemma:gamma}, $e_u^k$ cannot occur until
            $e_u^l$ executes \ppr{release}, \ppr{commit}, or \ppr{abort}, and
            since $e_u^k \prec_\traceos e_a^l$, then only \ppr{release} is
            viable.

            If $\tr_l \iso{\traceos}{\obj} \tr_k$, then from version order
            $\pv{\obj}{l} < \pv{\obj}{k}$, so from \rlemma{lemma:kappa},
            $\rv{l}{\obj} < \rv{k}{\obj}$. In order for $\tr_l$ to issue $e_a$,
            it must execute \ppr{abort} and satisfy the condition
            \rln{abort:recover-ro}. This means that \rln{abort:recover-cv} is
            executed, so $\cv{\obj} = \rv{l}{\obj}$.

            If subsequently $\tr_k$ issues $e_a$, then it must also satisfy the
            condition at \rln{abort:recover-ro}, so $\rv{k}{\obj} = \cv{\obj}$.
            But since $\rv{l}{\obj} < \rv{k}{\obj}$, then $\cv{\obj} <
            \rv{k}{\obj}$, which contradicts that $\cv{\obj} = \rv{l}{\obj}$.
            The execution of a recovery event on $\obj$ by $\tr_l$ is
            dominated by \rln{abort-commit}, which cannot be passed until
            $\pv{\obj}{l} - 1 = \ltv{\obj}$. Any transaction $\tr_n$ sets
            $\ltv{\obj}$ to $\pv{\obj}{n}$ as a last action during \ppr{commit}
            (\rln{commit-ltv}) or \ppr{abort} (\rln{abort-ltv}).  Hence $\tr_l$
            cannot proceed to abort until $\tr_n$ finishes committing or
            aborting. 
            Since $\tr_n$ cannot execute \rln{commit-ltv} or \rln{abort-ltv} if
            \rln{commit-commit} or \rln{abort-commit} was passed, then $\tr_n$
            cannot proceed to commit or abort until some other $\tr_m$ s.t
            $\pv{\obj}{n} - 1 = \pv{\obj}{m}$ committed or aborted.
            Hence $\tr_l$ cannot execute a recovery event until any $\tr_m$
            s.t. $\pv{\obj}{m} < \pv{\obj}{l}$ committed or aborted.

            If $\tr_k \iso{\traceos}{\obj} \tr_l$, then from version order
            $\pv{\obj}{k} < \pv{\obj}{l}$. 
            Hence, $\tr_l$ executes any events \ppr{abort} only after $\tr_k$ returns from
            \ppr{abort} or \ppr{commit}. Hence if $\tr_k$ executes $e_a^k$ ,
            then $e_a^k \prec_\traceos e_a^l$, which contradicts that $e_a^l
            \prec_\traceos e_a^k$.

            Thus, regardless of whether $\tr_l \iso{\traceos}{\obj} \tr_k$ or
            $\tr_k \iso{\traceos}{\obj} \tr_l$ \underline{there is a
            contradiction}.
            Therefore, $\tr_l$ cannot issue such $e_a^l$ between $e_u^k$ and
            another event in $\traceos|\tr_l$.

        \item Assume $\exists \tr_n$ s.t. $e \in \traceos|\tr_n$. 

            We assume without loss of generality that $\tr_n \views \tr_k$.
            Thus, there is a view event $e_v^n$ and possibly a routine update
            event $e_u^n$ in $\traceos|\tr_n$. From minimalism and
            \rlemma{lemma:thorn}: $e_v^n \prec_\traceos e_u^n$. Also, since
            $\tr_n \views \tr_k$, then $\tr_k \iso{\traceos}{\obj} \tr_n$, so
            from \rcor{cor:lambda}, $\pv{\obj}{k} < \pv{\obj}{n}$.

            If $\tr_l$ executes $e_a^l$, then from \rlemma{lemma:needed},
            $\exists e_u^l = \rset{l}{\obj}{\any}$ in $\traceos|\tr_l$ s.t.
            $e_u^k \prec_\traceos e_a^l$. Hence either $e_u^l \prec_\traceos
            e_u^k$ or $e_u^k \prec_\traceos e_u^l$. So either $\tr_k
            \iso{\traceos}{\obj} \tr_l \iso{\traceos}{\obj} \tr_n$ or $\tr_k
            \iso{\traceos}{\obj} \tr_n \iso{\traceos}{\obj} \tr_l$.
            Thus, from \rcor{cor:lambda}, either $\pv{\obj}{k} <
            \pv{\obj}{l} < \pv{\obj}{n}$ or $\pv{\obj}{k} < \pv{\obj}{n} <
            \pv{\obj}{l}$.

            If $\pv{\obj}{k} < \pv{\obj}{l} < \pv{\obj}{n}$, then either $e_a^l
            \prec_\traceos e_v^n$ or $e_v^n \prec_\traceos e_a^l$.

            If $e_a^l \prec_\traceos e_v^n$, then since $e_a^l$ sets $\obj$ to
            $\val'$ s.t. $\val' \neq \val''$ for any $\val''$ s.t. $\exists
            \rset{l}{\obj}{\val''} \in \traceos|\tr_l$ prior to the occurrence
            of $e_v^n$. Thus when $\tr_n$ subsequently executes
            \ppr{checkpoint} it issues $e_v^n = \get{n}{\obj}{\val'}$, and
            since $\val' \neq \val''$, this contradicts that $\tr_n \views
            \tr_k$.

            If $e_v^n \prec_\traceos e_a^l$, then since $\tr_l
            \iso{\traceos}{\obj} \tr_n$m then from version order
            $\pv{\obj}{l} < \pv{\obj}{n}$, so from \rlemma{lemma:kappa},
            $\rv{l}{\obj} < \rv{n}{\obj}$. In order for $\tr_n$ to issue $e_a$,
            it must execute \ppr{abort} and satisfy the condition
            \rln{abort:recover-ro}. This means that \rln{abort:recover-cv} is
            executed, so $\cv{\obj} = \rv{l}{\obj}$.

            If subsequently $\tr_l$ issues $e_u^n$ then it either executes
            \ppr{write\_buffer} or \ppr{commit}.
            Issuing an update event at \rln{write-buffer:rset} or
            \rln{commit:rset} is dominated by checking whether $\cv{\obj} =
            \rv{n}{\obj}$ (for all variables) at \rln{wb-abort} or
            \rln{commit-abort-1}, respectively. If the condition is failed, the
            transaction aborts instead.
            From \rlemma{lemma:kappa}, $\rv{l}{\obj} < \rv{n}{\obj}$, so if
            $\cv{\obj} = \rv{l}{\obj}$, then $\cv{\obj} \neq \rv{n}{\obj}$.
            Hence, $\tr_i$ will \ppr{abort} rather than issue an update event.
            Since during \ppr{abort} only a recovery event may be issued, and
            only if $\rv{n}{\obj} = \cv{\obj}$, then, similarly, no recovery
            event is issued. Hence $\tr_n$ cannot issue events on $\obj$
            following $e_a^l$.

            Since each occurrence of a routine update event or a view event
            checks $\forall{\objy}, \cv{\objy} = \rv{n}{\objy}$, then no other
            such event in $\tr_n$ can follow $e_a^l$. This is a contradiction.
            The execution of a recovery event on $\obj$ by $\tr_l$ is
            dominated by \rln{abort-commit}, which cannot be passed until
            $\pv{\obj}{l} - 1 = \ltv{\obj}$. Any transaction $\tr_n$ sets
            $\ltv{\obj}$ to $\pv{\obj}{o}$ as a last action during \ppr{commit}
            (\rln{commit-ltv}) or \ppr{abort} (\rln{abort-ltv}).  Hence $\tr_l$
            cannot proceed to abort until $\tr_o$ finishes committing or
            aborting. 
            Since $\tr_o$ cannot execute \rln{commit-ltv} or \rln{abort-ltv} if
            \rln{commit-commit} or \rln{abort-commit} was passed, then $\tr_o$
            cannot proceed to commit or abort until some other $\tr_m$ s.t
            $\pv{\obj}{o} - 1 = \pv{\obj}{m}$ committed or aborted.
            Hence $\tr_l$ cannot execute a recovery event until any $\tr_m$
            s.t. $\pv{\obj}{m} < \pv{\obj}{l}$ committed or aborted.

            If $\pv{\obj}{k} < \pv{\obj}{n} < \pv{\obj}{l}$, from version order
            $\pv{\obj}{n} < \pv{\obj}{l}$, so
            $\tr_l$ executes \ppr{abort} only after $\tr_k$ returns from
            \ppr{abort} or \ppr{commit}. Hence, since $e_u^n \prec_\traceos
            \res{n}{}{\ab_n}$ or  $e_u^n \prec_\traceos \res{n}{}{\co_n}$,
            either $\res{n}{}{\ab_n} \in \traceos|\tr_n$ or $\res{n}{}{\co_n}
            \in \traceos|\tr_n$, and since $\res{n}{}{\co_n} \prec_\traceos
            e_a^l$ or  $\res{n}{}{\ab_n} \prec_\traceos e_a^l$, then $e_u^n
            \prec_\traceos e_a^l$. This contradicts that  $e_a^l \prec_\traceos
            e_u^n$.

            Thus, regardless of whether \underline{there is a
            contradiction}.
            Therefore, $\tr_l$ cannot issue such $e_a^l$ between $e_u^k$ and
            another event in $\traceos|\tr_n$.
            By extension, it cannot issue $e_a^l$ between $e_u^k$ and another
            event in $\traceos|\tr_m$ for any $\tr_m \in
            \vchain{\traceos}{i}{j}$.
    \end{enumerate}
\end{proof}

\begin{lemma} [Chain Self-containment] \label{lemma:chain-self-containment}
    Given $\trace$, and any transactions $\tr_i, \tr_j \in \trace$, s.t.
    $\exists \vchain{\trace}{i}{j}$, $ \vchain{\trace}{i}{j}$ is
    self-contained.
\end{lemma}

\begin{proof}

    Given transaction $\tr_q$ s.t. $\exists e_v^q = \get{q}{\obj}{\val^q} \in \traceos|\tr_q$, 
    assuming that there are any update events on $\obj$ in $\traceos$ prior to $e_v^q$, 
    then $\exists e_u^r = \rset{r}{\obj}{\val^q} \in \traceos|\tr_r$ where $e_u^r \pref e_v^q$ 
    or $\exists e_a^r = \aset{r}{\obj}{\val^q} \in \traceos|\tr_r$ where $e_a^r \pref e_v^q$ 
    for some $\tr_r$.
    In addition, 
    from \rlemma{lemma:thorn}, $\exists e_v^r = \get{r}{\obj}{\val^r} \in \traceos|\tr_r$ s.t.
    $e_v^r \prec_\traceos e_u^r$ and $e_v^r \prec_\traceos e_a^r$ (as
    applicable).

    Then, similarly, 
    assuming that there are any update events on $\obj$ in $\traceos$ prior to $e_v^r$, 
    then $\exists e_u^s = \rset{s}{\obj}{\val^r} \in \traceos|\tr_s$ where $e_u^s \pref e_v^r$ 
    or $\exists e_a^s = \aset{s}{\obj}{\val^r} \in \traceos|\tr_s$ where $e_a^s \pref e_v^r$.
    And by analogy to $\tr_r$, 
    from \rlemma{lemma:thorn}, $\exists e_v^r = \get{r}{\obj}{\val^r}$ s.t.
    $e_v^r \prec_\traceos e_u^r$ and $e_v^r \prec_\traceos e_a^r$ (as
    applicable).

    It is then clear that as long as there are update events on $\obj$
    preceding a view event in some transaction, another transaction exists
    that both views and updates $\obj$ before that view event.

    Thus, given $\tr_k$ and $\tr_l \in \vchain{\trace}{i}{j}$ such that
    $k\neq l$ and $\exists e_u^k = \rset{k}{\obj}{\val} \in \trace|\tr$
    $\exists e_v^l = \get{l}{\obj}{\val'} \in \trace|\tr$ and $e_u^k
    \prec_\trace e_v^l$, there is a sequence of transactions $\mathscr{S}$ s.t.:

    \begin{enumerate}
        \item the first transaction is $\tr_k$,
        \item the last transaction is $\tr_l$, and
        \item given some transaction $\tr_m \in \mathscr{S}$, where $m\neq k$, 
            $\tr_m$ is preceded in $\mathscr{S}$ by some transaction $\tr_n$, s.t. 
            for $e_v^m = \get{m}{\obj}{\val^m} \in \traceos|\tr_m$,
            $\exists e_u^n = \rset{n}{\obj}{\val^m} \in \traceos|\tr_n$ where $e_u^n \pref e_v^m$ 
            or $\exists e_a^n = \aset{n}{\obj}{\val^m} \in \traceos|\tr_n$ where $e_a^n \pref e_v^m$.
    \end{enumerate}

    Given such $\mathscr{S}$, given some $\tr_m$, $m\neq k$, there is some
    $\tr_n$ that precedes $\tr_m$ in $\mathscr{S}$.
    
    If for  for $e_v^m = \get{m}{\obj}{\val^m} \in \traceos|\tr_m$,
    $\exists e_u^n = \rset{n}{\obj}{\val^m} \in \traceos|\tr_n$ where $e_u^n \pref e_v^m$, 
    then $\tr_m \views \tr_n$.

    If, on the other hand, 
    for $e_v^m = \get{m}{\obj}{\val^m} \in \traceos|\tr_m$,
    or $\exists e_a^n = \aset{n}{\obj}{\val^m} \in \traceos|\tr_n$ where $e_a^n \pref e_v^m$,
    then from chain isolation there cannot be a recovery event $e_a^k =
    \aset{k}{\obj}{\any}$ s.t. $e_a^k \prec_\traceos e_v^l$, so it
    follows that $k\neq n$.
    Since $e_a^n$ is conservative, $\exists e_v^n = \get{n}{\obj}{\val^m} \in
    \traceos|\tr_n$.

    Since $k\neq n$ and $\tr_n \in \mathscr{S}$, then there is some
    $\tr_o$ preceding $\tr_n$in $\mathscr{S}$.
    Then:

    \begin{enumerate}[a) ]
        \item If $o = k$, then $\tr_m \vviews \tr_k$.
        \item If $o \neq k$ and  
            $\exists e_u^n = \rset{o}{\obj}{\val^m} \in \traceos|\tr_o$ where $e_u^o \pref e_v^n$
            then $\tr_m \vviews \tr_o$.
        \item If $o \neq k$ and  
            $\exists e_a^n = \aset{o}{\obj}{\val^m} \in \traceos|\tr_o$ where $e_a^o \pref e_v^n$
            then by analogy, either case a), b) or c) applies to $\tr_o$ as it does to $\tr_n$. 
            So, by analogy, either
            \begin{inparaenum}[a)] 
                \item $\tr_m \vviews \tr_k$,
                \item $\tr_m \vviews \tr_o$, or
                \item there is another preceding transaction in $\mathscr{S}$, etc.
            \end{inparaenum}

            Note, however, that since $\mathscr{S}$ is finite, and $e_a^k$
            cannot precede $e_v^l$ in $\traceos$, then eventually for some such
            preceding $\tr_q \in \mathscr{S}$ case a) or b) and not c) will
            apply. Thus, there will be some $\tr_q \in \mathscr{S}$ s.t. $\tr_m \vviews \tr_q$
            (where either $q=k$ or $q\neq k$).
    \end{enumerate}

    Therefore, $\forall \tr_m \in \mathscr{S}$ s.t. $m\neq k$, $\exists \tr_n
    \in \mathscr{S}$ s.t. $\tr_m \vviews \tr_n$.

    In addition, for each such pair $\tr_m, \tr_n$ , there is therefore
    $\vchain{\traceos}{n}{m}$. Furthermore, if $\tr_m \vviews \tr_n$ and for
    some other $\tr_o$, $\tr_n \vviews \tr_o$, then there is
    $\vchain{\traceos}{o}{m}$. Thus, there is also $\vchain{\traceos}{k}{l}$,
    such that if $\tr_l \vviews \tr_m$ and $\tr_m \in \mathscr{S}$, then $\tr_m
    \in \vchain{\traceos}{k}{l}$. Since $\tr_k, \tr_l \in
    \vchain{\traceos}{i}{j}$, then if $\tr_m \in \vchain{\traceos}{k}{l}$,
    $\tr_m \in \vchain{\traceos}{i}{j}$.

    If $\mathscr{S} = \tr_k \cdot \tr_l$ then trivially $\val = \val'$.

    Otherwise, since $\tr_l \in \mathscr{S}$ and $l\neq k$, then $\exists \tr_m \in \mathscr{S}$ s.t. $\tr_l \vviews \tr_m$, so
      $\exists e_u^m = \rset{m}{\obj}{\val'} \in
      \trace|\tr_m$ for some $\tr_m \in \vchain{\trace}{i}{j}$ s.t. 
      $\tr_m$ precedes $\tr_l$ and follows $\tr_k$ in $\vchain{\trace}{i}{j}$ and
      $e_u^k \prec_\trace e_u^m \prec_\trace e_v^l$.
\end{proof}

\begin{lemma} [Tree Chain Consistency] \label{lemma:chain-consistency}
    Trace $\traceos$ is chain consistent.
\end{lemma}

\begin{proof}
    From \rlemma{lemma:chain-isolation} and \rlemma{lemma:chain-self-containment}.
\end{proof}

\begin{lemma} [Trace Harmony] \label{lemma:harmony}
    Trace $\traceos$ is harmonious.
\end{lemma}

\begin{proof}
    Trace $\traceos$ satisfies all of the following:
    \begin{inparaenum}[a)]
        \item minimalism from \rlemma{lemma:minimalism}
        \item consonance from \rlemma{lemma:consonance}, 
        \item obbligato from \rlemma{lemma:obbligato}, 
        \item coherence, commit accord, abort accord, and abort coda from Lemmas \ref{lemma:coherence}, \ref{lemma:commit-accord}, \ref{lemma:abort-accord}, and \ref{lemma:abort-coda},
        \item isolation from \rlemma{lemma:isolation}, 
        \item decisiveness from \rlemma{lemma:decisiveness},
        \item chain consistency from \rlemma{lemma:chain-consistency},
        \item unique writes (assumed).
    \end{inparaenum}
\end{proof}

\begin{corollary}
    History $\hist = \tohist\traceos$ is last-use opaque. In consequence OptSVA is last-use opaque.
\end{corollary}

\section{Conclusions}
\label{sec:conclusions}

The paper presents OptSVA, a highly optimized pessimistic TM that uses a number
of techniques to improve the length of interleavings produced by conflicting
transactions without losing the guarantees of the algorithm it was based on.
Some of these techniques are straightforward, like the read-write distinction
and buffering, but were sorely lacking in versioning algorithms thus far. More
importantly, using dedicated executor threads to do transactions' waiting for
them, so they can perform other computations in the mean time, is a novel
technique in the context of TM. The agglomeration of these techniques creates a
new algorithm that performs well in comparison to other versioning algorithms, 
both in theory, as well as in practice. 
On the other hand, OptSVA maintains the advantages of pessimistic TM with
regards to irrevocable operations, since it only aborts when an abort is
manually invoked.

In addition to the algorithm itself, the paper also presents trace harmony, a
proof technique that can be used for OptSVA and other buffered algorithms that
concentrate on memory accesses over transactional API operations to demonstrate
last-use opacity. Moreover, the proof technique and can be trivially modified to
demonstrate opacity (and related properties). This requires applying the tests
that we apply to events in committed transactions to all other transactions
uniformly.

Future work on OptSVA includes a comprehensive evaluation comparing it against
state-of-the-art TM algorithms. However, even though OptSVA can be implemented
successfully as a multicore system, we consider its intended application to be
in distributed TM, and such an implementation is out of scope of this paper,
even though, the results of the evaluation of such a system are promising.

\paragraph{Acknowledgments}
\mygrant

\bibliographystyle{abbrv}
\bibliography{top-arxiv}

\clearpage
\appendix
\section{Last-use Opacity from Harmony}
\label{sec:harmony-to-last-use-opacity}

\subsection{Last-use Opacity}

Given program $\prog$ and a set of processes $\processes$
executing $\prog$, since different interleavings of $\processes$ cause an
execution $\exec{\prog}{\processes}$ to produce different histories, then let
$\evalhist{\prog}{\processes}$ be the set of all possible histories that can be
produced by $\exec{\prog}{\processes}$, i.e., $\evalhist{\prog}{\processes}$ is
the largest possible set s.t. $\evalhist{\prog}{\processes} = \{ \hist \mid
\hist \models \exec{\prog}{\processes} \}$.

\begin{definition}[\Last{} Write Invocation \cite{SW15-arxiv}] \label{def:last-write-inv}
Given a program $\prog$, a set of processes $\processes$ executing $\prog$ and
a history $\hist$ s.t. $\hist \models \exec{\prog}{\processes}$, i.e.  $\hist \in
\evalhist{\prog}{\processes}$, an invocation $\inv{i}{}{\wop{\obj}{\val}}$ is
the \emph{\last{} write invocation} on some variable $\obj$ by transaction $\tr_i$
in $\hist$, if for any history $\hist' \in \evalhist{\prog}{\processes}$ for
which $\hist$ is a prefix (i.e., $\hist' = \hist \cdot R$) there is no
operation invocation $\inv{i}{}{\wop{\obj}{\valu}}$ s.t.
$\inv{i}{}{\wop{\obj}{\val}}$ precedes $\inv{i}{}{\wop{\obj}{\valu}}$ in
$\hist'|\tr_i$.
\end{definition}

\begin{definition}[\Last{} Write \cite{SW15-arxiv}] \label{def:last-write-op}
Given a program $\prog$, a set of processes $\processes$ executing $\prog$ and
a history $\hist$ s.t. $\hist \models \exec{\prog}{\processes}$, an operation
execution is the \emph{\last{} write} on some variable $\obj$ by transaction
$\tr_i$ in $\hist$ if it comprises of an invocation and a response other than
$\ab_i$, and the invocation is the \emph{\last{} write invocation} on $\obj$ by
$\tr_i$ in $\hist$.
\end{definition}

\begin{definition} [Transaction Decided on $\obj$ \cite{SW15-arxiv}] \label{def:decided} 
Given a program $\prog$, a set of processes $\processes$ and a history $\hist$
s.t. $\hist \models \exec{\prog}{\processes}$, we say transaction $\tr_i \in
\hist$ \emph{decided on} variable $\obj$ in $\hist$ iff $\hist|\tr_i$ contains
a complete write operation execution $\fwop{i}{\obj}{\val}{\ok_i}$ that is the
\last{} write on $\obj$.
\end{definition}

Given some history $\hist$, let $\cpetrans{\hist}$ be a set of transactions
s.t. $\tr_i \in \cpetrans{\hist}$ iff there is some variable $\obj$
s.t. $\tr_i$ decided on $\obj$ in $\hist$.

Given any $\tr_i \in \hist$, a \emph{decided transaction subhistory}, denoted 
$\hist\cpe\tr_i$, is the longest subsequence of $\hist|\tr_i$ s.t.:
\begin{enumerate}[a) ]
    \item $\hist\cpe\tr_i$ contains $\init_i\to\valu$, and
    \item for any variable $\obj$, if 
    $\tr_i$ decided on $\obj$ in $\hist$, 
    then
    $\hist\cpe\tr_i$ contains $(\hist|\tr_i)|\obj$.
\end{enumerate}

In addition, a \emph{decided transaction subhistory completion}, denoted
$\hist\cpeC\tr_i$, is a sequence s.t. $\hist\cpeC\tr_i =
\hist\cpe\tr_i \cdot [\tryC_i\to\co_i]$.

Given a sequential history $S$ s.t. $S\equiv\hist$,
$\luvis{S}{\tr_i}$ is the longest subhistory of $S$, s.t. for each $\tr_j
\in S$:
\begin{enumerate}[a) ]
    \item $S|\tr_j \subseteq \luvis{S}{\tr_i}$ if $i=j$ or $\tr_j$ is committed
    in $S$ and $\tr_j \prec_S \tr_i$, or 
    \item $S\cpeC\tr_j \subseteq \luvis{S}{\tr_i}$ only if $\tr_j$ is not
    committed in $S$ but $\tr_j \in \cpetrans{\hist}$ and $\tr_j \prec_S \tr_i$ and not $\tr_j \prec_\hist \tr_i$ .
\end{enumerate}

Given a sequential history $S$ and a transaction $\tr_i \in S$, we then say
that transaction $\tr_i$ is \emph{last-use legal in} $S$ if $\luvis{S}{\tr_i}$
is legal.

\begin{definition} [Final-state Last-use Opacity \cite{SW15-arxiv}] \label{def:final-state-lopacity} \label{def:fs-lopacity}
    A finite history $\hist$ is \emph{final-state last-use opaque} if, and only if,
    there exists a sequential history $S$ equivalent to any completion of
    $\hist$ s.t., 
    \begin{enumerate}[a) ] 
        \item $S$ preserves the real-time order of $\hist$,
        \item every transaction in $S$ that is committed in $S$ is
        legal in $S$,
        \item every transaction in $S$ that is not committed in $S$ is
        last-use legal in $S$. 
    \end{enumerate}
\end{definition}

\begin{definition} [Last-use Opacity \cite{SW15-arxiv}] \label{def:lopacity}
    A history $\hist$ is \emph{last-use opaque} if, and only if, every finite prefix of
    $\hist$ is final-state last-use opaque.
\end{definition}

\subsection{Composition Rules}

Given trace $\trace$ and a history $\hist = \tohist{\trace}$,
let $\hat{C} = \compl{\hist}$ be a completion of $\hist$ s.t. for every $\tr_i
\in \hist$, if $\tr_i$ is live or commit-pending in $\hist$, then $\tr_i$ is
aborted in $\hist_C$.
Let $\s\tr_i$ such a transaction in $\hat{C}$ that corresponds to a completion
of $\tr_i$ in $\hat{C}$.

\begin{definition}[Equivalent Sequential History Construction]
    \label{seqh-construction}
    Let $\seqH$ be a sequential history
    s.t. $\seqH \equiv \hist_C$
    and, given two transactions $\tr_i, \tr_j \in \hat{C}$:
    \begin{enumerate}
        \item if $\tr_i \prec_{\trace} \tr_j$, then $\tr_i \prec_{\seqH} \tr_j$,
        \item otherwise, if $\tr_i \isoorder_\trace \tr_j$ for any variable $\obj$, 
            then $\tr_i \prec_{\seqH} \tr_j$,
        \item otherwise, if $\exists \op_j = \fwop{j}{\obj}{\any}{\ok_j} \in \trace|\tr_j$
            and $\exists e^i = \get{i}{\obj}{\any} \in \trace|\tr_i$ or $e^i
            = \rset{i}{\obj}{\any} \in \trace|\tr_i$, then $\tr_i \prec_{\seqH}
            \tr_j$.
    \end{enumerate}
\end{definition}

\begin{definition}[Last-use Visible History Construction]
\label{lvis-construction}
Given transactions $\tr_i$ and $\tr_j$ in $\trace$:
\begin{enumerate}
    \item if $\tr_j$ is committed in $\trace$, then
        $\s\tr_j$ is included in $\luvis{\seqH}{\tr_i}$ as a whole, otherwise
    \item if $\tr_j$ is aborted in $\trace$ and $\tr_j \prec_\trace \tr_i$,
        $\s\tr_j$ is not included in $\luvis{\seqH}{\tr_i}$ at all, otherwise
    \item if there exists 
        $\vchain{\trace}{j}{i}$,
        then $\seqH\cpeC\s\tr_j$ is included in $\luvis{\seqH}{\tr_i}$, otherwise
    \item $\tr_j$ is not included in $\luvis{\seqH}{\tr_i}$ at all.
\end{enumerate}
\end{definition}

\subsection{Auxilia}

\begin{lemma} \label{lemma:direct-precedence-in-vis}
              \label{lemma:direct-precedence-in-luvis}
Let there be a consonant,
isolation-ordered, 
trace $\trace$ in obbligato and $\hist = \tohist\trace$ from which $\seqH$ is generated, and
$\tr_i, \tr_j \in \trace$.
Given any non-local $\op_i = \frop{i}{\obj}{\val} \in \trace|\tr_i$ s.t.
$\exists e_i = \get{i}{\obj}{\val} \in \trace|\tr_i$ and $\op_i \dependson
e_i$ and 
given any non-local $\op_j = \fwop{j}{\obj}{\val}{\ok_j} \in \trace|\tr_j$ s.t.
$\exists e_j = \in \rset{i}{\obj}{\val} \in \trace|\tr_j$ and $\op_j
\instigates e_j$, and $e_j \pref_{\trace} e_i$,
then $\nexists \tr_k \in \trace$ s.t. $\tr_k$ $\op_k =
\fwop{j}{\obj}{\val'}{\ok_j}$ s.t. $op_i \prec_{\seqH} \op_k \prec_{\seqH}
\op_j$ and $\tr_k$ is either committed or decided on $\obj$ in trace $\trace$.
\end{lemma}

\begin{proof}
Assume for the sake of contradiction that such $\tr_k$ exists in $\trace$.
Since both $\op_i$ and $\op_j$ are non-local, then $i \neq j \neq k$.

If $\tr_k$ is committed, then, from the definition of commit write obbligato, $\exists
e_k = \rset{k}{\obj}{\val'} \in \trace|\tr_k$ if
$\inv{k}{}{\twop{k}{\obj}{\val'}}$ is the invocation event of $\op_k$ then
$\inv{k}{}{\twop{k}{\obj}{\val'}} \prec_\trace e_k \prec_\trace
\res{k}{}{\co_k}$. 

If $\tr_k$ is decided on $\obj$ in $\trace$, then, from the
definition of closing write obbligato, $\exists
e_k = \rset{k}{\obj}{\val'} \in \trace|\tr_k$ s.t. if
$\inv{k}{}{\twop{k}{\obj}{\val'}}$ is the invocation event of $\op_k$ then
$\inv{k}{}{\twop{k}{\obj}{\val'}} \prec_\trace e_k \prec_\trace
e_i$.

Thus, in either of the above cases, $e_j \prec_\trace e_k$ and either $e_k
\prec e_i$ or $e_i \prec e_k$.
If then $e_k \prec e_i$, it is not true that $e_j \pref_{\trace}
e_i$, \underline{which is a contradiction}.
Alternatively, if $e_i \prec e_k$, then, since $\trace$ is isolation-ordered,
$\tr_i \iso{\obj}{\trace} \tr_k$, which implies that $\tr_i \prec_{\seqH} \tr_k$.
In this case, $\op_i \prec_{\seqH} \op_k$, \underline{which is a contradiction}.

Therefore, there can be no such $\tr_k$, which satisfies the lemma.
\end{proof}

\begin{lemma} \label{lemma:abset:views} %
    Given a consonant trace $\trace$, and $\tr_i \in \trace$, if $\tr_j$ is the
    first element of $\abset{\trace}{\tr_j}{\obj}$, then $\exists e_v =
    \get{j}{\obj}{\val} \in \trace|\tr_j$ that is initial and non-local, and
    either 
    \begin{enumerate}[a)]
        \item $\val = 0$ and $\nexists \tr_k\in\trace$ s.t. $e_u =
        \set{k}{\obj}{\val} \in \trace|\tr_k$ and $e_u \prec_\trace e_v$,
    \item $\val \neq 0$ and $\exists \tr_k\in\trace$ s.t. $e_u =
        \set{k}{\obj}{\val} \in \trace|\tr_k$ and $e_u \pref_\trace e_v$.
    \end{enumerate}
\end{lemma}

\begin{proof}
    Since $\tr_j$ is in $\absetx{i}$ then by definition, either $k=i$ or $e_a =
    \aset{j}{\obj}{\val} \in \trace|\tr_j$. In either case $e_v =
    \get{j}{\obj}{\val} \in \trace|\tr_j$ s.t. $e_v$ is initial and non-local
    (in the former case by definition of $\absetx{i}$ and in the latter by
    definition of recovery update consonance).
     
    Since $e_v$ is consonant and non-local, then either: 
    \begin{enumerate}[a)]
    \item $\val = 0$ and $\nexists \tr_k \in \trace$ s.t. $e_u
    \set{k}{\obj}{\val'} \in \trace|\tr_k$ $e_u \prec_\trace e_r$,
    \item $\val \neq 0$ and $\exists \tr_k \in \trace$ s.t. $e_u =
    \rset{j}{\obj}{\val} \in \trace|\tr_k$, $i\neq k$, $e_u \pref_\trace e_r$, $e_u$
    is consonant, and $e_u$ is the ultimate routine update on $\obj$ in
    $\trace|\tr_k$, or
    \item $\exists e_u \in \trace$ s.t. $e_u = \aset{j}{\obj}{\val}$ for some
    $tr_k$, $j\neq k$, $e_u \pref_\trace e_r$, $e_u$ is a consonant recovery event,
    and is the ultimate update on $\obj$ in $\trace|\tr_k$.
    \end{enumerate}
    
    In the latter-most case, if such $e_u$ exists in $\tr_k$ then, $\tr_k \in
    \absetx{j}$ so that $\tr_k$ preceded $\tr_j$ in $\absetx{j}$. Thus, $\tr_k$
    would precede $\tr_j$ in $\absetx{i}$, and therefore $\tr_j$ is not the
    first element of $\absetx{i}$. Thus, the latter-most case is impossible.
\end{proof}

\begin{lemma} \label{lemma:abset:not-committed} %
    Given a consonant trace $\trace$, and $\tr_i \in \trace$, $\forall \tr_j \in
    \abset{\trace}{\tr_i}{\obj}$ ($i\neq j$), $\tr_j$ is aborted or live in
    $\trace$.
\end{lemma}

\begin{proof}    
    Since $i \neq j$ then $\forall \tr_j \in \trace$, $\exists e_a =
    \aset{j}{\obj}{\val} \trace|\tr_j$. Since $\trace$ is consonant, then $e_a$
    is consonant, so $e_a$ is dooming. Thus $\tr_j$ is aborted or live in
    $\trace$.
\end{proof}

\begin{lemma} \label{lemma:abset:containment} %
    Given a consonant, abort abiding trace $\trace$ in obbligato, and a pair of
    transaction $\tr_i, \tr_j, \in
    \trace$, and $\tr_j$ is the first element in $\absetx{i}$, $\forall \tr_k
    \in \trace$ if $\tr_j \iso{\trace}{\obj} \tr_k \iso{\trace}{\obj} \tr_i$
    and $\op_k \fwop{k}{\obj}{\val} \in \trace|\tr_k$
    then $\tr_k$ is aborted or live in $\trace$.
\end{lemma}

\begin{proof}
    If $i=j$, then the lemma is \underline{vacuously true}.
    
    Since  $\tr_j \iso{\trace}{\obj} \tr_k \iso{\trace}{\obj} \tr_i$, then
    $\exists e_u = \rset{k}{\obj}{\val'} \in \trace|\tr_k$ or $e_v =
    \get{k}{\obj}{\val'} \in \trace|\tr_k$. Hence, either $e_u$ exists in
    $\trace|\tr_k$ or it does not.
     
    If $e_u$ does not exist, then, from commit write obbligato, $\tr_k$ cannot commit in
    $trace$, so \underline{$\tr_k$ is either live or aborted in $\trace$}.

    If $e_u$ exists, then, since $\tr_j \iso{\trace}{\obj} \tr_k
    \iso{\trace}{\obj} \tr_i$ and from the definition of $\absetx{i}$, there is
    some pair of transactions $\tr_\alpha$ and $\tr_\beta \in \trace$ s.t.
    $\tr_\alpha, \tr_\beta \in \absetx{i}$ and $\tr_\alpha$ immediately
    precedes $\tr_\beta$ in $\absetx{i}$ and $\tr_\alpha \iso{\trace}{\obj}
    \tr_k \iso{\trace}{\obj} \tr_\beta$.
    Therefore $\exists e_\alpha = \aset{\alpha}{\obj}{\val_\alpha} \in
    \trace|\tr_\alpha$ and $e_\beta = \get{\beta}{\obj}{\val_\beta} \in
    \trace|\tr_\beta$ s.t. $e_\alpha \pref_\trace e_\beta$.
    In addition, since $e_\alpha$ is consonant, then it is needed, so $\exists
    e_\alpha' = \rset{\alpha}{\obj}{\val_\alpha'} \in \trace|\tr_\alpha$ s.t.
    $e_\alpha' \pref_\trace e_\alpha$.
    Also, from definition of isolation order, $e_\alpha' \prec_\trace e_u
    \prec_\trace e_\beta$.
    Then,  $e_\alpha' \prec_\trace e_u \prec_\alpha \pref_\trace e_\beta$.
    Therefore, from the definition of abort accord,  \underline{$\tr_k$ is
    either live or aborted in $\trace$}.
\end{proof}

\begin{lemma}\label{lemma:abset:order}
    Given a consonant trace $\trace$, and $\tr_i \in \trace$, $\forall \tr_j,
    \tr_k \in \abset{\trace}{\tr_i}{\obj}$ ($i\neq j$), if $\tr_j$ precedes
    $\tr_k$ in $\abset{\trace}{\tr_i}{\obj}$ then $\tr_j \iso{\obj}{\trace}
    \tr_k$.
\end{lemma}

\begin{proof}
Given $\absetx{i}$, from \rlemma{lemma:abset:not-committed}, $\forall \tr_k \in
\absetx{i}$, $\tr_k$ is aborted or live in $\trace$. 
In addition, since for all $\tr_m \in \absetx{i}$ except the first, where $e_v^m
= \get{m}{\obj}{\val} \in \trace|\tr_m$ there is some $\tr_n$ that directly
precedes $\tr_m$ in $\absetx{i}$ and contains $e_a^n = \aset{n}{\obj}{\val}$
s.t. $e_a^n \pref_\trace e_v^m$.
Since $e_a^n$ is conservative, there is a preceding view
$e_v^n = \get{n}{\obj}{\val}$ s.t. $e_v^n \pref_\trace e_a^n$. Thus $e_v^n
\pref_\trace e_v^m$, so $\tr_n \iso{\obj}{\trace} \tr_m$. 
\end{proof}

\begin{corollary} \label{cor:abset:order}
    Given a consonant trace $\trace$, and $\tr_i \in \trace$, $\forall \tr_j \in
    \abset{\trace}{\tr_i}{\obj}$ ($i\neq j$), $\tr_j \iso{\obj}{\trace} \tr_i$.
\end{corollary}

\begin{lemma} 
    \label{lemma:corr}
    Given $\tr_i, \tr_j$ s.t. $\s\tr_j \sin \luvis{\seqH}{\tr_i}$,
    $\forall \tr_k$ if $\s\tr_k \sin \luvis{\seqH}{\tr_j}$,
    then $\s\tr_k \sin \luvis{\seqH}{\tr_i}$ and
    $\luvis{\seqH}{\tr_i}|\tr_k = \luvis{\seqH}{\tr_j}|\tr_k$ 
\end{lemma}    

\begin{proof}
    If $\tr_k$ is committed in $\trace$ and $\s\tr \sin \luvis{\seqH,\tr_j}$
    then $\seqH|\s\tr_k \subseteq \luvis{\seqH,\tr_j}$ and $\s\tr_k
    \prec_{\seqH} \s\tr_j$.
    Since $\tr_j \sin \luvis{\seqH, \tr_i}$ , then $\s\tr_j
    \prec_{\seqH} \s\tr_i$.
    Since $\tr_k$ is committed in $\trace$ and $\s\tr_j
    \prec_{\seqH} \s\tr_i$, then
    $\seqH|\s\tr_k \subseteq \luvis{\seqH}{\tr_i}$.

    If $\tr_k$ is not committed in $\trace$ and $\s\tr \sin \luvis{\seqH,\tr_j}$
    then $\seqH\cpeC\s\tr_k = \luvis{\seqH,\tr_j}|\tr_k$ and $\s\tr_k
    \prec_{\seqH} \s\tr_j$ and $\tr_k \nprec_\trace \tr_j$ and $\exists
    \vchain{\trace}{k}{j}$ (from \rdef{lvis-construction}).

    Since $\tr_k$  is not committed in $\trace$, and since $\trace$ is commit
    abiding, then from \rlemma{lemma:helper-lemma-1}, there cannot be
    $\vchain{\trace}{k}{j}$ s.t. $\tr_j$ is committed. Thus $\tr_j$ is not
    committed in $\trace$.
    Thus, if $\seqH\cpeC\s\tr_j$ then $\s\tr_j \prec_{\seqH} \s\tr_i$ and
    $\tr_j \nprec \tr_i$ and $\exists \vchain{\trace}{j}{i}$.

    If $\exists \vchain{\trace}{k}{j}$ and $\exists \vchain{\trace}{j}{i}$ then
    $\exists{\trace}{\tr_k}{\tr_i}$.

    Either $\tr_k$ aborts in $\trace$ (i.e. $\res{k}{}{\ab_k}$) or $\tr_k$ is
    live in $\trace$.
    In the latter case trivially $\tr_k \nprec_\trace \tr_i$. In the former
    case, from \rlemma{lemma:helper-lemma-2}, also $\tr_k \nprec_\trace \tr_i$.

    Since $\exists \vchain{\trace}{\tr_k}{\tr_i}$ and $\s\tr_k \prec_{\seqH}
    \s\tr_i$ and $\tr_k \nprec_\trace \tr_i$ then $\seqH\cpeC\s\tr_k =
    \luvis{\seqH,\tr_i}|\tr_k$ (from \rdef{lvis-construction}).
\end{proof}

\begin{lemma} 
    \label{lemma:not-corr}
    Given $\tr_i, \tr_j$ s.t. $\s\tr_j \sin \luvis{\seqH}{\tr_i}$,
    $\forall \tr_k$ if $\s\tr_k \nsin \luvis{\seqH}{\tr_j}$
    and $\exists \vchain{\trace}{\tr_k}{\tr_j}$
    then $\s\tr_k \sin \luvis{\seqH}{\tr_i}$.
\end{lemma}  

\begin{proof}
    If $\tr_k \nsin \luvis{\seqH}{\tr_j}$ and $\s\tr_k \prec_{\seqH} \s\tr_j$
    then $\s\tr_k$ is not committed in $\trace$. 
    
    If $\seqH\cpeC\s\tr_k \nsubset \luvis{\seqH}{\tr_j}$ and $\s\tr_k
    \prec_{\seqH} \s\tr_j$ then either $\tr_k \prec_\trace \tr_j$ or $\nexists
    \vchain{\trace}{i}{j}$. The latter case contradicts the assumptions of the
    lemma, hence $\tr_k \prec_\trace \tr_j$.

    If $\tr_k \prec_\trace \tr_j$, then $\exists r = \res{k}{}{\ab_k} \in
    \trace|\tr_k$ s.t. for every event $e$ in $\trace|\tr_j$, $r \prec_\trace
    e$.
    Since $\exists \vchain{\trace}{j}{i}$ then there is some view event $e_v^i$
    in $\trace|\tr_i$ and some update event $e_u^j$ in $\trace|\tr_j$ s.t.
    $e_u^j \prec_\trace e_v^i$. Therefore $r \prec_\trace e_u^j \prec_\trace
    e_v^i$.
    
    Since no events can occur in $\trace|\tr_k$ after $v$, then for all events
    in $e$ in $\trace|\tr_k$ apart from $r$, $e \prec_\trace r$.
    So, for any $\vchain{\trace}{k}{i}$ for any update event $e_u^k =
    \rset{k}{\obj}{\val} \in \trace|\tr_k$, $e_u^k \prec_\trace r \prec_\trace
    e_v^i$.
    
    From abort coda, $\exists e_a^l = \aset{l}{\obj}{\val'}$ s.t. $e_u^k
    \prec_\trace e_a^l \prec_\trace r$, and, from conservatism and unique
    routine updates, $\val \neq \val'$.
    Thus, since $e_u^k \prec_\trace e^l_a \prec_\trace e_v^i$ and $\val \neq
    \val'$, there cannot be such $\vchain{\trace}{k}{i}$ that satisfies chain
    isolation, and therefore $\nexists\vchain{\trace}{k}{i}$.

    Therefore, from \rlemma{lvis-construction}, $\seqH\cpeC\s\tr_k
    \nsubseteq \luvis{\seqH}{\tr_i}$.

    Thus, $\seqH|\s\tr_i \nsin \luvis{\seqH}{\tr_i}$.
\end{proof}

\begin{lemma}
    \label{lemma:helper-lemma-1}
    Given $\vchain{\trace}{i}{j}$, 
    if $\tr_j$ is committed in $\trace$,
    then $\forall \tr_k \in \vchain{\trace}{i}{j}$, 
    $\tr_k$ is committed in $\trace$.
\end{lemma}

\begin{proof}
    Given a pair of transaction $\tr_l, \tr_m \in \trace$ s.t. $\tr_m \views
    \tr_l$, from commit accord, if $\tr_m$ is committed in $\trace$, then
    $\tr_l$ is also committed in $\trace$.

    If $\vchain{\trace}{i}{j} = \tr_i \cdot \tr_j$, then n since $\tr_j$ is
    committed in $\trace$, then so is $\tr_i$.

    Since, $\vchain{\trace}{i}{j} = \tr_i \vchain{\trace}{i}{k} \cdot \tr_j$
    and $\tr_k$ is such that $\tr_j \views \tr_k$, then since $\tr_j$ is
    committed in $\trace$, then so is $\tr_k$.
    This follows recursively for $\vchain{\trace}{i}{k}$.

    Thus every transaction in $\vchain{\trace}{i}{j}$ is committed in $\trace$.
\end{proof}

\begin{lemma}
    \label{lemma:helper-lemma-2}
    Given $\vchain{\trace}{i}{j}$, 
    if $\tr_i$ aborts in $\trace$,
    then $\tr_i \nprec_\trace \tr_j$. 
\end{lemma}

\begin{proof}
    Assume for the sake of contradiction that $\tr_i \prec_\trace \tr_j$.

    Thus, there exists $\tr_k \in \vchain{\trace}{i}{j}$ s.t. $\tr_k \views
    \tr_i$, so $\exists e_u^i = \rset{i}{\obj}{\val} \in \trace|\tr_i$ and
    $e_v^k = \get{k}{\obj}{\val} \in \trace|\tr_k$ and $e_u^i \prec_\trace
    e_v^k$.

    If $\tr_i$ is aborted, then, from abort coda, $\exists e_a^l =
    \rset{l}{\obj}{\val'}$ s.t. $e_u^i \prec_\trace e_a^l \prec_\trace
    \res{i}{}{\ab_i}$ and from unique routine updates $\val \neq \val'$.

    Since $\tr_j$ is in $\vchain{\trace}{i}{j}$, $\exists e^j_v =
    \get{j}{\objy}{\val''}$ and since $\tr_i \prec_\trace \tr_j$, then
    $\res{i}{}{\ab_i} \prec_\trace e^j_v$. Thus, $e_u^i \prec_\trace
    e_a^l\prec_\trace e_v^j$.

    This contradicts chain isolation, so it is not true that $\tr_i
    \prec_\trace \tr_j$, so $\tr_i \nprec_\trace \tr_j$.
\end{proof}

\subsection{Main Lemmas}

Let there be a harmonious trace $\trace$ and $\hist =
\tohist\trace$ from which $\seqH$ is generated. Let there be such $\tr_i \in
\trace$ that $\tr_i$ is committed in $\trace$. Then:

\begin{lemma} [Unique Routine Updates] \label{lemma:unique-routine-updates}
If $\trace$ is consonant, and $\trace$ has unique writes, then given any
$\set{i}{\obj}{\val}$ and $\set{j}{\obj}{\val'}$ s.t. $\val \neq \val'$.
\end{lemma}

\begin{proof}
    Since both events are consonant, then for $\set{i}{\obj}{\val}$ there
    exists $\op_i = \fwop{i}{\obj}{\val^i}{\ok_i}$ s.t. $\val = \val^i$, and
    for $\set{j}{\obj}{\val'}$ there exists $\op_j =
    \fwop{j}{\obj}{\val^j}{\ok_j}$ s.t. $\val' = \val^j$.
    Since $\trace$ has unique writes, then $\val^i \neq \val^j$, so $\val \neq
    \val'$.
\end{proof}

\begin{lemma} [Non-local Read Consistency] 
\label{lemma:non-local-read-consistency}
For any $\op_i \in \trace|\tr_i$ s.t. $\op_i = \frop{i}{\op_i}{\val}$ and
$op_i$ is non-local, then either:
\begin{enumerate} [1)]
\item $\val \neq 0$ and $\exists \op_j \in \vis{\seqH}{\s\tr_i}$ for some
$\tr_j$ s.t. $\op_j = \fwop{j}{\obj}{\val}{\ok_j}$, $\op_j
\pref_{\vis{\seqH}{\s\tr_i}} \op_i$, or 

\item $\val = 0$ and $\nexists \op_j \in \vis{\seqH}{\s\tr_i}$ s.t. $\op_j =
\fwop{j}{\obj}{\val}{\ok_j}$ and $\op_j \prec_{\vis{\seqH}{\s\tr_i}} \op_i$.
\end{enumerate}
\end{lemma}

\begin{proof}[Proof for \rlemma{lemma:non-local-read-consistency}]
Since $\op_i$ is consonant and non-local, then
$\exists e_v = \get{i}{\obj}{\val} \in \trace$, s.t. $\op_i \dependson e_v$
and $e_v$ is consonant.
Then, from $e_v$'s consonance, either:
\begin{enumerate}[a)]
\item $\val = 0$ and $\nexists e_u = \set{j}{\obj}{\val'} \in \trace$ for some
$\tr_j \in \trace$ s.t. $e_u \prec_\trace e_v$. \label{case:first} 

In which case, if $\nexists \op_j \fwop{j}{\obj}{\val'}{\ok_i} \in
\trace|\tr_j$ s.t. $\op_j \prec_\trace \get{i}{\obj}{\val}$, then, 
$\nexists \tr_j \in \trace$ s.t. $\tr_j \isoorder_\trace \tr_i$ and $\op_j \in
\trace|\tr_j$.
Thus, from construction of $\seqH$, $\nexists \s\tr_j \in \trace$ s.t. $\s\tr_j
\prec_{\seqH} \s\tr_i$ and $\op_j \in \seqH|\s\tr_j$.
Thus, from construction of $\vis{\seqH}{\s\tr_i}$, for any such $\tr_j$,
$\seqH|\s\tr_j \nsubseteq \vis{\seqH}{\s\tr_i}$, so for any such $\tr_j$,
\underline{$\fwop{j}{\obj}{\val}{\ok_j} \not\in \vis{\seqH}{\s\tr_i}$ and $\val=0$}.

On the other hand, if $\exists \op_j
\fwop{j}{\obj}{\val'}{\ok_i} \in \trace|\tr_j$ s.t. $\op_j \prec_\trace
\get{i}{\obj}{\val}$, then
if $\tr_i$ is committed in $\trace$, then, from the definition of commit write obbligato,
$\rset{j}{\obj}{\val'} \in \trace|\tr_j$, which contradicts the assumption of
case \ref{case:first}).
Thus, $\tr_i$ is not committed in $\trace$, so $\s\tr_i$ is not committed in
$\seqH$, and therefore $\seqH|\s\tr_i \nsubseteq \vis{\seqH}{\s\tr_j}$. Thus
for any such $\tr_j$, \underline{$\fwop{j}{\obj}{\val'}{\ok_j} \not\in
\vis{\seqH}{\s\tr_i}$ and $\val=0$}.

\item $\val \neq 0$ and $\exists e_u = \rset{j}{\obj}{\val} \in \trace$ for
some $\tr_j \in \trace$ s.t. $e_u \pref_\trace e_v$ and $e_u$ is consonant.
\label{case:ordinary}

Since $e_u$ is consonant, then $\exists \op_j =
\fwop{j}{\obj}{\val}{\ok_j} \in \trace|\tr_j$ s.t. $\op_j$ is non-local and
consonant, and $\op_j \pref \rset{j}{\obj}{\val}$.
Thus, since $e_u \pref_\trace e_v$, $\tr_j \iso{\obj}{\trace} \tr_i$, then, by
construction, $\s\tr_j \prec{\seqH} \s\tr_i$.

Since $\tr_i$ is committed in $\trace$ and $\tr_i \views \tr_j$, and since
$\trace$ is commit-abiding, then $\tr_j$ must be committed in $\trace$.
Thus $\s\tr_j$ is also committed in $\seqH$.
Thus, $\seqH|\s\tr_j \subseteq \vis{\seqH}{\tr_i}$, and therefore $\op_j
\prec_{\vis{\seqH}{\tr_i}} \op_i$.
Then, from \rlemma{lemma:direct-precedence-in-vis}, $\op_j
\pref_{\vis{\seqH}{\tr_i}} \op_i$. Thus,
\underline{$\fwop{j}{\obj}{\val}{\ok_j} \pref_{\vis{\seqH}{\s\tr_i}} \op_i$ and
$\val\neq 0$}. 

\item $\exists e_a = \aset{j}{\obj}{\val} \in \trace$ for some $\tr_j \in
\trace$ s.t. $e_a \pref_\trace e_v$ and $e_u$ is consonant.  \label{case:rever}

Given $\absetx{i}$, from \rlemma{lemma:abset:not-committed}, $\forall \tr_k \in
\absetx{i}$, $\tr_k$ is aborted or live in $\trace$. So, by construction,
$\s\tr_k$ is aborted in $\seqH$, and therefore excluded from
$\vis{\seqH}{\tr_i}$. Thus for any $\tr_k \in \absetx{i}$, $\forall \op_k =
\fwop{k}{\obj}{\val} \in \trace|\tr_k$, $\op_k \not\in \vis{\seqH}{\tr_i}$.

Given $\absetx{i}$, from \rlemma{lemma:abset:views}, $\exists e_v' =
\get{k}{\obj}{\val} \in \trace|\tr_k$ s.t. $\tr_k$ is the first element of
$\absetx{i}$ that is initial and non-local, and either of the following is true:
\begin{enumerate}[i)]
    \item $\val = 0$ and $\nexists \tr_l\in\trace$ s.t. $e_u' =
    \set{l}{\obj}{\val} \in \trace|\tr_l$ and $e_u' \prec_\trace e_v'$.
    
    Then, either $\exists e_u' \in \set{l}{\obj}{\val} \in \trace|\tr_l$ and
    $e_v' \prec_\trace e_u'$ or  $\nexists e_u' = \set{l}{\obj}{\val} \in
    \trace|\tr_l$.

    If $\exists e_u' \in \set{l}{\obj}{\val'} \in \trace|\tr_l$ and $e_v'
    \prec_\trace e_u'$, then from \rlemma{lemma:abset:views},  $e_u' =
    \rset{l}{\obj}{\val'}$. Thus, be definition of isolation order, $\tr_k
    \iso{\obj}{\trace} \tr_l$. 
    Thus, if $\tr_l \iso{\obj}{\trace} \tr_j$, then, from
    \rlemma{lemma:abset:containment}, $\tr_l$ is aborted or live in $\trace$, 
    so, by construction, $\s\tr_l$ is aborted in $\seqH$. Therefore $\s\tr_l
    \nsubseteq \vis{\seqH}{\tr_i}$, so for any write operation execution $\op_l
    = \fwop{l}{\obj}{\val'}{\ok_l}$ in any such $\tr_l$, \underline{$\op_l
    \not\in \vis{\seqH}{\tr_i}$ (and $\val=0$)}. 
    Alternatively, if  $\tr_j \iso{\obj}{\trace} \tr_l$, then since $e_a
    \pref_\trace e_v$, then it is not possible that $e_a \prec_\trace e_u'
    \prec e_v$. By corollary, from the definition of isolation order, it is not
    possible that $\tr_j \iso{\obj}{\trace} \tr_l \iso{\obj}{\trace} \tr_i$.
    Then, by construction, $\s\tr_i \prec_{\seqH} \s\tr_l$, so $\s\tr_l
    \nsubseteq \vis{\seqH}{\tr_i}$. Therefore, for any write operation execution
    $\op_l = \fwop{l}{\obj}{\val'}{\ok_l}$ in any such $\tr_l$,
    \underline{$\op_l \not\in \vis{\seqH}{\tr_i}$ (and $\val=0$)}. 
    
    On the other hand, if $\nexists e_u' \in \set{l}{\obj}{\val} \in
    \trace|\tr_l$, then either $\trace|\tr_l$ contains some write operation
    $\op_l = \fwop{l}{\obj}{\val'}{\ok_l}$ or it does not. 
    If it does not, then vacuously, for any write operation execution $\op_l
    = \fwop{l}{\obj}{\val'}{\ok_l}$ in any such $\tr_l$, \underline{$\op_l
    \not\in \vis{\seqH}{\tr_i}$ (and $\val=0$)}.
    On the other hand, if $\op_l \in \trace|\tr_l$, then from commit write
    obbligato, since $\nexists e_u' = \rset{l}{\obj}{\val} \in
    \trace|\tr_l$, then $\tr_l$ is not committed in $\trace$. Thus, $\s\tr_l$ is
    aborted in $\seqH$ and $\s\tr_l \nsubseteq \vis{\seqH}{\tr_i}$. Thus, for
    any write operation execution $\op_l = \fwop{l}{\obj}{\val'}{\ok_l}$ in any
    such $\tr_l$, \underline{$\op_l \not\in \vis{\seqH}{\tr_i}$ (and
    $\val=0$)}.

    \item $\val \neq 0$ and $\exists \tr_l\in\trace$ s.t. $e_u' =
    \rset{l}{\obj}{\val} \in \trace|\tr_l$ and $e_u' \pref_\trace e_v'$.

    Since $\trace$ is consonant, then $e_u'$ is consonant, so $\exists \op_l =
    \fwop{l}{\obj}{\val}{\ok_l}$ s.t $\op_l \instigates e_u'$. 
    
    Since for all $\tr_m \in \absetx{i}$, $e_v'' = \get{m}{\obj}{\val''}$ there
    is some $\tr_n$ that directly precedes $\tr_m$ in $\absetx{i}$ and contains
    $e_a'' = \aset{n}{\obj}{\val''}$ s.t. $e_a'' \pref_\trace e_v''$. Since
    $e_a''$ is conservative, there is a preceding view $e_v''' =
    \get{n}{\obj}{\val''}$ s.t. $e_v''' \pref_\trace e_a''$. Thus $e_v'''
    \pref_\trace e_v''$, so $\tr_n \iso{\obj}{\trace} \tr_m$.  Therefore,
    $\tr_k \iso{\obj}{\trace} \tr_i$, and, by extension, since $e_u'
    \pref_\trace e_v'$, $\tr_l \iso{\obj}{\trace} \tr_k$, then $\tr_l
    \iso{\obj}{\trace} \tr_i$. 
    Since $\tr_i$ is committed in $\trace$, then since $\tr_l
    \iso{\obj}{\trace} \tr_i$, then, from commit coherence, $\tr_l$ is either
    committed or aborted in $\tr_j$.

    Transaction $\tr_li$ cannot be aborted in $\trace$, as follows. 
    Let us assume by contradiction that $\tr_l$ is aborted (i.e. $\exists r_a =
    \res{l}{}{\ab_l} \in \trace|\tr_l$).
    Then, since $\trace$ has coda, then for some $\tr_n$, $\exists e_a''' =
    \aset{n}{\obj}{\val'''}$ s.t. $e_u' \prec_\trace e_a''' \prec_\trace r$.
    Since $e_a'$ is consonant, then since it is clean and $\tr_l
    \iso{\obj}{\trace} \tr_k$, then there is no recovery event following $e_v'$
    and preceding $e_a'$. 
    In addition from commit coherence, $\tr_l$ must
    abort before $\tr_i$ commits, so, by extension $e_a'''$ must precede
    $\res{i}{}{\co_i}$ in $\trace$. 
    Thus either $e_a''' \prec_\trace e_v'$ or $e_a' \prec_\trace e_a'''
    \prec_\trace r_a$.
    In the former case, 
    if $e_a''' \prec_\trace e_u'$, then this contradicts that $e_a'''$ is
    consonant (ending),
    and if $e_u' \prec_\trace e_a''' \prec_\trace e_v'$, it contradicts that
    $e_u' \pref_\trace e_v'$.
    On the other hand, if $e_a' \prec_\trace e_a''' \prec_\trace r_a$, then one
    of three scenarios is possible.
    If for some $\tr_m \in \absetx{i}$, s.t. $i\neq m$ and $\get{m}{\obj}{\any}
    \prec_\trace e_a''' \prec_\trace \aset{m}{\obj}{\any}$, then this
    contradicts that $\aset{m}{\obj}{\any}$ is clean.
    Alternatively, if for a pair  $\tr_m, \tr_n \in \absetx{i}$, s.t. $\tr_m$
    directly precedes $\tr_n$ in $\absetx{i}$, and $\aset{m}{\obj}{\any}
    \prec_\trace e_a''' \prec_\trace \get{n}{\obj}{\any}$, then this
    contradicts that $\aset{m}{\obj}{\any} \pref_\trace \get{n}{\obj}{\any}$.
    Finally, if $e_v \prec e_a'''$, then this violates abort coda of $\trace$
    (case b), and is also a contradiction.
    Thus, there cannot be such $e_a'''$, and therefore $\tr_l$ cannot  be
    aborted in $\trace$.

    Hence, $\tr_l$ must be committed in $\trace$.
    Then since 
    $\tr_l \iso{\obj}{\trace} \tr_i$ (so
    $\s\tr_l \prec_{\seqH} \s\tr_i$), $\op_l \in \vis{\seqH}{\tr_i}$.

    Since $e_u'$ is non-local, then it is not followed in $\trace|\tr_l$ by
    another $e_u'' = \rset{l}{\obj}{\val}$.
    Since $\op_l \instigates e_u'$, $\nexists \op_l' =
    \fwop{l}{\obj}{\val'}{\ok_l}$ s.t. $\op_l \prec_{\trace|\tr_i} \op_l'$.
    
    Let $\tr_m$ be a transaction in $\trace$ s.t. $\tr_l \iso{\obj}{\trace}
    \tr_m \iso{\obj}{\trace} \tr_j$.
    If $\tr_m \in \absetx{i}$, then from \rlemma{lemma:abset:not-committed},
    $\tr_m$ is not committed in $\trace$. 
    If $\tr_m \not\in \absetx{i}$, then from \rlemma{lemma:abset:containment},
    $\tr_m$ is also not committed in $\trace$. In either case, $\s\tr_m$ is
    aborted in $\seqH$.
    Thus, for any $\tr_m$ s.t. $\tr_l \iso{\obj}{\trace} \tr_m
    \iso{\obj}{\trace} \tr_j$, $\seqH|\s\tr_m \nsubseteq \vis{\seqH}{i}$.
    Therefore, for any write operation execution $\op_m =
    \fwop{m}{\obj}{\val''}{\ok_m} \in \trace|\tr_m$, $\op_m \not\in \vis{\seqH}{\tr_i}$.
    
    Let $\tr_m \in \trace$ be any transaction s.t. $\tr_j \iso{\obj}{\trace}
    \tr_m \iso{\obj}{\trace} \tr_i$. Since $e_a$ consonant, it is is needed, so
    $\exists e_u'' = \set{j}{\obj}{\val''} \in \trace|\tr_j$ s.t. $e_u''
    \prec_\trace e_a$. In addition, since $\tr_j \iso{\obj}{\trace} \tr_m
    \iso{\obj}{\trace} \tr_i$, then  by definition of isolation order, $\exists
    e = \set{m}{\obj}{\val'''}\in \trace|\tr_m$, or $e =
    \get{m}{\obj}{\val'''}\in \trace|\tr_m$, so $e_u'' \prec_\trace e
    \prec_\trace e_v$. Since $e_a \pref_\trace e_v$, then $e_u'' \prec_\trace e
    \prec_\trace e_u$. Thus, by definition of abort accord, $\tr_m$ is not
    committed in $\trace$, so, by construction, $\s\tr_m$ is aborted in
    $\seqH$. Thus, for any $\tr_m$ s.t.  $\tr_j \iso{\obj}{\trace} \tr_m
    \iso{\obj}{\trace} \tr_i$, $\seqH|\s\tr_m \nsubseteq \vis{\seqH}{i}$.
    Therefore, for any write operation execution $\op_m =
    \fwop{m}{\obj}{\val''}{\ok_m} \in \trace|\tr_m$, $\op_m \not\in
    \vis{\seqH}{\tr_i}$.
    
    Since no other write operation execution follows $\op_l$ in $\trace|\tr_m$,
    and since there is no transaction $\tr_m \in \trace$ s.t.  $\s\tr_l
    \prec_{\seqH} \s\tr_m \prec_{\seqH} \s\tr_i$ (and therefore $\s\tr_l
    \prec_{\vis{\seqH}{\s\tr_i}} \s\tr_m \prec_{\vis{\seqH}{\s\tr_i}} \s\tr_i$)
    s.t. $\exists \op_m = \fwop{m}{\obj}{\val''}{\ok_m} \in \trace|\tr_m$ and
    $\op_m \in \vis{\seqH}{\s\tr_i}$, then
    \underline{$\fwop{l}{\obj}{\val}{\ok_l} \pref_{\vis{\seqH}{\s\tr_i}} \op_i$
    and $\val\neq 0$}. 
\end{enumerate}
\end{enumerate}
\end{proof}

\begin{corollary} [Total Non-local Read Consistency] 
    \label{cor:all-non-local-read-consistency} 
    By extension of the above, since, by definition, if for some sequential
    history $S$, $\seqH|tr_j \in \vis{S}{\tr_i}$, then $\vis{\seqH}{\tr_j}$ is a
    prefix of $\vis{S}{\tr_i}$, then for any $\op_j \in \trace|\tr_j$ s.t.
    $\op_j = \frop{i}{\op_j}{\val}$ either: 
    \begin{enumerate} [1)] 
        \item $\val \neq 0$ and $\exists \op_k \in \vis{\seqH}{\tr_k}$ for
            some $\tr_k$ s.t.  $\op_k = \fwop{k}{\obj}{\val}{\ok_j}$, $\op_k
            \pref_{\vis{\seqH}{\s\tr_i}} \op_j$, or 
        \item $\val = 0$ and $\nexists \op_k \in \vis{\seqH}{\tr_i}$ s.t.
            $\op_k = \fwop{k}{\obj}{\val}{\ok_k}$ and $\op_k
            \prec_{\vis{\seqH}{\tr_i}} \op_j$.
    \end{enumerate}
\end{corollary}

\begin{lemma} [Local Read Consistency] 
    \label{cor:local-read-consistency} 
    For any $\op_i \in \trace|\tr_i$ s.t. $\op_i = \frop{i}{\op_i}{\val}$ and
    $op_i$ is local, then $\exists \op_i' \in \vis{\seqH}{\tr_i}$ s.t.
    $\op_i' = \fwop{i}{\obj}{\val}{\ok_i}$, $\op_i'
    \pref_{\vis{\seqH}{\tr_i}} \op_i$. 
\end{lemma}

\begin{proof}
    From local read consonance it follows that For any $\op_i \in \trace|\tr_i$
    s.t. $op_i$ is local, then $\exists \op_i' \in \trace|\tr_i$ and $\op_i'
    \pref_{\trace|\tr_i} \op_i$.  Thus, $\op_i' \pref_{\seqH|\s\tr_i} \op_i$.
    Since $\tr_i \subseteq \vis{\seqH}{\tr_i}$ then $\op_i'
    \pref_{\vis{\seqH}{\tr_i}} \op_i$.
\end{proof}

\begin{corollary} [Total Local Read Consistency] 
    \label{cor:all-local-read-consistency} 
    For any $\op_i \in \trace|\tr_i$ s.t. $\op_i = \frop{i}{\op_i}{\val}$ and
    $op_i$ is local, then for any $\tr_j$, $\exists \op_i' \in
    \vis{\seqH}{\tr_j}$ s.t.  $\op_i' = \fwop{i}{\obj}{\val}{\ok_i}$, $\op_i'
    \pref_{\vis{\seqH}{\tr_j}} \op_i$. 
\end{corollary}

Let there instead be such $\tr_i \in \trace$ that $\tr_i$ is either committed
or not committed $\trace$. Then:

\begin{lemma} [Non-local Read Last-use Consistency]
\label{lemma:non-local-read-lu-consistency}
For any $\op_i \in \trace|\tr_i$ s.t. $\op_i = \frop{i}{\op_i}{\val}$ and
$op_i$ is non-local, then either
\begin{enumerate} [i)]
\item $\val \neq 0$ and $\exists \op_j \in \luvis{\seqH}{\s\tr_i}$ for some
$\tr_j$ s.t. $\op_j = \fwop{j}{\obj}{\val}{\ok_j}$, $\op_j
\pref_{\luvis{\seqH}{\s\tr_i}} \op_i$, or 

\item $\val = 0$ and $\nexists \op_j \in \luvis{\seqH}{\s\tr_i}$ s.t. $\op_j =
\fwop{j}{\obj}{\val}{\ok_j}$ and $\op_j \prec_{\luvis{\seqH}{\s\tr_i}}
\op_i$.
\end{enumerate}
\end{lemma}

\begin{proof}
Since $\op_i$ is consonant and non-local, then
$\exists e_v = \get{i}{\obj}{\val} \in \trace$, s.t. $\op_i \dependson e_v$.
Then, since $\trace$ is consonant, then $e_v$ is also consonant, so one of the
following is true:
\begin{enumerate}[a)]
\item $\val = 0$ and $\nexists e_u = \set{j}{\obj}{\any} \in \trace$ for some
$\tr_j \in \trace$ s.t. $e_u \prec_\trace e_v$. \label{case:lu:first} 

In which case, if $\nexists \op_j = \fwop{j}{\obj}{\any}{\ok_i} \in
\trace|\tr_j$ s.t. $\op_j \prec_\trace \get{i}{\obj}{\val}$, then, 
$\nexists \tr_j \in \trace$ s.t. $\tr_j \isoorder_\trace \tr_i$ and $\op_j \in
\trace|\tr_j$.
Thus, from construction of $\seqH$, $\nexists \s\tr_j \in \trace$ s.t. $\s\tr_j
\prec_{\seqH} \s\tr_i$ and $\op_j \in \seqH|\s\tr_j$.
Thus, from construction of $\luvis{\seqH}{\s\tr_i}$, for any such $\tr_j$,
$\seqH|\s\tr_j \nsubseteq \luvis{\seqH}{\s\tr_i}$, so for any such $\tr_j$,
\underline{$\fwop{j}{\obj}{\val}{\ok_j} \not\in \luvis{\seqH}{\s\tr_i}$ and $\val=0$}.

On the other hand, if $\exists \op_j \fwop{j}{\obj}{\any}{\ok_i} \in
\trace|\tr_j$ s.t. $\op_j \prec_\trace \get{i}{\obj}{\val}$, then 
if $\tr_j$ is committed in $\trace$, then, from the definition of commit write obbligato,
$\exists \rset{j}{\obj}{\any} \in \trace|\tr_j$, which contradicts the assumption of
case \ref{case:lu:first}).
If, however, $\tr_j$ is not committed in $\trace$, then either $\tr_j$ is
decided on $\obj$ in $\trace$, or it is not. 
In the former of those two cases, from the definition of closing write obbligato,
$\exists \rset{j}{\obj}{\any} \in \trace|\tr_j$, which also contradicts the
assumption of case \ref{case:lu:first}).
In the latter case, since $\s\tr_j$ is neither committed in $\trace$ nor decided
on $\obj$ in $\trace$, then neither is it committed in $\seqH$ nor decided on
$\obj$ in $\seqH$. Therefore, by definition of $\luvis{\seqH}{\s\tr_i}$,
$\seqH|\s\tr_j|\obj \nsubseteq \luvis{\seqH}{\s\tr_i}$. Thus for any such
$\tr_j$, \underline{$\fwop{j}{\obj}{\any}{\ok_j} \not\in
\luvis{\seqH}{\s\tr_i}$ and $\val=0$}.
\item $\val \neq 0$ and $\exists e_u = \rset{j}{\obj}{\val} \in \trace$ for
some $\tr_j \in \trace$ s.t. $e_u \pref_\trace e_v$.
\label{case:lu:ordinary}

Since $e_u$ is consonant, then $\exists \op_j =
\fwop{j}{\obj}{\val}{\ok_j} \in \trace|\tr_j$ s.t. $\op_j$ is non-local and
consonant, and $\op_j \pref \rset{j}{\obj}{\val}$.
Thus, since $e_u \pref_\trace e_v$, $\tr_j \iso{\obj}{\trace} \tr_i$, then, by
construction, $\s\tr_j \prec_{\seqH} \s\tr_i$.

If $\tr_i$ is not committed in $\trace$, since $\tr_i \views \tr_j$, and since
$\trace$ is decisive, then $\tr_j$ is decided on $\obj$ in $\trace$.
Thus, from \rdef{lvis-construction}, $\seqH|\s\tr_j|\obj \subseteq
\luvis{\seqH}{\tr_i}$, and therefore $\op_j \prec_{\luvis{\seqH}{\tr_i}}
\op_i$.
Alternatively, if $\tr_i$ is committed in $\trace$, then, from commit accord,
$\tr_j$ is also committed in $\trace$. Then, by definition of
$\luvis{\seqH}{\tr_i}$, $\seqH|\s\tr_j \subseteq \luvis{\seqH}{\tr_i}$, and
thus $\op_j \prec_{\luvis{\seqH}{\tr_i}} \op_i$.
Then, from \rlemma{lemma:direct-precedence-in-luvis}, $\op_j
\pref_{\luvis{\seqH}{\tr_i}} \op_i$. Thus,
\underline{$\fwop{j}{\obj}{\val}{\ok_j} \pref_{\luvis{\seqH}{\s\tr_i}} \op_i$ and
$\val\neq 0$}.

\item $\exists e_a^j = \aset{j}{\obj}{\val} \in \trace|\tr_j$ s.t. $e_a^j
    \pref_\trace e_v$. \label{case:lu:rever}

Since $\trace$ is consonant, then $e_a^j$ is consonant, so $e_a^j$ is conservative,
and thus $\exists e_v^j = \get{j}{\obj}{\val} \in \trace|\tr_j$.

From \rcor{cor:abset:order} $\forall \tr_n \in \absetx{i}$
    ($n\neq i$) $\tr_n \iso{\obj}{\trace} \tr_i$.
So, by construction, for every such $\tr_n$, $\s\tr_n$ is aborted in
$\seqH$ and therefore not included as a whole in $\luvis{\seqH}{\tr_i}$.
In addition, since $e_a^n$ is needed there is a preceding routine update
$e_u^n = \rset{n}{\obj}{\val'}$. Because of unique writes, $\val' \neq \val$. 
Therefore, it is not true that $\tr_i \vviews \tr_n$.
Furthermore, since $e_u^n \prec_\trace e_a^n$ and $e_a^n \prec_\trace e_a^j$
and $e_a^j \prec_\trace e_v$, then,  $e_u^n \prec_\trace e_a^n \prec_\trace
e_v$, so, from chain isolation, $\nexists \vchain{\trace}{n}{m}$.
Thus, from \rdef{lvis-construction} $\seqH\cpeC\tr_n$ is not included in
$\luvis{\seqH}{\tr_i}$.
Thus for any $\tr_n \in \absetx{i}$ ($n\neq i$), $\forall \op_n =
\fwop{n}{\obj}{\val} \in \trace|\tr_n$, $\op_n \not\in \luvis{\seqH}{\tr_i}$.

Given $\absetx{i}$, from \rlemma{lemma:abset:views}, $\exists e_v^k =
\get{k}{\obj}{\val} \in \trace|\tr_k$ s.t. $\tr_k$ is the first element of
$\absetx{i}$ that is initial and non-local, and either of the following is true:
\begin{enumerate}[i)]
    \item $\val = 0$ and $\nexists \tr_l\in\trace$ s.t. $e_u^l =
    \set{l}{\obj}{\any} \in \trace|\tr_l$ and $e_u^l \prec_\trace e_v^k$.    

    Then, either $\exists e_u^l = \set{l}{\obj}{\val} \in \trace|\tr_l$ and
    $e_v^k \prec_\trace e_u^l$ or $\nexists e_u^l = \set{l}{\obj}{\any} \in
    \trace|\tr_l$.

    If $\exists e_u^l = \set{l}{\obj}{\any} \in \trace|\tr_l$ and $e_v^k
    \prec_\trace e_u^l$, then from \rlemma{lemma:abset:views}, $e_u^l =
    \rset{l}{\obj}{\any} \in \trace|\tr_l$. Thus, by definition of isolation
    order, $\tr_k \iso{\obj}{\trace} \tr_l$.

    Thus, if $\tr_l \iso{\obj}{\trace} \tr_i$, then, from
    \rlemma{lemma:abset:containment}, $\tr_l$ is aborted or live in $\trace$,
    so, by construction, $\s\tr_l$ is aborted in $\seqH$.
    Since $\s\tr_l$ is not committed in $\seqH$, then $\s\tr_l|\seqH$ is not
    included as a whole in $\luvis{\seqH}{\tr_i}$. Furthermore,
    $\seqH\cpeC\tr_l$ can be omitted from $\luvis{\seqH}{\tr_i}$.
    From unique routine updates, there cannot be $\tr_n \in
    \vchain{\trace}{l}{i}$ s.t. $\exists \rset{n}{\obj}{0}$, so since $\val =
    0$ and from self-containment, $\nexists \vchain{\trace}{l}{i}$.
    Thus, from \rdef{lvis-construction}, $\seqH\cpeC\tr_k
    \nsubseteq \luvis{\seqH}{\tr_i}$. 
    Then, for any write operation execution $\op_l =
    \fwop{l}{\obj}{\val'}{\ok_l}$ in any such $\tr_l$, \underline{$\op_l
    \not\in \luvis{\seqH}{\tr_i}$ (and $\val=0$)}.

    Alternatively, if  $\tr_i \iso{\obj}{\trace} \tr_l$, then since
    $e_a\pref_\trace e_v$, then $e_v \prec_\trace e_u^l$, so $\tr_i
    \iso{\obj}{\trace} \tr_l$, and thus $\tr_i \prec_{\seqH} \tr_l$, which means that
    $\seqH\cpeC\tr_k \nsubseteq \luvis{\seqH}{\tr_i}$.
    In either case for any write operation execution $\op_l =
    \fwop{l}{\obj}{\val'}{\ok_l}$ in any such $\tr_l$, \underline{$\op_l
    \not\in \vis{\seqH}{\tr_i}$ (and $\val=0$)}.

    On the other hand, if $\nexists e_u^l = \set{l}{\obj}{\any} \in
    \trace|\tr_l$, either $\tr_l \isoorder_\trace \tr_i$, or 
    $\tr_l \nisoorder_\trace \tr_i$.
    In the latter case, if $\tr_i \isoorder_\trace \tr_l$, then, trivially, no
    subset of $\seqH|\s\tr_l$ is contained in $\luvis{\seqH}{\tr_l}$.
    If $\tr_i \nisoorder_\trace \tr_l$, then there does not exist
    $\vchain{\trace}{l}{i}$, so, from \rdef{lvis-construction}, no subset of
    $\seqH|\s\tr_l$ is contained in $\luvis{\seqH}{\tr_l}$.
    If $\tr_l \isoorder_\trace \tr_i$, then either $\trace|\tr_l$ contains some
    write operation $\op_l = \fwop{l}{\obj}{\val'}{\ok_l}$ or it does not.
    However, from view write obbligato, since $\tr_l \isoorder_\trace \tr_i$,
    there must be $\set{l}{\obj}{\any} \in \trace|\tr_l$, so, there is no such
    $\op_l$.
    Then vacuously, for any write operation execution $\op_l
    = \fwop{l}{\obj}{\any}{\ok_l}$ in any such $\tr_l$, \underline{$\op_l
    \not\in \luvis{\seqH}{\tr_i}$ (and $\val=0$)}.

    \item $\val \neq 0$ and $\exists \tr_l\in\trace$ s.t. $e_u^l =
    \rset{l}{\obj}{\val} \in \trace|\tr_l$ and $e_u' \pref_\trace e_v'$.

    Since $\trace$ is consonant, then $e_u^l$ is consonant, so $\exists \op_l =
    \fwop{l}{\obj}{\val}{\ok_l}$ s.t $\op_l \instigates e_u^l$.

    Let us first assume that $\tr_l$ is committed in $\trace$.
    Then since $\tr_l \iso{\obj}{\trace} \tr_i$ (so $\s\tr_l \prec_{\seqH}
    \s\tr_i$), $\op_l \in \luvis{\seqH}{\tr_i}$.
  
    If, on the other hand, $\tr_l$ is not committed in $\trace$,     
    then, since $e_v^k$ is consonant, then $e_u^l$ is the ultimate routine
    update event in $\trace|\tr_l$.
    Therefore, from decisiveness, $e_u^l$ is either the closing routine update
    event on $\obj$ in $\trace|\tr_l$, or $e_u^l \prec_\trace \res{l}{}{\co_l}
    \prec_\trace e_v$. Since $\tr_l$ is not committed in $\seqH$, $e_u^l$ is
    the closing routine update event, so $\op_l$ is the closing write on $\obj$
    in $\tr_l$.
    Because of this, and since $\s\tr_l \prec_{\seqH} \s\tr_i$, $\seqH\cpeC\tr_l$ can
    be included in $\luvis{\seqH}{\tr_i}$. Then, since $\tr_i \vviews \tr_l$,
    then there exists $\vchain{\trace}{l}{i}$, so according to 
    \rdef{lvis-construction}, $\seqH\cpeC\s\tr_l$ is included in $\luvis{\seqH}{\tr_i}$, and
    therefore \underline{$\op_l \in \luvis{\seqH}{\tr_i}$}.

    From minimalism, $e_u^l$ is not followed in $\trace|\tr_l$ by
    another $\rset{l}{\obj}{\any}$.
    Since $\op_l \instigates e_u^l$, $\nexists \op_l' =
    \fwop{l}{\obj}{\any}{\ok_l}$ s.t. $\op_l \prec_{\trace|\tr_i} \op_l'$.

    Let $\tr_m$ be any transaction in $\trace$ s.t. $\tr_l \isoorder{\trace}
    \tr_m \isoorder{\trace} \tr_i$. If $\nexists \fwop{m}{\obj}{\any}{\ok_m}
    \in \trace|\tr_m$, then trivially, for any such $\tr_m$  $\nexists
    \fwop{m}{\obj}{\val}{\ok_m} \in \luvis{\seqH}{\tr_i}$. Thus, let there be
    $\op_m = \fwop{m}{\obj}{\val'}{\ok_m} \in \trace|\tr_m$.
    
    If $\tr_m \in \absetx{\trace}{i}$, then, as shown above, $\seqH|\s\tr_m
    \nsubseteq \luvis{\seqH}{\tr_i}$ and therefore $\op_m \not\in
    \luvis{\seqH}{\tr_i}$. Hence, let $\tr_m \not\in \absetx{\trace}{i}$.

    Either $tr_m \isoorder_{\trace} \tr_i$ or $tr_m \nisoorder_{\trace} \tr_i$.
    In the latter case, if $\tr_m$ is not committed, $\tr_m$ can be excluded
    from $\luvis{\seqH}{\tr_i}$. Since in that case there doe not exist
    $\vchain{\trace}{m}{i}$, then, by \rdef{lvis-construction}, $\seqH|\s\tr_m
    \nsubseteq \luvis{\seqH}{\tr_i}$ and therefore $\op_m \not\in
    \luvis{\seqH}{\tr_i}$. 
    If $\tr_m$ is committed, then if $\fwop{m}{\obj}{\val'}{\ok_m} \in
    \trace|\tr_m$, then by commit write obbligato, there would have to exist
    $e_u^m = \rset{m}{\obj}{\val'} \in \trace|\tr_m$, which would imply that
    $\tr_m \iso{\trace}{\obj} \tr_i$, which contradicts the assumption that
    $tr_m \nisoorder_{\trace} \tr_i$. Therefore, there is no such transaction.

    If $\tr_m \isoorder_{\trace} \tr_i$, then, if $\tr_m \isoorder_{\trace}
    \tr_l$, then $\s\tr_m \prec_{\seqH} \s\tr_l$, and therefore $\op_m
    \prec_{\luvis{\seqH}{\tr_i}} \op_l$, which has no bearing on whether
    $\op_i$ is preceded by a corresponding write operation.
    Hence, let $\tr_l \isoorder_{\trace} \tr_m \isoorder_{\trace} \tr_i$.
    Then, $\tr_m$ is either committed in $\trace$ or not.
    
    If $\tr_m$ is committed, then from commit write obbligato $\exists e_u^m =
    \rset{m}{\obj}{\val'} \in \trace|\tr_m$. From isolation, it is impossible
    that $\tr_i \iso{\trace}{\obj} \tr_m$ or $\tr_m \iso{\trace}{\obj} \tr_l$,
    then $\tr_l \iso{\trace}{\obj} \tr_m \iso{\trace}{\obj} \tr_i$.
    
    Then, if $\tr_j \iso{\trace}{\obj} \tr_m \iso{\trace}{\obj} \tr_i$,    
    since $e_a$ is consonant, it is is needed, so
    $\exists e_u^j = \set{j}{\obj}{\any} \in \trace|\tr_j$ s.t. $e_u^j
    \prec_\trace e_a$. 
    In addition, since $\tr_j \iso{\obj}{\trace} \tr_m \iso{\obj}{\trace}
    \tr_i$, then  by definition of isolation, $\exists e =
    \set{m}{\obj}{\any}\in \trace|\tr_m$, or $e \get{m}{\obj}{\any}\in
    \trace|\tr_m$, so $e_u^j \prec_\trace e \prec_\trace e_v$. 
    Since $e_a \pref_\trace e_v$, then $e_u^j \prec_\trace e
    \prec_\trace e_u$. Thus, by definition of abort accord, $\tr_m$ 
    cannot be in $\trace$, thus there is no such $\tr_m$.
          
    If $\tr_l \iso{\obj}{\trace} \tr_m \iso{\obj}{\trace} \tr_j$, then, since
    $\tr_m \not\in \absetx{\trace}{i}$, then, from
    \rlemma{lemma:abset:containment}, $\tr_m$ cannot be committed in $\trace$,
    thus, there is also no such $\tr_m$.

    Since there is no such $\tr_m$ that is both committed and contains an
    operation execution such as $\op_m$, then for any such $\tr_m$,
    $\seqH|\s\tr_m \nsubseteq \luvis{\seqH}{\tr_i}$ and therefore $\op_m
    \not\in \luvis{\seqH}{\tr_i}$.

    On the other hand, if $\tr_m$ is not committed in $\trace$, then, since
    $\op_m$ is consonant, either $\op_m \instigates e_u^m$ where $e_u^m =
    \rset{m}{\obj}{\val'} \in \trace|\tr_m$, or $\nexists \rset{m}{\obj}{\any}
    \in \trace|\tr_m$. 
    Since, from view write obbligato the latter case is
    impossible, then  $\exists e_u^m = \rset{m}{\obj}{\val'} \in \trace|\tr_m$.
    Then, from the definition of isolation order, $\tr_l \iso{\obj}{\trace}
    \tr_m \iso{\obj}{\trace} \tr_i$.
    Since $\tr_m$ is not committed, then $\s\tr_m$ can be omitted in
    $\luvis{\seqH}{\tr_i}$.
    Due to unique routine updates, there cannot be any $\tr_n \in \trace$ s.t.
    $\rset{n}{\obj}{\val''}$ where $\val'' = \val$. Therefore, given any
    $\vchain{\trace}{m}{i}$ and there is no transaction to satisfy
    self-containment. Thus, there is no such  $\vchain{\trace}{m}{i}$. Thus, by
    \rdef{lvis-construction}, $\seqH|\s\tr_m \nsubseteq \luvis{\seqH}{\tr_i}$ and
    therefore $\op_m \not\in \luvis{\seqH}{\tr_i}$.

    Because there is no $\tr_m$ s.t.  $\seqH|\s\tr_m \subseteq
    \luvis{\seqH}{\tr_i}$ and $\op_m \in \luvis{\seqH}{\tr_i}$, then there is
    no write operation execution $\op'$ on $\obj$ s.t. $\op_l
    \prec_{\luvis{\seqH}{\tr_i}} \op' \prec_{\luvis{\seqH}{\tr_i}} \op_i$.
    Therefore \underline{$\fwop{l}{\obj}{\val}{\ok_l}
    \pref_{\vis{\seqH}{\s\tr_i}} \op_i$ and $\val\neq 0$}.

\end{enumerate}
\end{enumerate}
\end{proof}

\begin{lemma} [All Non-local Read Last-use Consistency] 
    \label{lemma:all-non-local-read-lu-consistency}
    If for some sequential
    history $S$, $\seqH|tr_j \in \vis{S}{\tr_i}$, for any $\op_j \in \trace|\tr_j$ s.t.
    $\op_j = \frop{i}{\op_j}{\val}$ either: 
    \begin{enumerate} [1)] 
        \item $\val \neq 0$ and $\exists \op_k \in \vis{\seqH}{\s\tr_k}$ for
            some $\tr_k$ s.t.  $\op_k = \fwop{k}{\obj}{\val}{\ok_j}$, $\op_k
            \pref_{\vis{\seqH}{\s\tr_i}} \op_j$, or 
        \item $\val = 0$ and $\nexists \op_k \in \vis{\seqH}{\s\tr_i}$ s.t.
            $\op_k = \fwop{k}{\obj}{\val'}{\ok_k}$ and $\op_k
            \prec_{\vis{\seqH}{\s\tr_i}} \op_j$.
    \end{enumerate}
\end{lemma}

\begin{proof}
    Either $\val \neq 0$ or $\val = 0$.
    \begin{enumerate}[1)]
        \item
        If $\val \neq 0$,from \rlemma{lemma:non-local-read-lu-consistency},
        since $\val \neq 0$ then $\exists \op_k = \fwop{k}{\obj}{\val}{\ok_k}
        \in \trace|\tr_k$ s.t.  $\op_k \pref_{\luvis{\seqH}{\tr_j}} \op_j$.
        
        From \rlemma{lemma:corr}, $\forall \tr_l$ if $\s\tr_l \sin
        \luvis{\seqH}{\tr_j}$
        then $\s\tr_l \sin \luvis{\seqH}{\tr_i}$ and
        $\luvis{\seqH}{\tr_i}|\tr_l = \luvis{\seqH}{\tr_j}|\tr_l$.
        Hence, $\tr_k \sin \luvis{\seqH}{\tr_i}$ and $\op_k \in \luvis{\seqH}{\tr_i}$. 
        Furthermore, if $\tr_k \sin \luvis{\seqH}{\tr_j}$, then $\tr_k
        \prec_{\seqH} \tr_j$, so  $\op_k \prec_{\luvis{\seqH}{\tr_i}} \op_j$.
         
        For the sake of contradiction, let us assume there exists $\op_l =
        \fwop{l}{\obj}{\val}{\ok_l} \in \trace|\tr_l$ s.t. $\op_k
        \prec_{\luvis{\seqH}{\tr_i}} \op_l \prec_{\luvis{\seqH}{\tr_i}} \op_j$.
        Since from \rlemma{lemma:non-local-read-lu-consistency}, there is no
        such transaction in $\luvis{\seqH}{\tr_i}$, then $\tr_l$ is such that
        $\s\tr_l \sin \luvis{\seqH}{\tr_i}$ s.t. and $\s\tr_l \nsin
        \luvis{\seqH}{\tr_j}$ (and $\s\tr_l \prec_{\seqH} \s\tr_j$).
        
        If $\s\tr_l \sin \luvis{\seqH}{\tr_i}$, then, by definition, $\tr_l$ is
        not committed in $\trace$.

        If $\tr_l \isoorder{\obj}{\trace} \tr_j$, then this
        is a contradiction by analogy to \rlemma{lemma:non-local-read-lu-consistency}.

        If $\tr_j \isoorder{\obj}{\trace} \tr_l$, then, from \rdef{seqh-construction},
        $\s\tr_j \prec_{\seqH} \s\tr_l$, which implies that
        $\op_j \prec_{\luvis{\seqH}{\tr_j}} \op_l$, which is a contradiction.

        If $\tr_l \nisoorder{\obj}{\trace} \tr_j$, then $\nexists e_u^k =
        \rset{l}{\obj}{\any} \in \trace|\tr_l$. Hence, from \rdef{seqh-construction}
        , $\s\tr_j \prec_{\seqH} \s\tr_l$, which imp`lies that $\op_j
        \prec_{\luvis{\seqH}{\tr_j}} \op_l$, which, again, is a contradiction.

        Thus, there is no such $\tr_l$, and, therefore, $\op_k
        \pref_{\luvis{\seqH}{\tr_i}} \op_j$ (and $\val \neq 0$).

        \item  
        If $\val = 0$, let us assume by contradiction, that there exists such
        $\tr_k$ and $\op_k$.
        From \rlemma{lemma:non-local-read-lu-consistency}, since $\val = 0$
        then $\nexists \op_l = \fwop{l}{\obj}{\any}{\ok_l} \in \trace|\tr_l$ s.t.
        $\op_l \prec_{\luvis{\seqH}{\tr_j}} \op_j$.
        Hence, $\tr_k$ must be such that $\s\tr_k \sin \luvis{\seqH}{\tr_i}$
        s.t. and $\s\tr_k 
        \nsin 
        \luvis{\seqH}{\tr_j}$ (and $\tr_k \prec_{\seqH}
        \tr_j$).
        
        From \rlemma{lemma:not-corr}, 
        $\forall \tr_l$ if $\s\tr_l \nsin \luvis{\seqH}{\tr_j}$
        and $\exists \vchain{\trace}{l}{j}$
        then $\s\tr_l \sin \luvis{\seqH}{\tr_i}$. 
        Thus, if $\tr_k$ is such that $\s\tr_k \sin \luvis{\seqH}{\tr_j}$ and
        $\s\tr_k \sin \luvis{\seqH}{\tr_i}$, then $\nexists
        \vchain{\trace}{k}{j}$.

        If $\s\tr_k \sin \luvis{\seqH}{\tr_i}$, then, by definition, $\tr_k$ is
        not committed in $\trace$.

        If $\tr_k \isoorder{\obj}{\trace} \tr_j$, then this
        is a contradiction by analogy to \rlemma{lemma:non-local-read-lu-consistency}.

        If $\tr_j \isoorder{\obj}{\trace} \tr_k$, then, from \rdef{seqh-construction},
        $\s\tr_j \prec_{\seqH} \s\tr_k$, which implies that
        $\op_j \prec_{\luvis{\seqH}{\tr_j}} \op_k$, which is a contradiction.

        If $\tr_k \nisoorder{\obj}{\trace} \tr_j$, then $\nexists e_u^k =
        \rset{k}{\obj}{\any} \in \trace|\tr_k$. Hence, from \rdef{seqh-construction},
        $\s\tr_j \prec_{\seqH} \s\tr_k$, which implies that $\op_j
        \prec_{\luvis{\seqH}{\tr_j}} \op_k$, which, again, is a contradiction.

        Thus, there is no such $\tr_k$, and, therefore,  $\nexists
        \op_k \in \luvis{\seqH}{\s\tr_i}$ s.t.  $\op_k =
        \fwop{k}{\obj}{\val'}{\ok_k}$ and $\op_k \prec_{\luvis{\seqH}{\s\tr_i}}
        \op_j$ (and $\val=0$).
    \end{enumerate}
\end{proof}

\begin{lemma} [Local Read Last-use Consistency] 
    \label{cor:local-read-lu-consistency} 
    For any $\op_i \in \trace|\tr_i$ s.t. $\op_i = \frop{i}{\op_i}{\val}$ and
    $op_i$ is local, then $\exists \op_i' \in \luvis{\seqH}{\tr_i}$ s.t.
    $\op_i' = \fwop{i}{\obj}{\val}{\ok_i}$, $\op_i'
    \pref_{\luvis{\seqH}{\tr_i}} \op_i$. 
\end{lemma}

\begin{proof}
    From local read consonance it follows that for any $\op_i \in \trace|\tr_i$
    s.t. $op_i$ is local, then $\exists \op_i' \in \trace|\tr_i$ and $\op_i'
    \pref_{\trace|\tr_i} \op_i$.  Thus, $\op_i' \pref_{\seqH|\s\tr_i} \op_i$.
    Since $\op_i$ and $\op_j$ operate on the same variable, then trivially,
    $\op_i \in\luvis{\seqH}{\tr_i} \iff \op_i' \luvis{\seqH}{\tr_i}$.  Hence,
    $\op_i' \pref_{\vis{\seqH}{\tr_i}} \op_i$.
\end{proof}

\begin{corollary} [Total Local Read Last-use Consistency] 
    \label{cor:all-local-read-lu-consistency} 
    For any $\op_i \in \trace|\tr_i$ s.t. $\op_i = \frop{i}{\op_i}{\val}$ and
    $op_i$ is local, then for any $\tr_j$, $\exists \op_i' \in
    \luvis{\seqH}{\tr_j}$ s.t.  $\op_i' = \fwop{i}{\obj}{\val}{\ok_i}$, $\op_i'
    \pref_{\luvis{\seqH}{\tr_j}} \op_i$. 
\end{corollary}

\begin{lemma} [Total Write Consistency] \label{lemma:write-consistency}
    For any $\tr_i$, $\tr_j$, $\forall \op_i \in \luvis{\seqH}{\tr_j}$ s.t.
    $\op_i = \fwop{i}{\op_i}{\val}{\ok_i} \in \trace|\tr_i$, $\val$ is in the
    domain of $\obj$.
\end{lemma}

\begin{proof}
    Follows from write consonance.
\end{proof}

\subsection{Proof for \rthm{thm:harmony-to-lu-opacity}}

\begin{proof} [Proof for \rthm{thm:harmony-to-lu-opacity}: Trace Last-use Opacity]
    For every transaction $\tr_i \in \trace$, given $\seqH$ constructed by
    \rdef{seqh-construction},
    \begin{enumerate}[i)]
        \item If $\tr_i$ is committed in $\hist$, from
            \rcor{cor:all-non-local-read-consistency}, $\forall \op_j \in
            \hist|\tr_j$ s.t.  $\op_j = \frop{i}{\op_j}{\val}$ and $\op_j$ is non-local,
            either: 
            \begin{enumerate} [a)] 
                \item $\val \neq 0$ and $\exists \op_k \in
                    \vis{\seqH}{\s\tr_k}$ for some $\tr_k$ s.t.  $\op_k =
                    \fwop{k}{\obj}{\val}{\ok_j}$, $\op_k
                    \pref_{\vis{\seqH}{\s\tr_i}} \op_j$, or 
                \item $\val = 0$ and $\nexists \op_k \in \vis{\seqH}{\s\tr_i}$
                    s.t.  $\op_k = \fwop{k}{\obj}{\val}{\ok_k}$ and $\op_k
                    \prec_{\vis{\seqH}{\s\tr_i}} \op_j$.
            \end{enumerate}

            In addition, from \rcor{cor:all-local-read-consistency}, $\forall \op_j \in
            \hist|\tr_j$ s.t.  $\op_j = \frop{i}{\op_j}{\val}$ and $\op_j$ is local,
            $\exists \op_j \in \vis{\seqH}{\s\tr_j}$ s.t.  $\op_j' =
            \fwop{j}{\obj}{\val}{\ok_j}$, $\op_j' \pref_{\vis{\seqH}{\s\tr_i}} \op_j$

            Furthermore, from \rlemma{lemma:write-consistency}, $\forall \op_j'
            \in \hist|\tr_j$ s.t.  $\op_j = \fwop{i}{\op_j}{\val}{\ok_j}$,
            $\val$ is in the domain of $\obj$.

            Thus, $\vis{\seqH}{\tr_i}$ is legal, and therefore $\tr_i$ is legal.

         \item If $\tr_i$ is not committed in $\hist$, from
            \rlemma{lemma:all-non-local-read-lu-consistency},
            $\forall \op_j \in \hist|\tr_j$ s.t.  $\op_j =
            \frop{i}{\op_j}{\val}$ and $\op_j$ is non-local, either: 
            \begin{enumerate} [a)] 
                \item $\val \neq 0$ and $\exists \op_k \in
                    \luvis{\seqH}{\s\tr_k}$ for some $\tr_k$ s.t.  $\op_k =
                    \fwop{k}{\obj}{\val}{\ok_j}$, $\op_k
                    \pref_{\luvis{\seqH}{\s\tr_i}} \op_j$, or 
                \item $\val = 0$ and $\nexists \op_k \in \luvis{\seqH}{\s\tr_i}$
                    s.t.  $\op_k = \fwop{k}{\obj}{\val}{\ok_k}$ and $\op_k
                    \prec_{\luvis{\seqH}{\s\tr_i}} \op_j$.
            \end{enumerate}

            In addition, from \rcor{cor:all-local-read-lu-consistency}, $\forall \op_j \in
            \hist|\tr_j$ s.t.  $\op_j = \frop{i}{\op_j}{\val}$ and $\op_j$ is local,
            $\exists \op_j \in \luvis{\seqH}{\s\tr_j}$ s.t.  $\op_j' =
            \fwop{j}{\obj}{\val}{\ok_j}$, $\op_j' \pref_{\luvis{\seqH}{\s\tr_i}} \op_j$

            Furthermore, from \rlemma{lemma:write-consistency}, $\forall \op_j'
            \in \hist|\tr_j$ s.t.  $\op_j = \fwop{i}{\op_j}{\val}{\ok_j}$,
            $\val$ is in the domain of $\obj$.

            Thus, $\luvis{\seqH}{\tr_i}$ is legal, and therefore $\tr_i$ is
            last-use legal.
    \end{enumerate}

    Since every committed transaction $\tr_i \in \hist$ is legal if it is
    committed and last-use legal if it is not committed, then $\hist$ is
    final-state last-use opaque.

    Since a prefix of a harmonious $\trace$ is trivially also harmonious,
    then for every prefix $\trace'$ of $\trace$, $\hist' = \tohist\trace'$ is
    also final-state last-use opaque. Thus, $\hist$ is last-use opaque.
\end{proof}

\end{document}